%% file: 1planar-v2.tex
\documentclass[11pt,prx,aps,amsfonts,showpacs,nofootinbib,longbibliography,superscriptaddress]{revtex4-1}
\input{macro.tex}

\begin{document}

\title{Excitation-detector principle and the algebraic theory of planon-only abelian fracton orders}

\author{Evan Wickenden}
\affiliation{Department of Physics, University of Colorado, Boulder, CO 80309, USA}
\affiliation{Center for Theory of Quantum Matter, University of Colorado, Boulder, CO 80309, USA}
\author{Wilbur Shirley}
\affiliation{Kadanoff Center for Theoretical Physics, University of Chicago, Chicago, IL 60637, USA}
\affiliation{School of Natural Sciences, Institute for Advanced Study, Princeton, NJ}
\author{Agn\`es Beaudry}
\affiliation{Department of Mathematics, University of Colorado, Boulder, CO 80309, USA}
\author{Michael Hermele}
\affiliation{Department of Physics, University of Colorado, Boulder, CO 80309, USA}
\affiliation{Center for Theory of Quantum Matter, University of Colorado, Boulder, CO 80309, USA}
\date{\today}

\begin{abstract}
	\input{abstract}
\end{abstract}

\maketitle

\tableofcontents

\input{intro.tex}

\input{bulk.tex}

\input{decoupled.tex}

\input{discussion.tex}

\appendix
 
\input{localization.tex}

\input{ktheory.tex}

\bibliography{bibliography}

\end{document}

%% file: macro.tex
\usepackage{graphicx} % Required for inserting images
\usepackage{amsmath}
\usepackage{amssymb}
\usepackage{amsthm}
\usepackage{mathrsfs}
\usepackage{extarrows} 
\usepackage{xcolor}
\usepackage[margin=0.8 in]{geometry}
\usepackage{tikz-cd}
\usepackage[normalem]{ulem}
\usepackage{scalerel}
\usepackage{bbm}
\usepackage{hyperref}
\usepackage{framed}
\usepackage{cleveref}
\usepackage{enumerate}
\usepackage[all]{xy}
\usepackage{todonotes}

\let\bigopsize\bigoplus
\def\bigominus{{\scalerel*{\boldsymbol\ominus}{\bigopsize}}}

\usepackage{array}   % for \newcolumntype macro
\newcolumntype{L}{>{$}r<{$}} % math-mode version of "l" column type

\newcolumntype{C}{>{$}c<{$}} % math-mode version of "l" column type

\newcommand{\Z}{\mathbb Z}

\newcommand{\T}{\mathcal T}

\newcommand{\R}{\mathbb R}
\newcommand{\Q}{\mathbb Q}
\newcommand{\N}{\mathbb N}
\newcommand{\Qlp}{\mathcal Q}
\newcommand{\Ztpm}{\mathcal Z}
\newcommand{\tZ}{\tilde{\Ztpm}}
\newcommand{\bR}{\mathbf R}

\newcommand{\tQlp}{\tilde{\Qlp}}

\newcommand{\im}{\operatorname{im}}
\newcommand{\Hom}{\operatorname{Hom}}

\newcommand{\HomZ}{\Hom_{\Z}}
\newcommand{\HomlocZ}{\Hom^{\mathrm{loc}}_{\Z}}

\newcommand{\cA}{\mathcal A}
\newcommand{\cR}{\mathcal R}

\newtheorem{theorem}{Theorem}[section]
\newtheorem{proposition}[theorem]{Proposition}
\newtheorem{lemma}[theorem]{Lemma}
\newtheorem{corollary}[theorem]{Corollary}

\theoremstyle{definition}
\newtheorem{defn}[theorem]{Definition}
\newtheorem{remark}[theorem]{Remark}
\newtheorem{example}[theorem]{Example}

\newcommand{\m}{m}

\newcommand{\id}{\mathrm{id}}

	\makeatletter
\newsavebox{\@brx}
\newcommand{\llangle}[1][]{\savebox{\@brx}{\(\m@th{#1\langle}\)}%
  \mathopen{\copy\@brx\kern-0.5\wd\@brx\usebox{\@brx}}}
\newcommand{\rrangle}[1][]{\savebox{\@brx}{\(\m@th{#1\rangle}\)}%
  \mathclose{\copy\@brx\kern-0.5\wd\@brx\usebox{\@brx}}}

\makeatother

\newcommand{\Tor}{\mathrm{Tor}}

%% file: abstract.tex
%!TEX root = 1planar-v2.tex

We study abelian planon-only fracton orders: a class of three-dimensional (3d) gapped quantum phases in which all fractional excitations are abelian particles restricted to move in planes with a common normal direction. In such systems, the mathematical data encoding fusion and statistics comprises a finitely generated module over a Laurent polynomial ring $\mathbb{Z}[t^\pm]$ equipped with a quadratic form giving the topological spin. The principle of remote detectability requires that every planon braids nontrivially with another planon. While this is a necessary condition for physical realizability, we observe -- via a simple example -- that it is not sufficient. This leads us to propose the \emph{excitation-detector principle} as a general feature of gapped quantum matter. For planon-only fracton orders, the principle requires that every detector -- defined as a string of planons extending infinitely in the normal direction -- braids nontrivially with some finite excitation.  We prove this additional constraint is satisfied precisely by perfect theories of excitations -- those whose quadratic form induces a perfect Hermitian form. To justify the excitation-detector principle, we consider the 2d abelian anyon theory obtained by spatially compactifying a planon-only fracton order in a transverse direction.  We prove the compactified 2d theory is modular if and only if the original 3d theory is perfect, showing that the excitation-detector principle gives a necessary condition for physical realizability that we conjecture is also sufficient. A key ingredient is a structure theorem for finitely generated torsion-free modules over $\mathbb{Z}_{p^k} [t^\pm]$, where $p$ is prime and $k$ a natural number. Finally, as a first step towards classifying perfect theories of excitations, we prove that every theory of prime fusion order is equivalent to decoupled layers of 2d abelian anyon theories.

%% file: intro.tex
%!TEX root = 1planar-v2.tex

\section {Introduction}
\label{sec:intro}

Fracton orders in three spatial dimensions lie at the frontier of quantum matter \cite{Chamon_2005, Haah_2011, Pretko_2017,Nandkishore_2019, Pretko_2020_review, You_2025} .  Among other challenges, we lack a general theory of fracton order that efficiently encodes universal properties free of microscopic details.  A successful model is the algebraic theory of anyons \cite{Moore_1989,Kitaev_2006}, a mathematical structure describing the fusion and braiding of excitations in a two-dimensional (2d) topologically ordered system, which takes a particularly simple form in abelian anyon systems \cite{Belov_2005, Stirling_2008, Kapustin_2011, Galindo_2016, Lee_2018}. It is an important problem to develop a similar algebraic theory of abelian fracton orders that characterizes the universal properties of excitations above the ground state.

The defining characteristic of fracton orders is the presence of restricted-mobility point-like excitations.  These can be excitations restricted to move within a plane (planons), along a line (lineons), or that are completely immobile when isolated from other excitations (fractons).  In this paper, we propose an algebraic theory of fracton orders in the special case where all the nontrivial excitations are abelian planons of finite fusion order.\footnote{The fusion order $n$ of an excitation is the smallest natural number such that a composite of $n$ copies of the excitation is a trivial (\emph{i.e.} locally-createable) excitation.}  While naively this may seem like an oversimplification, planon-only fracton orders are interesting because they include examples that are not simply decoupled layers of 2d topological orders \cite{MacDonald_1989, MacDonald_1990, Sondhi_2000, Sondhi_2001,Shirley_2020, Ma_2022}.

Our algebraic theory is comprised of data capturing fusion, mobility and statistical properties of planon excitations.  While the data that should be included is clear, and some of the conditions are straightforward generalizations from the theory of abelian anyons, we identify a new condition that is necessary for physical realizability. This condition originates from our observation that the principle of remote detectability \cite{Kitaev_2006, Lan_2018}, which plays a fundamental role in the theory of topological orders, needs to be upgraded for fracton orders to a property that we dub the \emph{excitation-detector principle}.

While we focus on systems with only planon excitations, our results have important consequences for more general fracton orders.  First, we expect the excitation-detector principle to play an essential role in developing more general algebraic theories of gapped fracton matter.   Second, as an application of our algebraic theory, we obtain a result with consequences for entanglement renormalization group (RG) flows of p-modular (planon-modular) fracton orders, which are a simple class of fracton orders recently introduced in Ref.~\onlinecite{Wickenden_2024}. The class of p-modular fracton orders includes examples such as Chamon's model \cite{Chamon_2005}, the X-cube model \cite{Vijay_2016}, and the four color cube (FCC) model \cite{Ma_2017}, but does not encompass Haah's cubic code \cite{Haah_2011} and other type II fracton orders.  It was shown previously that if a p-modular fracton order admits an entanglement RG fixed point, the system that is integrated out in the course of an RG step must be a planon-only fracton order \cite{Wickenden_2024}.  If this system is a stack of decoupled 2d layers, then by definition we have a foliated RG in the sense of Ref.~\onlinecite{Shirley_2018}, and we say that the fracton order is foliated.  In this paper, we prove that if a planon-only fracton order has prime fusion order $p$ for all nontrivial excitations, then the fracton order is a stack of 2d layers.  Therefore, a p-modular fracton order where all excitations have prime $p$ fusion order is a foliated fracton order (if it admits an entanglement RG fixed point at all).  This result also tells us that the interesting examples of planon-only fracton order, \emph{i.e.} systems that are not simply decoupled 2d layers, must have excitations of composite fusion order.

To set the stage for a more detailed overview of our main results, we first review the algebraic theory of abelian anyons \cite{Belov_2005, Stirling_2008, Kapustin_2011, Galindo_2016, Lee_2018}, focusing throughout on bosonic systems; see Sec.~\ref{sec:theories} for formal definitions. For a given system one specifies two pieces of data, a finite abelian group $\cA$, and a quadratic form $\theta : \cA \to \Q / \Z$.\footnote{To fully characterize a 2d gapped phase, one also needs to include the chiral central charge.  We omit this piece of data because we are only concerned with the fusion and braiding properties of the anyon excitations.}  Elements of $\cA$ are superselection sectors (also referred to as topological charges or particle types), which are equivalence classes of point-like excitations under the action of local operators.  The group operation is fusion of excitations, and the identity element is the class of locally createable excitations. The quadratic form gives the topological spin of an excitation $x \in \cA$ by $e^{2\pi i \theta(x)}$.  There is also a symmetric bilinear form $b(x,y) = \theta(x+y) - \theta(x) - \theta(y)$ giving the mutual braiding statistics of two excitations $x,y \in \cA$ by $e^{2\pi i b(x,y)}$.

The data $(\cA, \theta)$ obeys a condition coming from the following:
\begin{framed}
\noindent
\textbf{Principle of remote detectability.} Any nontrivial (\emph{i.e.} non-locally-createable) excitation can be detected by acting with some operator whose support lies very far from the excitation 
 \cite{Kitaev_2006, Lan_2018}.
\end{framed}
\noindent While this principle was formulated in the literature for topologically ordered systems, it is believed to hold very generally in gapped quantum systems, including fracton orders.  The intuition is that if an excitation is not locally createable, then it must be created by a non-local operator extending outside the region where the excitation itself is supported.  The action of the non-local operator should be detectable by some other operator, supported far away from the excitation in question, that does not commute with the non-local operator.  In anyon systems, the operators we need for remote detection are anyon string operators, and remote detectability reduces to braiding nondegeneracy, the statement that every nontrivial excitation should be remotely detectable by braiding. That is, given a nonzero $x \in \cA$, there exists $y \in \cA$ such that $b(x,y) \neq 0$. 
The mathematical term for this property is nondegeneracy of $\theta$ (and $b$), which is also referred to as modularity in this context.  It is known that every modular theory of abelian anyons can be realized by a multi-component ${\rm U}(1)$ Chern-Simons theory \cite{Wall_1963, Wall_1972, Nikulin_1980, Wang_2020}.  Moreover, it is strongly believed that a theory of abelian anyons can be realized by a quantum system with a gapped, local Hamiltonian in two dimensions if and only if the theory is modular.

Now, and throughout this paper, we consider abelian fracton orders in three-dimensional bosonic systems with a gap to all excitations.  We further focus on the case where all nontrivial excitations are point-like.  This excludes hybrid fracton orders \cite{Tantivasadakarn_2021}, which have extended loop excitations, but includes most of the examples studied in the literature.  Abelian fracton orders are those where the set of superselection sectors, denoted $S$, is an abelian group under fusion.  Unlike in theories of abelian anyons, $S$ is infinite, and is typically not even finitely generated. In order to describe the mobility of excitations in a precise manner, we impose three-dimensional lattice translation symmetry, with symmetry group $\T \cong \Z^3$.  Translations act on $S$, and this makes $S$ into a module over $\Z \T$, the group ring of $\T$ with integer coefficients.  Translation symmetry plays a subtle role in the theory of fracton orders; see Sec.~\ref{subsec:definitions} for a brief discussion.  It is well-established that the $\Z\T$-module $S$ encodes the fusion and mobility properties of excitations \cite{Haah_2013, Pai_2019} .

Braiding and exchange have been generalized in various ways to fractons and other restricted-mobility excitations, which thus have properties analogous to self and mutual statistics of anyons \cite{Song_2019, Pai_2019, You_2020, Song_2024}.  Unlike for theories of abelian anyons, it is not yet understood in general how to encapsulate these statistical properties in an algebraic theory.  However, this is clear in the case of planon-only fracton orders, because self and mutual statistics of planons behave as for anyons in 2d.  Any planon is mobile under translations that span a two-dimensional subspace of $\R^3$; we call this subspace the planon's orientation.  In planon-only fracton orders, all the nontrivial excitations are planons of the same orientation, \emph{i.e.} they all move within parallel planes.\footnote{If all the nontrivial excitations of an abelian fracton order are planons, it follows from Proposition~A.10 of Ref.~\onlinecite{Wickenden_2024} that all the planons have the same orientation.}  The data of an algebraic theory of planon-only fracton orders should thus consist of the pair $(S, \theta)$, where $S$ is the module of superselection sectors, and $\theta : S \to \Q / \Z$ is a quadratic form giving the topological spin of each planon as above.  As for abelian anyons, the quadratic form determines a symmetric bilinear form $b : S \times S \to \Q / \Z$ giving the mutual statistics between planons.  We assume that all planons are of finite  order under fusion; that is, for each $x \in S$ there is some $n \in \N$ with $n x = 0$.  Without this assumption, $\theta$ could take values in $\R / \Z$; indeed, there are interesting planon-only fracton orders with planons of infinite fusion order and irrational statistics \cite{MacDonald_1989,MacDonald_1990,Sondhi_2000, Sondhi_2001, Ma_2022, Chen_2023}. We will leave consideration of these examples to future work.

Some of the conditions that the data $(S,\theta)$ should satisfy are also clear.  First, the module $S$ should be a finitely generated $\Z\T$-module, even though it is not a finitely generated abelian group; this amounts to the property that only a finite number of excitation types can fit into a finite spatial volume, which is certainly true in lattice spin models where the Hilbert space associated with each crystalline unit cell is finite-dimensional.  Second, spatial locality demands that the mutual statistics between two planons should approach zero in the limit of large transverse spatial separation.  Finally, the principle of remote detectability requires that $\theta$ is nondegenerate.  Just like for abelian anyons in 2d, planons should be detectable via braiding with other planons, and remote detectability is equivalent to braiding nondegeneracy.

At this point, we might expect that every theory $(S, \theta)$ satisfying the above conditions can be realized by some quantum system. This is a reasonable expectation, because the principle of remote detectability is enough to ensure physical realizability for theories of 2d abelian anyons.  Surprisingly, the expectation is incorrect.  In Section~\ref{sec:example}, we give an example of a theory $(S, \theta)$ satisfying the above properties, including nondegeneracy of $\theta$, which cannot be realized.  The reason is that, upon compactifying to 2d space, we obtain a non-modular theory of anyons.  This raises the following question:  what is the missing ingredient needed to ensure that a theory of planon-only fracton order is physically realizable?

In fact, this is a special case of a more general question about the p-modular fracton orders mentioned above.  The defining property of a p-modular fracton order, dubbed p-modularity, is that every nontrivial point-like excitation (not necessarily a planon) can be detected by braiding with a planon; this property implies that the principle of remote detectability holds.  In planon-only fracton orders, p-modularity, braiding nondegeneracy and the principle of remote detectability are all identical.  In Ref.~\onlinecite{Wickenden_2024}, the module $S$ and data of braiding between planons and arbitrary excitations were packaged into an algebraic structure dubbed a p-theory, which in planon-only fracton orders essentially reduces to the pair $(S,b)$.
It would be reasonable to expect that every p-modular p-theory is realizable because remote detectability is satisfied, but the example of Sec.~\ref{sec:example} shows this is not the case.

In this paper, we propose that fracton orders should satisfy the following, which subsumes the principle of remote detectability:
\begin{framed}
\noindent
\textbf{Excitation-detector principle.}  Every nontrivial excitation can be detected by some detector, where detectors are operators whose support lies at spatial infinity; this is the principle of remote detectability.  \emph{In addition, every detector is effective, meaning that it detects at least one nontrivial excitation.}
\end{framed}
\noindent  Here we introduce the notion of detectors, which are certain operators supported at spatial infinity; this clearly needs a precise definition for the excitation-detector principle to be applied.  In this paper, we do not give a general definition, but in the case of planon-only fracton orders we identify the detectors as string operators of planons that may have infinite spatial support in the transverse direction.  The module of such planon excitations is denoted $\widetilde{S}$ (see Sec.~\ref{sec:infinite} for a precise definition), and we identify detectors with elements of $\widetilde{S}$.   The crucial new feature as compared to the  principle of remote detectability is the additional condition on effectiveness of detectors, which turns out to be superfluous in the case of abelian anyons in 2d.  For planon-only fracton orders, this condition becomes nondegeneracy of a bilinear form $\tilde{b} : S \times \widetilde{S}  \to \Q / \Z$, which records the mutual statistics between finite and infinite planon excitations.

To substantiate our proposal and obtain a convenient algebraic theory of planon-only fracton orders, we prove the following theorem, whose second and third enumerated properties are explained below.

\begin{theorem}[Informal]
Given a planon-only fracton order with data $(S, \theta)$, where $\theta$ is nondegenerate, the following are equivalent:
\begin{enumerate}
\item The excitation-detector principle holds.  Precisely, the bilinear form $\tilde{b} : S \times \widetilde{S}  \to \Q / \Z$ is nondegenerate.
\item There exists some $N_0 \in \N$ such that for all $N \geq N_0$, the compactified theory $(S_N, \theta_N)$ is a modular theory of 2d abelian anyons.
\item The theory $(S,\theta)$ is perfect.
\end{enumerate}
\label{thm:main-intro}
\end{theorem}
\noindent  We give the formal statement of the theorem in Sec.~\ref{sec:main-proof} as Theorem~\ref{thm:main-detailed}, adding two additional equivalent properties that are primarily technical in nature.
To explain the second property, if we compactify space by imposing periodic boundary conditions in a direction transverse to the planes of mobility, we obtain a theory of 2d anyons denoted $(S_N, \theta_N)$, where $N \in \N$ is the transverse size of the compactification.  For a given system, such compactifications are only guaranteed to be physical for sufficiently large $N$, \emph{i.e.} for $N \geq N_0$.  Property \#2 is certainly a necessary condition for physical realizability, so our theorem shows that the excitation-detector principle is also necessary, but is not guaranteed to be sufficient.  We conjecture it is both necessary and sufficient for physical realizability; establishing this is left for future work.

Perfectness, \emph{i.e.} the third property of Theorem~\ref{thm:main-intro}, is technically very useful, and is more convenient to work with than the other equivalent properties.  Schematically, the bilinear form $b$ associated with $\theta$ gives rise to a certain map $\Phi$.  Nondegeneracy is the condition that $\Phi$ is injective, while perfectness of $(S,b)$ means that $\Phi$ is an isomorphism.  Beyond technical utility, there is a direct physical interpretation of property \#3, which we give in Sec.~\ref{subsec:perfect}.

We now outline the remainder of the paper. In Section~\ref{sec:example}, we begin by giving an example of a planon-only fracton order that has braiding nondegeneracy, and thus satisfies the principle of remote detectability, but which cannot be realized physically. In Section~\ref{sec:remote}, we discuss the excitation-detector principle for general gapped quantum systems and give a heuristic discussion of the  set $D$ of detectors. We then specialize to planon-only fracton orders where we can identify $D$ precisely.

In Section~\ref{sec:theories}, we begin by recalling the mathematical framework we use to study planon-only fracton orders, reviewing background material on bilinear, sesquilinear and quadratic forms in Section~\ref{subsec:prelims}. In Section~\ref{subsec:definitions}, we define finite order planon-only p-theories and theories of excitations, the concept of perfectness discussed above, as well as other fundamental concepts used throughout the rest of the paper. A few examples of planon-only fracton orders are presented using our formalism in Section~\ref{subsec:examples}.  We provide additional physical intuition for perfectness in Section~\ref{subsec:perfect}.  Section~\ref{sec:compact} introduces the compactified theories associated to our theories of excitations; these are theories of abelian anyons obtained by forcing periodic boundary conditions in the direction transverse to the planes of mobility. Section~\ref{sec:infinite} is devoted to the study of planons with infinite support; this is where we introduce the module of detectors $\widetilde{S}$ and discuss the braiding between finitely and infinitely supported planons. A physical justification of the formal definitions of Section~\ref{subsec:infinite-formal} is given in Section~\ref{subsec:infinite-physical}. Section~\ref{sec:structure} is dedicated to the study of 
 the modules $S$ of superselection sectors of p-modular p-theories of finite fusion order. These are finitely generated and torsion free, and this section is dedicated to the study of such modules and their duals. In particular, we provide a structure theorem (Theorem~\ref{thm:f-presentations}) in the case when the fusion order is a prime power.
 
Section~\ref{sec:main-proof} is dedicated to the proof of the main result of the paper, Theorem~\ref{thm:main-intro}, restated as Theorem~\ref{thm:main-detailed}. In Section~\ref{sec:decoupled}, as a first illustration of the power of our algebraic theory of planon-only fracton orders, we prove that if a theory satisfying the equivalent conditions of Theorem~\ref{thm:main-intro} has $p x = 0$ for a prime $p$ and all $x \in S$, then the theory consists of decoupled 2d layers.  The paper concludes with a discussion of open issues in Section~\ref{sec:discussion}.  Appendix~\ref{app:local} contains the mathematical background which justifies our focus on prime power fusion order, and Appendix~\ref{app:ktheory} discusses certain $K$-theory invariants used in our study of the modules of superselection sectors.  

%% file: bulk.tex
%!TEX root = 1planar-v2.tex
\section{An unphysical planon-only fracton order}
\label{sec:example}

Here we give an example of a planon-only fracton order that has braiding nondegeneracy, and thus satisfies the principle of remote detectability, but nonetheless cannot be physically realized.  We need to specify the data $(S, \theta)$.  We define $S$ as the abelian group with generators $p_i$ for $i \in \Z$, which satisfy the relations $2 p_i = 0$.  Each $p_i$ represents a planon moving in the $xy$-plane with $z = i$, with $\Z_2$ fusion due to the relation.  A general element of $S$ can be expressed as $x = \sum_{i \in \Z} a_i p_i$, with $a_i \in \Z_2$ and only finitely many $a_i$ nonzero.  Because translations within $xy$-planes act trivially on the planons, we only need to consider perpendicular translations with symmetry group $\T_\perp = \Z$.  Denoting the generator of $\T_\perp$ by $t$, we have $t p_i = p_{i+1}$.  Note that we are working in infinite three-dimensional space.  

\begin{figure}
	\includegraphics[width=0.6\textwidth]{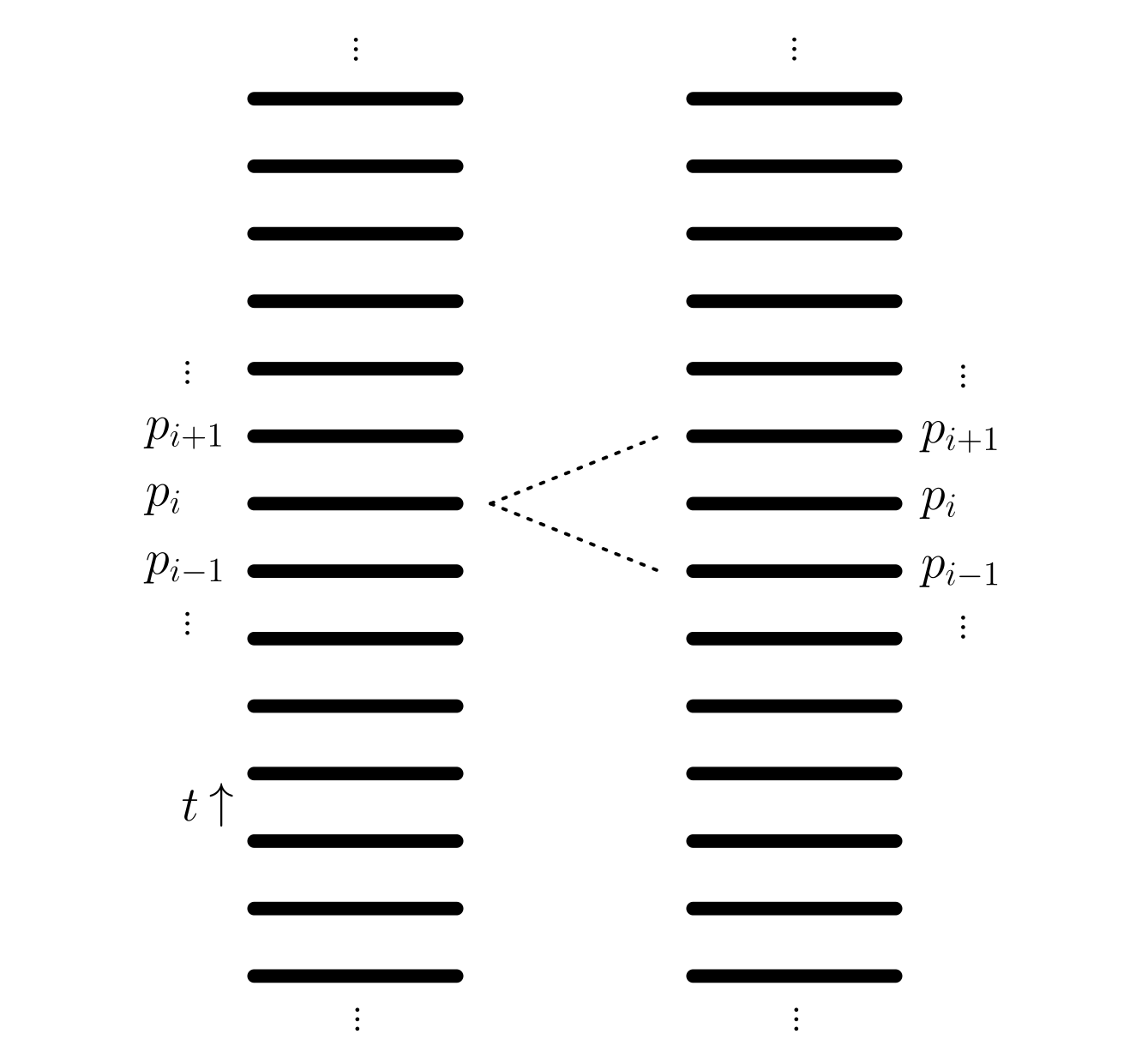}
	\caption{An unphysical planon-only fracton order. Each $p_i$ represents a planon moving in the $xy$-plane at $z = i$ (for $i \in \Z$).  The mutual statistics angle for braiding between nearest neighbors is $\pi$ (indicated by dashed lines); otherwise, the braiding is trivial.}
	\label{fig:example}
\end{figure}

To specify the statistics, we choose $\theta(p_i) = 0$ for all $i \in \Z$, and
\begin{equation}
b(p_i, p_j) = \left\{ \begin{array}{ll}
\tfrac{1}{2} , \qquad & j = i \pm 1 \\
0 , \qquad & \text{ otherwise}
\end{array}\right.  \label{eqn:example-b}
\end{equation}
This is illustrated in Fig.~\ref{fig:example}.  We recall  our convention that $\theta$ and $b$ take values in $\Q / \Z$, so $b(p_i, p_j) = 1/2$ corresponds to a mutual statistics angle of $\pi$.  By bilinearity, Equation~\ref{eqn:example-b} determines the value of $b$ on general elements of $S$.  We can also use $\theta(x+y) = \theta(x) + \theta(y) + b(x,y)$ to determine $\theta$ on all elements of $S$, and we find
\begin{equation}
\theta \Big( \sum_i a_i p_i \Big) =  \sum_{i < j} a_i a_j b(p_i, p_j) \text{,}  
\end{equation}
which can be evaluated using Equation~\ref{eqn:example-b}.
It is straightforward to check that $\theta$ satisfies the definition of a quadratic form (Definition~\ref{defn:quadratic}), and in particular 
 $b(x,y) = \theta(x+y) - \theta(x) - \theta(y)$ is satisfied for all $x,y \in S$, so $b$ is the bilinear form associated with $\theta$.
 
The theory $(S, \theta)$ satisfies the desired properties outlined in Sec.~\ref{sec:intro}.  In particular, $S$ is a finitely generated $\Z\T_\perp$-module, with any of the $p_i$ a single generator.  The bilinear form $b$ is compatible with spatial locality, with the mutual statistics vanishing between planons whose support in the $z$-direction is sufficiently far apart.  It is straightforward to check that $b$ satisfies the formal locality property given in Definition~\ref{defn:ti-local}.  Finally, $(S, \theta)$ has nondegenerate braiding.  Let $x = \sum_i a_i p_i$ be any nonzero element of $S$, and let $i_0 \in \Z$ be the largest integer for which $a_i \neq 0$.  Then $b(p_{i_0 + 1}, x) = 1/2 \neq 0$.  

We now argue that this planon-only fracton order cannot have a physical realization, despite satisfying the above properties.  We compactify space to two dimensions by imposing periodic boundary conditions in the $z$-direction.  We let $N \in \N$ be the thickness of the resulting slab along the $z$-direction, in units of the translation symmetry.  The resulting compactified theory has the group of superselection sectors $S_N \cong (\Z_2)^N$ generated by $p_0, \dots, p_{N-1}$, where $t p_{N-1} = p_0$. A general element $x \in S_N$ can thus be written $x = \sum_{i = 0}^{N-1} a_i p_i$.  The braiding statistics $b_N : S_N \times S_N \to \Q / \Z$ is still given by Equation~\ref{eqn:example-b}, interpreting the index $i + 1$ modulo $N$.  In order to guarantee that the compactification is physical, we assume $N \geq N_0$ for some sufficiently large $N_0$; see the more detailed comments on compactified theories in Sec.~\ref{sec:compact}.  We at least need $N_0 > 1$ for the compactified version of Equation~\ref{eqn:example-b} to make sense.  See Definition~\ref{defn:compact-theory} for a general formula giving the self-statistics quadratic form in the compactified theory.

We consider the excitation $y_N = \sum_{i=0}^{N-1} p_i$ of the compactified theory, which is a composite of a planon in every layer.  It is clear that $b(p_i, y_N) = 0$ for all the generators $p_i$, as illustrated in Fig.~\ref{fig:example}. Therefore $b(x, y_N) = 0$ for all $x \in S_N$.  The compactified theory is thus not modular and cannot be realized physically, hence the original theory $(S, \theta)$ must also be unphysical.

We can also reach the same conclusion without compactifying space if we allow for certain infinitely supported excitations.  We define $\widetilde{S}$ to be the module whose elements are sums $\tilde{y} = \sum_{j \in \Z} c_j p_j$, with no restriction on the $c_j$; in particular, infinitely many of the $c_j$ are allowed to be nonzero.  By locality of braiding statistics, we have a well-defined braiding between infinite excitations in $\widetilde{S}$ and finite excitations in $S$, namely $\tilde{b} : S \times \widetilde{S} \to \Q / \Z$ given by
\begin{equation}
\tilde{b}(x,\tilde{y}) = \sum_{i, j \in \Z} a_i c_j b(p_i, p_j) \text{,}
\end{equation}
where $x \in S$ is $x = \sum_{i \in \Z} a_i p_i$ and
 $\tilde{y} = \sum_{j \in \Z} c_j p_j$.  This is well-defined because only finitely many of the $a_i$ are nonzero.  Let $i_{{\rm max}}$ and $i_{{\rm min}}$ the largest and smallest values of $i$ for which $a_i \neq 0$.  Then we observe that the summand is only nonzero for $i_{{\rm min}} - 1 \leq j \leq i_{{\rm max}} + 1$.  

In parallel to the compactified case, we define the infinite excitation $\tilde{y}_{\infty} = \sum_{j \in \Z} p_j$, which again is an (infinite) composite of a planon in every layer.  This excitation is clearly nontrivial, in the sense that it cannot be created by any operator supported within an infinite cylinder extending along the $z$-direction and of finite radius in the $xy$-plane, so it should have non-vanishing statistics with some other excitation.  But it is clear that $\tilde{b}(x,\tilde{y}_\infty) = 0$ for all $x \in S$, so the theory $(S, \theta)$ has an (infinite) transparent excitation.

\section{The excitation-detector principle}
\label{sec:remote}

In this section, we discuss the principle of remote detectability and the excitation-detector principle for general gapped quantum systems.  We then specialize to planon-only fracton orders.  To begin, we consider systems in $d$ spatial dimensions with a finite range Hamiltonian and an energy gap.  Working in infinite Euclidean space, we assume there is a unique ground state.\footnote{This assumption \emph{does} hold for topological and fracton orders. The well-known ground state degeneracy is a consequence of periodic (or other nontrivial) boundary conditions.}  For simplicity, we focus on bosonic systems where all the nontrivial excitations are point-like and abelian.  We also assume discrete translation invariance with symmetry group $\T \cong \Z^d$.

As discussed in Sec.~\ref{sec:intro}, the set $S$ of superselection sectors is an abelian group and a module over $\Z\T$, the group ring of $\T$ with integer coefficients.  We also introduce the set $D$ of detectors.  We emphasize that we do not give a precise general definition of $D$, so our discussion is heuristic until we specialize to planon-only fracton orders where we can identify $D$ precisely.  In upcoming work, we will give a precise definition of $D$ that is valid more generally \cite{Shirley_inprogress}.  

Detectors are certain unitary operators whose support lies at spatial infinity.  We can think of them as large operators in a limit where the size of the operator is taken to infinity, so the operator has no support within any ball of fixed size.  Each detector thus moves excitations around at infinity, but does not create any excitations.  This implies that detectors commute when acting on the ground state, and we expect -- and assume -- that they can be chosen to commute in general, so $D$ is an abelian group.  Even though the group operation corresponds to operator multiplication, we think of $D$ as an additive group.  Translation symmetry acts on detectors, making $D$ into a module over $\Z\T$.

Remote detection is expressed through a bilinear pairing (dubbed the {\bf remote detection pairing}) between $D$ and $S$, namely
\begin{equation}
\Theta : S\times D  \to \R / \Z \text{,}
\end{equation}
where for $x \in S$ and $\hat{d} \in D$, $e^{2\pi i \Theta(x,\hat{d})}$ is the eigenvalue of $\hat{d}$ in a state containing the excitation $x$.  This does not depend on the representative state realizing a superselection sector $x$, because the support of $\hat{d}$ lies at spatial infinity.

In this language, the  principle of remote detectability says that every nontrivial excitation is detected by some detector.  That is, for any nonzero $x \in S$, there exists some $\hat{d} \in D$ such that $\Theta(x,\hat{d}) \neq 0$.  This is the statement of left-nondegeneracy of the bilinear pairing $\Theta$.  The excitation-detector principle is simply the statement that $\Theta$ is nondegenerate.  That is, $\Theta$ is both left- and right-nondegenerate, where the latter condition means that for any nonzero $\hat{d} \in D$, there is some $x \in S$ with $\Theta(x,\hat{d}) \neq 0$.  Every nontrivial detector is thus effective, in the sense that it detects some nontrivial excitation.  Here we observe a subtlety that arises in thinking about the definition of detectors.  Suppose we consider a commuting Pauli Hamiltonian, and construct a detector out of the local stabilizers that appear as terms in the Hamiltonian, using only stabilizers supported ``at infinity.''  Such an operator will not detect any excitations at all, and right-nondegeneracy of $\Theta$ would be impossible.  Therefore we only want to consider certain ``operators at spatial infinity'' as detectors.

For abelian anyons in 2d, detectors are simply anyon string operators, and we identify $D = S = \cA$.  The remote detection pairing $\Theta$ thus becomes the familiar symmetric bilinear form $b : \cA \times \cA \to \Q / \Z$ encoding mutual statistics.  By symmetry of $b$, there is no difference between remote detectability and the excitation-detector principle; the additional condition of effectiveness of detectors is superfluous.

Now we turn to planon-only fracton orders.  A first guess would be to again identify $D = S$, which would again give rise to a symmetric bilinear form $b : S \times S \to \R / \Z$.  However, if we proceed in this way, the condition that detectors are effective is again superfluous, and the excitation-detector principle does not rule out the unphysical example of Sec.~\ref{sec:example}.

Given that detectors are operators supported at spatial infinity, it is unnatural to restrict to string operators of planons of finite transverse support, where transverse refers to the common plane of mobility of the planons.  Instead we choose $D$ to be the module of planon string operators including those of infinite transverse support, and we identify $D = \widetilde{S}$, the module of possibly infinite planon excitations as defined in Sec.~\ref{sec:infinite}.  The remote detection pairing  is then $\tilde{b} : S \times \widetilde{S} \to \R / \Z$, and encodes the braiding between finite and infinite planons.  With these choices, the excitation-detector principle identifies the example of Sec.~\ref{sec:example} as unphysical, as discussed above.

Because planon mutual statistics should be local in the normal direction as discussed in Sec.~\ref{sec:intro}, it should not be necessary to work with infinite excitations and their braiding to determine whether a theory of a planon-only fracton order can be physically realized.  That is, we expect that all mutual statistics information encoded by $\tilde{b}$ is also encoded in the braiding $b$ between finite planon excitations.  Therefore, we should ask whether $b$ more directly satisfies some condition that is equivalent to the excitation-detector principle, \emph{i.e.} to nondegeneracy of $\tilde{b}$.  Indeed we find in Theorem~\ref{thm:main-intro} that the desired condition is that $(S,b)$ is a perfect p-theory, which is defined in Sec.~\ref{subsec:definitions}

In planon-only fracton orders, it is relatively straightforward to identify $D$ because we fully understand the statistical processes.  In more general fracton orders, the understanding of generalized self and mutual statistics of restricted mobility excitations is limited to specific examples.  Therefore, even for relatively simple fracton orders beyond the planon-only case, it is nontrivial to describe all  possible remote detection processes.  In fact,  $D$ can be determined given the module $S$, as will be discussed in future work \cite{Shirley_inprogress}.

\section{Theories of planon-only fracton orders}
\label{sec:theories}

The goal of this section is to introduce our formal framework for planon-only fracton orders.
We begin in Section~\ref{subsec:prelims} by reviewing concepts from the theory of bilinear, sesquilinear and quadratic forms that will be used in the rest of the paper. In Section~\ref{subsec:definitions} we introduce the formal framework, giving definitions of theory of excitations and p-theory for planon-only fracton orders. This is also where we define the notion of perfectness for such theories, one of the criteria for physical realizability. We illustrate our framework in Section~\ref{subsec:examples} by using it to describe a few known examples.  In Section~\ref{subsec:perfect}, we provide additional physical intuition for perfectness.

\subsection{Mathematical preliminaries: bilinear, sesquilinear and quadratic forms}
\label{subsec:prelims}

Here we review some standard definitions and facts on bilinear, sesquilinear and quadratic forms that will be used throughout the paper.  First we briefly establish some notational conventions.  The natural numbers $\N$ is the set $\N = \{1, 2, 3, \dots\}$, \emph{i.e.} we do not consider $0$ to be an element of $\N$.  All rings have an identity element.  For any commutative ring $R$,  we denote by $M_n(R)$ the set of $n \times n$ square matrices with elements in $R$, and let $M_{n,m}(R)$ be the set of $n \times m$ matrices with elements in $R$.

We often work with Laurent polynomial rings $R[t^\pm]$, where $R$ is a commutative ring and where the elements of $R[t^\pm]$ are polynomials of the form $f = \sum_{m \in \Z} c_m t^m$ with only finitely many $c_m \in R$ nonzero.  In particular we denote $\Ztpm = \Z[t^\pm]$. The set of Laurent polynomials with coefficients in $\Q / \Z$ is denoted $\Qlp$; this is not a ring, but it is a $\Ztpm$-module.  The coefficient of $t^m$ in a polynomial $f$ is written $(f)_m$.   The ring of {\bf ordinary polynomials} with coefficients in $R$, \emph{i.e.} those with no terms of negative degree, is denoted by $R[t]$.  

\begin{defn}
Let $f = \sum_{m \in \Z} c_m t^m$ be a Laurent polynomial; either $f \in R[t^\pm]$ or $f \in \Qlp$.  The {\bf upper degree} (resp. {\bf lower degree})  of $f$ is the largest (resp. smallest) value of $m \in \Z$ with $c_m \neq 0$.  These degrees are denoted $\deg_+ f$ and $\deg_- f$, respectively.  If $f = 0$, we define $\deg_+ f = -\infty$ and $\deg_- f = \infty$. The {\bf total degree} is the pair $\delta f = (\deg_- f, \deg_+ f)$.  If $f \neq 0$, we define the {\bf relative degree} to be $\deg f = \deg_+ f - \deg_- f$.
\label{defn:degrees}
\end{defn}

\begin{defn}
A {\bf commutative ring with involution} is a commutative ring $R$ together with a map $R \to R$ (called the involution), written $r \mapsto \overline{r}$, satisfying $\overline{r_1 + r_2} = \overline{r_1} + \overline{r_2}$, $\overline{r_1 r_2} = \overline{r_1} \, \overline{r_2}$, $\overline{1} = 1$, and $\overline{\overline{r}} = r$.
\end{defn}
\noindent Using the commutativity of $R$, the involution is clearly a ring automorphism.  Any commutative ring can be given the trivial involution $\overline{r} =r$.  When working with a Laurent polynomial ring $R[t^\pm]$, we always consider the involution defined by $\overline{a t^m} = a t^{-m}$ for all $a \in R$ and $m \in \Z$.

\begin{defn}
If $R$ is a commutative ring with involution and $M$ and $N$ are $R$-modules, a function $\phi : M \to N$ is an $R${\bf-linear map} of modules, also referred to as an $R$-module {\bf homomorphism}, if $\phi(m_1 + m_2) = \phi(m_1) + \phi(m_2)$ and $\phi( r m_1) = {r} \phi(m_1)$ for all $m_1, m_2 \in M$ and $r \in R$. The set of $R$-linear maps is denoted ${\Hom_R}(M,N)$, and made into an $R$-module by letting $(r \phi)(m) = \phi(\bar{r}  m) = \bar{r} \phi(m)$ for $r \in R$ and $\phi \in {\Hom_R}(M,N)$.

 A function $\psi : M \to N$ is an $R${\bf-antilinear map} of modules, also referred to as an $R$-module {\bf anti-homomorphism}, if $\psi(m_1 + m_2) = \psi(m_1) + \psi(m_2)$ and $\psi( r m_1) = \bar{r} \psi(m_1)$ for all $m_1, m_2 \in M$ and $r \in R$. The set of $R$-antilinear maps is denoted $\overline{\Hom_R}(M,N)$, and made into an $R$-module by letting $(r \psi)(m) = \psi(\bar{r}  m) = {r} \psi(m)$ for $r \in R$ and $\psi \in \overline{\Hom_R}(M,N)$.
  \end{defn}
\noindent Note that the action of $R$ on $\Hom_R(M, N)$ is twisted by the involution.
\begin{defn}
Let $R$ be a commutative ring with involution.  An $R${\bf-module with involution} $M$ is an $R$-module equipped with an $R$-antilinear isomorphism $M \to M$, written $m \mapsto \overline{m}$ for $m \in M$, which is its own inverse, \emph{i.e.} $\overline{\overline{m}} = m$ for all $m \in M$.
\end{defn}

The proof of the following proposition is routine:
\begin{proposition} 
\label{prop:iso-duals}
If $R$ is a commutative ring with involution, $M$ is an $R$-module, and $L$ is an $R$-module with involution, then $\Hom_R(M, L) \cong \overline{\Hom_R}(M,L)$ as $R$-modules. The $R$-linear isomorphism $\Hom_R(M, L) \to \overline{\Hom_R}(M,L)$ and its inverse are both given by $\phi \mapsto \overline{\phi}$, where by definition $\overline{\phi}(m) = \overline{\phi(m)}$.  
\end{proposition}

Every commutative ring with involution is a module with involution over itself.  An example of a module with involution that is not a ring is $\Qlp$ with the usual involution for Laurent polynomials; it is a $\Ztpm$-module with involution. 

\begin{defn}
 \label{def:dual}
Let $R$ be a ring with involution (which could be the trivial involution), $L$ a module with involution, and $M$ an $R$-module.  The $L$-valued {\bf dual} of $M$ is $M^* = {\Hom_R}(M, L)$.  Typically the modifier ``$L$-valued'' will be clear from context and is thus omitted. 
\end{defn}

\begin{remark}
 \label{remark:duals}
In the literature it is more common to work with $R$-valued duals, but in our applications it is convenient to consider more general $L$-valued duals.  In particular, when working with $R = \Ztpm$ we usually take $L = \Qlp$, which generalizes the Pontryagin dual.  
\end{remark}

\begin{defn}
If $R$ is a commutative ring with involution, and $M$, $N$ and $L$ are $R$-modules, an $L$-valued {\bf sesquilinear form} is a function $b : M \times N \to L$ satisfying the following properties:
\begin{enumerate} 
\item For all $m,m' \in M$ and $n, n' \in N$, $b(m + m', n) = b(m,n) + b(m',n)$  and $b(m, n + n') = b(m,n) + b(m,n')$.
\item For all $r \in R$, $m \in M$ and $n \in N$, $r b(m,n) = b( m, rn) = b(\overline{r} m,  n)$.
\end{enumerate}
A function where the above properties hold for the trivial involution is called a {\bf bilinear form}.  
\end{defn}

Sometimes we refer to $R$-bilinear or $R$-sesquilinear forms to make it clear which ring we are working with.  For example, a function can be $\Z$-bilinear (\emph{i.e.} a bilinear form on abelian groups) but not $R$-bilinear, even if $M$, $N$ and $L$ are $R$-modules.

\begin{defn}
A sesquilinear form $b : M \times N \to L$ is {\bf left-nondegenerate} if for any nonzero $m \in M$, there exists $n \in N$ such that $b(m,n) \neq 0$, and is {\bf right-nondegenerate} if for any nonzero $n \in M$, there exists $m \in M$ such that $b(m,n) \neq 0$.  A sesquilinear form is {\bf nondegenerate} if it is both left- and right-nondegenerate.
\end{defn}

\begin{defn}
Let $R$ be a commutative ring.  An element $r \in R$ is {\bf regular} if it is not a zero divisor.
\end{defn}

\begin{defn}
Let $R$ be a commutative ring and $M$ an $R$-module.  An element $m \in M$ is {\bf torsion} if $r m = 0$ for some regular $r \in R$.  The module $M$ is {\bf torsion-free} if it has no nonzero torsion elements.
\end{defn}
\noindent A ring is always torsion-free as a module over itself.

The following simple result plays an important role in the proof of Theorem~\ref{thm:main-intro}, because it allows us to assume that the module of superselection sectors is torsion-free.

\begin{proposition}
 \label{prop:nd-implies-tf}
Let $R$ be a commutative ring with involution, $M$ and $N$ $R$-modules, $L$ a torsion-free $R$-module, and $b : M \times N \to L$ a nondegenerate sesquilinear form.  Then both $M$ and $N$ are torsion-free. 
\end{proposition}

\begin{proof}
Suppose $n \in N$ is a nonzero torsion element, then there exists a regular $r \in R$ with $r n = 0$.  By nondegeneracy, there exists $m \in M$ with $b(m,n) \neq 0$.  Then $r b(m,n) \neq 0$ because $L$ is torsion-free.  But $r b(m, n) = b( m, rn) = 0$, a contradiction.  The argument that $M$ is torsion-free is identical, using the fact that $r$ is regular if and only if $\overline{r}$ is regular, which follows from the fact that involution is a ring automorphism.
\end{proof}

\begin{defn}
  \label{defn:amap}
Given a sesquilinear form $b : M \times N \to L$, the {\bf associated maps} $\Phi_{\mathrm{left}} : M \to {\Hom_R}(N, L)$ and 
$\Phi_{\mathrm{right}} : N \to \overline{\Hom_R} (M,L)$ are given by $\Phi_{\mathrm{left}}(m) = b(m, \cdot)$ and $\Phi_{\mathrm{right}}(n) = b(\cdot, n)$. 
\end{defn}
\noindent Clearly both $\Phi_{\mathrm{left}}$ and $\Phi_{\mathrm{right}}$ are $R$-linear.  Moreover, $b$ is left- (resp. right-) nondegenerate if and only if $\Phi_{\mathrm{left}}$ (resp. $\Phi_{\mathrm{right}}$) is injective.

\begin{defn}
A sesquilinear form $b : M \times N \to L$  is {\bf left-perfect} if $\Phi_{\mathrm{left}}$ is an isomorphism, and is {\bf right-perfect} if $\Phi_{\mathrm{right}}$ is an isomorphism. A sesquilinear form is {\bf perfect} if it is both left- and right-perfect.
\end{defn}

\begin{defn}
Let $R$ be a commutative ring with involution and let $M, L$ be $R$-modules.  A bilinear form $b : M \times M \to L$ is {\bf symmetric} if $b(m, m') = b(m', m)$ for all $m, m' \in M$.  Now let $L$ be an $R$-module with involution.  A sesquilinear form $b : M \times M \to L$ is {\bf Hermitian} if  $b(m, m') = \overline{b(m', m)}$ for all $m, m' \in M$.  We refer to an $R$-sesquilinear Hermitian form as an $R$-Hermitian form.
\end{defn}

\noindent For symmetric bilinear forms, we have $\Phi_{\mathrm{right}} = \Phi_{\mathrm{left}}$, while for Hermitian forms we have $\Phi_{\mathrm{right}} = \overline{\Phi_{\mathrm{left}}}$.  In both cases we thus define a single associated map $\Phi \equiv \Phi_{\mathrm{left}}$, and nondegeneracy is equivalent to injectivity of $\Phi$. Similarly, perfectness is equivalent to the condition that $\Phi$ be an isomorphism.

\begin{remark}
Given an $L$-valued $R$-Hermitian form, the associated map takes values in $M^*$, \emph{i.e.} $\Phi \colon M \to M^* ={\Hom_R}(M, L)$.  We could have instead defined the associated map to be $\Phi_{\mathrm{right}}\colon M \to \overline{\Hom_R}(M, L)$.
These choices are equivalent because $\Phi_{\mathrm{right}} = \overline{\Phi_{\mathrm{left}}}$, \emph{i.e.} these two maps are related by composition with the isomorphism $\overline{\Hom_R}(M, L) \to \Hom_R(M,L)$ defined in Proposition~\ref{prop:iso-duals}.
\end{remark}

Associated maps encode exactly the same data as forms:
\begin{proposition}
  \label{prop:b-Phi}
Let $\mathfrak b$ be the set of Hermitian forms $M \times M \to L$, and let $\mathfrak F \subset \Hom_R(M, {\Hom_R}(M, L))$ consist of all maps $\Phi$ such that $\Phi(m)(m') = \overline{\Phi(m')(m)}$ for all $m, m' \in M$.   Respectively, let $\mathfrak b$ be the set of symmetric bilinear forms $M \times M \to L$, and let $\mathfrak F \subset \Hom_R(M, \Hom_R(M, L))$ consist of all maps $\Phi$ such that $\Phi(m)(m') = \Phi(m')(m)$ for all $m, m' \in M$.  In both cases, there is a bijection $\mathfrak b \to \mathfrak F$ given by taking the associated $R$-linear map of $b \in \mathfrak b$.  That is, $b \mapsto \Phi_b$ where $\Phi_b(m) = b(m, \cdot)$. 
\end{proposition}

\begin{proof}
We prove the proposition in the Hermitian case; the proof of the symmetric bilinear case is essentially the same. Given $b \in \mathfrak b$, the associated map $\Phi_b$ satisfies $\Phi_b(m)(m') = b(m, m') = \overline{b(m',m)} = \overline{\Phi_b(m')(m)}$, so indeed $\Phi_b \in \mathfrak F$.

We define the inverse map by $\Phi \mapsto b_\Phi$, where $b_{\Phi}(m, m') = \Phi(m)(m')$.  This is anti-linear in the first argument and linear in the second, and $b_{\Phi}(m', m) = \Phi(m')(m) = \overline{\Phi(m)(m')} = \overline{b_\Phi(m, m')}$, so $b_\Phi$ is Hermitian.   It is straightforward to check this is in fact the inverse map.
\end{proof}

\begin{defn}
\label{defn:ortho-sum} 
Let $b_i : M_i \times N_i \to L$ be a finite collection of bilinear or sesquilinear forms, all valued in $L$ and indexed by $i$.  Then their {\bf orthogonal sum} $b = \perp_i b_i$ is a bilinear or sesquilinear form $b : M \times N \to L$, where $M = \bigoplus M_i$ and $N = \bigoplus N_i$, and where 
$b(m, n) = \sum_i b_i(m_i, n_i)$ for any $m \in M$ and $n \in N$.  Here $m = \sum_i m_i$ and $n = \sum_i n_i$ are the unique expressions of $m$ and $n$ in terms of elements $m_i \in M_i$ and $n_i \in N_i$.  For the orthogonal sum of two forms, we write $b = b_1 \perp b_2$. 
\end{defn}
\noindent  It is clear that if the $b_i$ are symmetric (or Hermitian), then so is their orthogonal sum $b$.

\begin{proposition}
If $b = \perp_i b_i$ as in Definition~\ref{defn:ortho-sum}, then $\Phi_{\mathrm{left}} = \bigoplus_i \Phi_{\mathrm{left},i}$ and 
$\Phi_{\mathrm{right}} = \bigoplus_i \Phi_{\mathrm{right},i}$, where $\Phi_{\mathrm{left}}$ (resp. $\Phi_{\mathrm{left},i}$) is the left associated map of $b$ (resp. $b_i$), and similarly for $\Phi_{\mathrm{right}}$ and $\Phi_{\mathrm{right},i}$.
\end{proposition}

\begin{proof}
We give the proof for the left associated maps; the proof for the right associated maps is similar.  We have ${\Hom_R}(N,L) \cong 
\bigoplus_i {\Hom_R}(N_i,L)$, where the isomorphism is $\lambda \mapsto ( \lambda|_{N_i} )_i$.  We identify ${\Hom_R}(N,L)$ with $\bigoplus_i {\Hom_R}(N_i,L)$, so we can think of $\Phi_{\mathrm{left}}$ as a map $\Phi_{\mathrm{left}} : \bigoplus_i M_i \mapsto \bigoplus_i {\Hom_R}(N_i,L)$, where $\Phi_{\mathrm{left}}(m) = ( \Phi_{\mathrm{left}}(m)|_{N_i} )_i =  ( \Phi_{\mathrm{left},i}(m_i) |_{N_i} )_i$, because
$\Phi_{\mathrm{left}}(m)(n) = \sum_i \Phi_{\mathrm{left},i}(m_i)(n_i)$ by definition of the orthogonal sum.  
\end{proof}

\begin{proposition}\label{prop:nd-respects-stacking}
Let $b_i : M_i \times N_i \to L$ be a finite collection of bilinear or sesquilinear forms.  Their orthogonal sum $b = \perp_i b_i$ is left-nondegenerate (resp. right-nondegenerate) if and only if all the $b_i$ are left-nondegenerate (resp. right-nondegenerate).  Similarly,  $b $ is left-perfect (resp. right-perfect) if and only if all the $b_i$ are left-perfect (resp. right-perfect).
\end{proposition}
\begin{proof}
We deal only with left-nondegeneracy and left-perfectness, the arguments for the right case being similar. The conditions of left-nondegeneracy and of left-perfectness are 
 equivalent to conditions on the associated maps.  We have $\Phi_{\mathrm{left}} = \bigoplus_i \Phi_{\mathrm{left},i}$. The direct sum of $R$-module maps is injective (resp. surjective) if and only if each of its components are  injective (resp. surjective).
\end{proof}

Now we turn to quadratic forms.  We start with a proposition showing that three properties used to define quadratic forms in the literature are equivalent.

\begin{proposition}
Let $A$ and $C$ be abelian groups and $q : A \to C$ a function such that $b : A \times A \to C$ defined by
\begin{equation}
b(x,y) = q(x+y) - q(x) - q(y) \nonumber
\end{equation}
for all $x,y \in A$ is $\Z$-bilinear.  Then $q(0) = 0$ and the following are equivalent:
\begin{enumerate}
\item $b(x,x) = 2 q(x)$ for all $x \in A$.
\item $q(n x) = n^2 q(x)$ for all $x \in A$ and all $n \in \N$.
\item $q(x) = q(-x)$ for all $x \in A$.
\end{enumerate}  \label{prop:q-equiv}
\end{proposition}

\begin{proof}
We have $q(0) = q(0 + 0) = q(0) + q(0) +b(0,0) = q(0) + q(0)$, so $q(0) = 0$.

$1 \implies 2$:  Property \#2 holds trivially for $n = 1$.  Assume it holds for $n \leq n_0$.  Then take $x \in A$ and consider
\begin{equation}
q((n_0 + 1) x ) = q(n_0 x) + q(x) + n_0 b(x,x) = (n_0^2 + 1 + 2 n_0) q(x) = (n_0 + 1)^2 q(x) \text{.} \nonumber
\end{equation}

$2 \implies 3$:  Take $n=-1$.

$3 \implies 1$:  $b(x,x) = - b(x,-x) = - ( q(x-x) - q(x) - q(-x) ) = q(x) + q(-x) = 2 q(x)$.
\end{proof}

\begin{defn}
Let $A$ and $C$ be abelian groups. A {\bf quadratic form} on $A$ is a function $q : A \to C$ such that $b : A \times A \to C$ defined by
\begin{equation}
b(x,y) = q(x+y) - q(x) - q(y) \nonumber
\end{equation}
for all $x,y \in A$ is $\Z$-bilinear, and such that $b(x,x) = 2 q(x)$ for all $x \in A$.  We say $q$ is nondegenerate if the bilinear form $b$ is nondegenerate.   \label{defn:quadratic}
\end{defn}

We will be interested in quadratic forms taking values in $\R / \Z$, where $e^{2 \pi i q(x)}$ is the topological spin of a point-like excitation $x$, and $e^{2 \pi i b(x,y)}$ is the mutual statistics between $x$ and $y$.  In this setting, the property $b(x,x) = 2 q(x)$ corresponds to the fact that a full braid of $x$ around itself is topologically equivalent to two half-braids (\emph{i.e.} two exchanges).

\begin{defn}
The {\bf exponent} of an abelian group $A$ is the smallest natural number $n \in \N$ such that $n x = 0$ for all $x \in A$.  If no such natural number exists, we say $A$ has infinite exponent.  \label{defn:exponent}
\end{defn}

We note that $\R / \Z$ has a unique subgroup isomorphic to $\Z_n$ (\emph{i.e.} the integers modulo $n$), generated by $1/n$.  In the following proposition, $\Z_n$ refers to this subgroup.

\begin{proposition}
Let $A$ be an abelian group and $q : A \to \R / \Z$ a quadratic form.  If $x \in A$ has order $n$, then $q(x) \in \Z_n$ for $n$ odd and $q(x) \in \Z_{2n}$ for $n$ even.

If all elements of $A$ are of finite order, then $q$ takes values in $\Q / \Z \subset \R / \Z$.   If $A$ has exponent $n$, then $q$ takes values in $\Z_{2n}$ and $b$ takes values in $\Z_n$.  
\end{proposition}

\begin{proof}
Because $b(x,x) \in \Z_n$, property \#1 of Proposition~\ref{prop:q-equiv} gives $q(x) \in \Z_{2n}$.  Moreover, property \#2 gives $q(x) \in \Z_{n^2}$.  We have $\Z_{2n} \cap \Z_{n^2} = \Z_{n}$ if $n$ is odd and $\Z_{2n} \cap \Z_{n^2} = \Z_{2n}$ if $n$ is even.  If all elements of $A$ are of finite order, applying the above result to each element implies $q$ takes values in $\Q / \Z$.  If $A$ has exponent $n$, then the order of every $x \in A$ divides $n$, and the result follows.
\end{proof}

The following definition is analogous to orthogonal sum of bilinear and sesquilinear forms:
\begin{defn}
Let $q_i : A_i \to C$ be a collection of quadratic forms all valued in $C$ and indexed by $i$.  Their {\bf orthogonal sum} is a quadratic form $q = \perp_i q_i$, where $q : A \to C$ with $A = \bigoplus_i A_i$, and $q(x) = \sum_i q_i(x_i)$ for $x \in A$, with $x = \sum_i x_i$ the unique expression with $x_i \in A_i$.
\end{defn}
\noindent It is easy to see that $q$ is indeed a quadratic form.  Moreover, if $b_i$ are the bilinear forms associated with each $q_i$, then $b = \perp_i b_i$, where $b$ is the bilinear form associated with $q$.

\subsection{Formal framework for planon-only fracton orders}
\label{subsec:definitions}

Here we give the formal definitions of theory of excitations and p-theory for planon-only fracton orders.  Our presentation follows in many respects the development of p-theories in Ref.~\onlinecite{Wickenden_2024}, but with considerable simplifications from the planon-only assumption.  Throughout we are interested in 1-planar fracton orders of bosonic gapped systems in three spatial dimensions, where all the nontrivial excitations are abelian planons.  We assume discrete translation symmetry with the group $\T \cong \Z^3$, and we have $\T \cong \T_\parallel \times \T_\perp$, where $\T_\parallel \cong \Z^2$ is the group of symmetries within the common plane of mobility of planon excitations.  Nonzero elements of $\T_\perp \cong \Z$ are translations transverse to the plane of mobility and act nontrivially on planon excitations.

As discussed in previous works, translation symmetry allows us to discuss mobility of excitations in a precise manner \cite{Haah_2013, Pai_2019}.  However, this comes with the disadvantage -- for our purposes -- of additional phase invariants coming from the translation symmetry, similar to phenomena in symmetry-enriched topological (SET) phases \cite{Pretko_2020}.  A solution is to impose translation symmetry coarsely, which means that when defining phase equivalence we should allow for breaking of $\T$ down to subgroups $\T' \subset \T$ with $\T' \cong \Z^3$; that is, we allow for enlarging the crystalline unit cell by a finite amount \cite{Haah_2013, Haah_2014, Pai_2019, Hermele_KITP_talk, Hermele_UQM_talk, Hermele_inprep}.  Upon suitable coarsening of the translation symmetry, we can assume that $\T_\parallel$ acts trivially on all superselection sectors, and henceforth we ignore $\T_\parallel$ and focus on $\T_\perp$. In the present paper we are not concerned with classification or invariants of phases, and indeed we do not even need a definition of phase, so coarse translation symmetry will play no further role.

From now on we identify $\T_\perp = \Z$ (rather than the slightly weaker assumption that $\T_\perp$ is isomorphic to $\Z$), which simplifies the presentation somewhat from the more general treatment of Ref.~\onlinecite{Wickenden_2024}.  We will work with modules over $\Z\T_\perp$, the group ring of $\T_\perp$ with integer coefficients. It is natural to identify $\Z\T_\perp = \Ztpm$, the ring of Laurent polynomials in the indeterminate $t$, \emph{i.e.} recall $\Ztpm = \Z[t^\pm]$.

\begin{defn}
A {\bf planon-only fusion theory} is a finitely generated $\Ztpm$-module $S$, where for any nonzero $x \in S$ the subgroup of $\T_\perp$ leaving $x$ invariant is trivial.
\end{defn}

The latter condition is the statement that there are no nontrivial fully mobile excitations.  We will not rely on this condition, because we will only be interested in p-modular theories, and p-modularity implies this condition. Nonetheless, we keep it to parallel the earlier treatment of Ref.~\onlinecite{Wickenden_2024}.

Given a fusion theory where every element of $S$ is of finite order (\emph{i.e.} for all $x \in S$ there is some $m \in \N$ with $mx=0$), because $S$ is finitely generated, the exponent of $S$ is finite (see Definition~\ref{defn:exponent}).  

\begin{defn}
A planon-only fusion theory is of {\bf finite order} if every $x \in S$ is of finite order. In a finite order fusion theory, the exponent $n \in \N$ of $S$ is called the {\bf fusion order} of $S$.
\end{defn}

While there are planon-only fracton orders with infinite-order excitations \cite{MacDonald_1989,MacDonald_1990,Sondhi_2000, Sondhi_2001,Ma_2022}, in this paper we will exclusively be interested in finite order fusion theories.  We will comment briefly on the more general case below.  Given fusion theory with fusion order $n$, we observe the ideal $(n) \subset \Z\T_\perp$ annihilates $S$, and we set $R = \Ztpm/(n) = \Z_n[t^\pm]$.  It is important to emphasize that when discussing a finite order fusion theory, as well as the definitions below based on it, we always reserve the symbol $R$ for the ring $R = \Z_n[t^\pm]$.  It is often convenient to view $S$ as an $R$-module.

Like anyons in two spatial dimensions, planons are characterized by self and mutual statistics as discussed in Sec.~\ref{sec:intro}.  We incorporate the statistical data in Definition~\ref{defn:theory} below.

\begin{defn}
Let $\cR$ be a commutative ring, $M$ an $\cR[t^\pm]$-module, $L$ an $\cR$-module, and $b : M \times M \to L$ a symmetric $\cR$-bilinear form.  We say that $b$ is {\bf translation-invariant} if $b(x,y) = b(t^k x, t^k y)$ for all $x,y \in M$ and $k \in \Z$.  A translation-invariant bilinear form $b$ is {\bf local} if for any $x,y \in M$, there exists $\ell \in \N$ such that $|k| \geq \ell$ implies $b(t^k x, y) = 0$.

Letting $C$ be an abelian group and $q : M \to C$ a quadratic form, we say $q$ is {\bf translation-invariant} if $q(x) = q(t^k x)$ for all $k \in \Z$.  (Note this implies the associated bilinear form is translation-invariant.)  A translation-invariant quadratic form $q$ is local if the associated bilinear form is local. \label{defn:ti-local}
\end{defn}

\begin{defn}
A {\bf finite order planon-only theory of excitations} $(S, \theta)$ is a finite order planon-only fusion theory $S$ and a translation-invariant, local quadratic form $\theta : S \to \Q / \Z$.  In this paper we are only concerned with finite order planon-only theories, so we omit the qualifiers and refer simply to a {\bf theory of excitations}.  A theory is {\bf p-modular} if $\theta$ is nondegenerate.   An isomorphism $\alpha : (S, \theta) \to (S', \theta')$ between two theories is a $\Ztpm$-module isomorphism $\alpha : S \to S'$ such that $\theta'(\alpha(x)) = \theta(x)$ for all $x \in S$.  We write $(S, \theta) \cong (S', \theta')$ to indicate that two theories are isomorphic.
\label{defn:theory}
\end{defn}

While the quadratic form $\theta$ is important physically, in this paper the associated bilinear form $b$ plays a more important role.  The following definition reflects this emphasis.

\begin{defn}
A {\bf p-theory} $(S, b)$ is a finite order planon-only fusion theory $S$ together with a translation-invariant, local symmetric $\Z$-bilinear form $b : S \times S \to  \Q / \Z$.  A p-theory is {\bf quadratic} if $b$ is induced by some translation-invariant quadratic form $\theta : S \to  \Q / \Z$.  The theory is {\bf p-modular} if $b$ is nondegenerate. An isomorphism $\alpha : (S, b) \to (S', b')$ between two p-theories is a $\Ztpm$-module isomorphism $\alpha : S \to S'$ such that $b'(\alpha(x), \alpha(y)) = b(x,y)$ for all $x,y \in S$. We write $(S, b) \cong (S', b')$ to indicate that two p-theories are isomorphic.
\end{defn}

\begin{remark}
In the mathematical literature, with perhaps some additional assumptions, a theory of excitations $(S,\theta)$ may be called a \emph{quadratic space} and a p-theory $(S,b)$ a \emph{Hermitian space}. In that context, isomorphisms are typically called \emph{isometries}. 
\end{remark}

There is an obvious forgetful map from theories of excitations to p-theories that sends $(S,\theta) \mapsto (S,b)$; quadratic p-theories are precisely those lying in the image of this map.  In the physical setting of planon-only fracton orders, we are only interested in quadratic p-theories.  However we separate this condition from the definition of p-theory, partly to parallel the treatment of Ref.~\onlinecite{Wickenden_2024}.  That work studied p-modular fracton orders with planons of more than one orientation, where -- unlike the planon-only case -- the full statistical data beyond planon mutual statistics is not known.  In addition, none of the results proved for p-theories in this paper rely on a p-theory being quadratic.  The following result shows there is no real distinction between p-theories and theories of excitations in the case of odd fusion order:

\begin{proposition}
Let $(S,b)$ be a p-theory of odd fusion order $n$.  Then $(S,b)$ is quadratic, and moreover there is a unique quadratic form $\theta : S \to \Q / \Z$ such that $b$ is the bilinear form associated with $\theta$.
\end{proposition}

\begin{proof}
We view $S$ as a $\Z_n$-module, and $b$ as a bilinear form $b : S \times S \to \Z_n \subset \Q / \Z$. We observe that $2$ is invertible in $\Z_n$ because $n$ is odd.  Defining $\theta : S \to \Z_n$ by $\theta(x) = 2^{-1} b(x,x)$ for $x \in S$, we have $b(x,y) = \theta(x+y) - \theta(x) - \theta(y)$ for $x,y \in S$, and moreover $\theta$ is clearly a quadratic form.  To show uniqueness, suppose $\theta' : S \to \Z_n$ is a quadratic form for which $b$ is the associated bilinear form.  Then we have
\[
\theta(x) = 2^{-1} b(x,x) = 2^{-1} \big( \theta'(2 x) - \theta'(x) - \theta'(x) \big) = \theta'(x) \text{.} \qedhere 
\]
\end{proof}

Given a p-theory, it is often convenient to repackage the symmetric bilinear form $b : S \times S \to \Q / \Z$ into a $\Ztpm$-Hermitian form  $B : S \times S \to \Qlp$, where we recall that  $\Qlp$ is the $\Ztpm$-module with involution consisting of Laurent polynomials with coefficients in $\Q / \Z$.  The following proposition explains how this works.

\begin{proposition}
 \label{prop:bijection}
Let $S$ be a $\Ztpm$-module.  Let $\mathfrak b_{{\rm loc}}(S)$ be the set of symmetric, translation-invariant, local $\Z$-bilinear forms $S \times S \to \Q / \Z$.  Let $\mathfrak B (S)$ be the set of $\Ztpm$-Hermitian forms $S \times S \to \Qlp$.  Then there is a bijection $\mathfrak b_{{\rm loc}}(S) \to \mathfrak B(S)$ given by $b \mapsto B_b$, where
\begin{equation}
B_b(x,y) = \sum_{m \in \Z} b(t^{m} x, y) t^m  \text{.} \label{eqn:B-from-b}
\end{equation}
for $x,y \in S$.
Moreover, the inverse map $\mathfrak B(S) \to \mathfrak b_{{\rm loc}}(S)$ is given by $B \mapsto b_B$, where $b_B(x,y) = (B(x,y))_0$, \emph{i.e.} one takes the coefficient of $t^0$ in $B(x,y)$.

In addition, $b$ is nondegenerate if and only if $B_b$ is nondegenerate.
\end{proposition}

\begin{proof}
First we observe that Equation~\ref{eqn:B-from-b} is well-defined because $b$ is local, which guarantees only finitely many terms in the sum are nonzero.  We check that $B_b$ is $\Ztpm$-sesquilinear.  It is clear that $B_b$ is $\Z$-bilinear, so $\Ztpm$-sesquilinearity follows from the fact that  $B_b(t^k x, y) = B_b(x, t^{-k} y) = t^{-k} B_b(x,y)$ for all $k \in \Z$. Hermiticity is also straightforward to check, using translation-invariance and symmetry of $b$.

Next, given $B \in \mathfrak B(S)$, it is clear that $b_B$ is $\Z$-bilinear, and routine to check that $b_B$ is symmetric, translation-invariant and local.

It is clear that $b_{B_b} = b$.  Finally we consider
\begin{equation}
B_{b_B} (x,y) =   \sum_{m \in \Z} ( B(t^{m} x, y) )_0  t^m 
= \sum_{m \in \Z} (  t^{-m} B( x, y) )_0  t^m 
= \sum_{m \in \Z} (  B( x, y) )_m  t^m  = B(x,y) \text{.} \nonumber
\end{equation}
Thus the maps are inverses as claimed.

Now suppose $b$ is nondegenerate, so for any $x \in S$ there exists $y \in S$ with $b(x,y) \neq 0$.  Clearly also $B_b(x,y) \neq 0$ because $(B_b(x,y))_0 = b(x,y)$.  Suppose $b$ is degenerate, then there is some $x \in S$ for which $b(x,y) = 0$ for all $y \in S$.  Then also
\begin{equation}
B_b(x,y) = \sum_{m \in \Z} b(x, t^{-m} y) t^m = 0 \nonumber
\end{equation}
for all $y \in S$.
\end{proof}

By this result, we can equivalently specify the data of a p-theory as a pair $(S,b)$ or $(S,B)$.  By Proposition~\ref{prop:b-Phi}, we can also replace $B$ with its associated map $\Phi : S \to S^*$, where we set $S^* = {\Hom_{\Ztpm}}(S, \Qlp)$, 
and specify the p-theory as a pair $(S, \Phi)$.  We will use these different ways of giving the data of a p-theory interchangeably.  We note that nondegeneracy of $b$, nondegeneracy of $B$, and injectivity of $\Phi$ are all equivalent, and are the same as p-modularity.  Physically, this corresponds to the principle of remote detectability.

For a p-theory with fusion order $n$, $b$ takes values in $\Z_n \subset \Q / \Z$, and thus $B$ takes values in $R = \Z_n[t^\pm] \subset \Qlp$.  Moreover, because $n S = 0$, $B$ can be viewed as an $R$-Hermitian form $B : S \times S \to R$.  Similarly we can identify $S^* = {\Hom_{\Ztpm}}(S, \Qlp) = {\Hom_R}(S, R)$ and view $\Phi$ as an $R$-linear map $\Phi : S \to  {\Hom_R}(S, R)$. This point of view will be convenient when working with theories of a specific fixed fusion order.  However, when working with theories of different fusion orders, as when we consider stacking (see below), it is easier to consider $\Ztpm$-modules and forms valued in $\Q / \Z$ or $\Qlp$.  

We can now define perfectness of a p-theory, which is property \#3 of Theorem~\ref{thm:main-intro}:

\begin{defn}
A p-theory is {\bf perfect} if its associated map $\Phi : S \to S^*$ is an isomorphism.  A theory of excitations is perfect if its associated p-theory is perfect.
\end{defn}

As discussed in Sections~\ref{sec:intro} and~\ref{sec:example}, not every p-modular theory of excitations is physically realizable, but we conjecture that every perfect theory of excitations can be realized by a quantum lattice model with a spatially local Hamiltonian.  The term ``perfect'' comes from the fact that bilinear forms whose associated map is an isomorphism are sometimes called perfect pairings.  

Now we define stacking, which is induced by the usual stacking operation on quantum systems.

\begin{defn}
Let $(S_i, \theta_i)$ be a finite collection of theories of excitations indexed by $i$.  
Their {\bf stack} is a theory of excitations $(S, \theta) = \bigominus_i (S_i, \theta_i)$, where $S = \bigoplus_i S_i$ and $\theta = \perp_i \theta_i$.  Stacking of p-theories is defined in the same way, producing a stacked p-theory $(S, b)$ with $b = \perp_i b_i$.

\end{defn}
\noindent  Note also that $B = \perp_i B_i$. If each $(S_i, \theta_i)$ has fusion order $n_i$, then the fusion order $n$ of the stack is the least common multiple of the $n_i$.  In determining whether a theory is a stack of some other theories, it is enough to consider p-theories:

\begin{proposition}\label{prop:thetasandbswithstacking}
Let $(S,\theta)$ be a theory of excitations with $(S,b)$ the corresponding p-theory.  If $(S,b) = \bigominus_i (S_i, b_i)$, then 
$(S,\theta) = \bigominus_i (S_i, \theta_i)$, where $\theta_i(x_i) = \theta(x_i)$ for $x_i \in S_i$.  Moreover, each $(S_i, b_i)$ is quadratic and is
 the p-theory associated to $(S_i, \theta_i)$.  
\end{proposition}

\begin{proof}
Consider $x = \sum_i x_i \in S$, then  $\theta(x) = \sum_i \theta(x_i) + \sum_{i < j} b(x_i, x_j)$.  But for $i \neq j$, $b(x_i, x_j) = b_i(x_i,0) + b_j( 0, x_j) = 0$, so $\theta(x) = \sum_i \theta(x_i) = \sum_i \theta_i(x_i)$.  Moreover, if $x_i, y_i \in S_i$, we have $b_i(x_i,y_i) = 
b(x_i, y_i) = \theta(x_i + y_i) - \theta(x_i) - \theta(y_i) = \theta_i(x_i + y_i) - \theta_i(x_i) - \theta_i(y_i)$.
\end{proof}

\begin{defn}
A property of a theory of excitations {\bf respects stacking} if for any theory $(S, \theta) = \bigominus_i (S_i, \theta_i)$, the property holds for $(S, \theta)$ if and only if it holds for all the summands $(S_i, \theta_i)$.  The corresponding definition for p-theories is identical.  
\end{defn}

\begin{proposition}
Perfectness and p-modularity of theories of excitations (and of p-theories) respect stacking.  \label{prop:pmod-perfect-respect-stacking}
\end{proposition}

\begin{proof}
This is an immediate consequence of Proposition~\ref{prop:nd-respects-stacking}.
\end{proof}

In our proof of Theorem~\ref{thm:main-intro}, and for many other purposes, it turns out to be enough to focus on the case of prime-power fusion order due to the following result.  We note that an analogous result for Pauli codes was proved as Proposition~29 of Ref.~\onlinecite{Ruba_2022}.

\begin{theorem}
Let $(S, b)$ be a $p$-theory with fusion order $n = p_1^{k_1} \cdots p_m^{k_m}$, where the $p_i$ are the distinct prime factors of $n$.  Then
$(S, b) \cong \bigominus_{i=1}^m (S_i, b_i)$, where the fusion order of $(S_i, b_i)$ is $p_i^{k_i}$.  
\label{thm:prime-decomp}
\end{theorem}

The proof of Theorem~\ref{thm:prime-decomp} uses the standard theory of localization from commutative algebra and is given in Appendix~\ref{app:local}. As a consequence of Proposition~\ref{prop:thetasandbswithstacking}, the same decomposition holds for theories of excitations:

\begin{corollary}
Let $(S, \theta)$ be a theory of excitations with fusion order $n = p_1^{k_1} \cdots p_m^{k_m}$, where the $p_i$ are the distinct prime factors of $n$.  Then
$(S, \theta) \cong \bigominus_{i=1}^m (S_{i}, \theta_i)$, where the fusion order of $(S_{i}, \theta_{i} )$ is $p_i^{k_i}$. 
\end{corollary}

To conclude this section, we explain why perfectness is not a familiar property in theories of abelian anyons; in that setting, it is superfluous.  This is also a good opportunity to give a formal definition of a theory of abelian anyons, which will be needed later on.

\begin{defn}\label{defn:theoryofabeliananyons}
A {\bf theory of abelian anyons} is a pair $(\cA, \theta)$, where $\cA$ is a finite abelian group and $\theta : \cA \to \Q / \Z$ is a quadratic form.  We say $(\cA, \theta)$ is {\bf modular} if $\theta$ is nondegenerate.  A {\bf p-theory of abelian anyons} is a pair $(\cA, b)$, where $b : \cA \times \cA \to \Q / \Z$ is a symmetric bilinear form. The p-theory $(\cA, b)$ is modular if $b$ is nondegenerate.
\end{defn} 

As discussed in the introduction, modularity is a necessary and sufficient condition for a theory of abelian anyons to be physically realizable.  We will find it useful to allow for non-modular (and thus unphysical) theories of abelian anyons, so we do not include modularity in the definition.  

Given a theory of abelian anyons, we consider the mutual statistics bilinear form $b$ induced by $\theta$.  We have the associated map $\phi : \cA \to \cA^* = \Hom_{\Z}(\cA, \Q / \Z)$ given by $\phi(x) = b(x, \cdot)$, and also -- because $\cA$ is a finite abelian group -- we have $\cA \cong \cA^*$.  A homomorphism between isomorphic finite groups is injective if and only if it is surjective.  Therefore modularity implies that $\phi$ is an isomorphism and there is no need to consider this as an additional condition.

\begin{remark}
As noted above, there are interesting planon-only fracton orders with infinite-order excitations \cite{MacDonald_1989,MacDonald_1990,Sondhi_2000,Sondhi_2001,Ma_2022}.  These are realized by infinite-component Chern-Simons theories, and they have the property that the mutual statistics between excitations is irrational and decays exponentially with distance.  To extend our formal setup to encompass such fracton orders, we would need a different definition of locality, which should be a condition roughly of the form that for any $x, y \in S$, $b(t^k x, y) \to 0$ as $|k| \to \infty$.  We might also want to impose a condition on the rate at which $b(t^k x , y)$ approaches zero.  This will be interesting to explore in future work.
\end{remark}

\subsection{Examples of planon-only fracton orders}
\label{subsec:examples}

Here we give a few well-known examples of planon-only fracton orders within the formal framework presented above.  All of the examples can be realized physically by an exactly solvable lattice model or an infinite-component Chern-Simons theory, and all are perfect theories of excitations.  In each case we specify the module of superselection sectors $S$ in terms of a free $R$-module, where $R = \Z_n[t^\pm]$ and $n$ is the fusion order.  If $\Omega$ is a set, we recall the standard notation that $R\Omega$ is the free $R$-module with basis $\Omega$.  The quadratic form $\theta : S \to \Q / \Z$ is fully specified by giving the values of $\theta$ and $B$ on a generating set for $S$.

By forming a 3d system out of decoupled layers of any 2d topological order, we obtain a very simple class of planon-only fracton orders.  Even though such examples are somewhat trivial, it is worthwhile to see how they are described within our formalism.

\begin{example}[Toric code layers]
	This is perhaps the simplest fracton order, consisting of layers of the 2d $\Z_2$ toric code, which can obviously be realized as a commuting Pauli Hamiltonian.  The fusion order is $n=2$ and the theory of excitations is given by
	\begin{align}
		S &= R\{e, m\},
		\\
		\theta(e) &= \theta(m) = 0,
		\\
		B(e, m) &= \tfrac{1}{2},
		\\
		B(e, e) &= B(m, m) = 0.
	\end{align}
\end{example}

\begin{example}[Three-fermion layers]
This example also consists of decoupled layers of 2d topological orders, and it is physically realizable because the underlying 2d topological order is physical.  The module $S$ and planon mutual statistics $B$ is the same as for toric code layers, and thus the two examples have the same p-theory.  The only way to tell these examples apart is by planon self-statistics; that is, by the values of the quadratic form.  Here, we have $\theta(e) = \theta(m) = 1/2$.  Note that together with $B(e,m) = 1/2$, this implies $\theta(f) = 1/2$, where $f = e+m$ is the third nontrivial particle in a single layer; hence all three nontrivial particles in each layer are fermions.
\end{example}

\begin{example}[Twisted 1-foliated model]
\label{ex:t1f}
	This example, introduced in Ref.~\onlinecite{Shirley_2020}, is the simplest planon-only fracton order that does not consist of decoupled layers of a 2d topological order. It was constructed by condensing pairs of charges starting from decoupled layers of twisted $\Z_2 \times \Z_2$ toric codes.  It can be realized in an exactly solvable lattice model, and also as an infinite-component Chern-Simons theory.
	
	The fusion order is $n = 4$, and the module $S$ is not free but can be conveniently specified as a quotient of a free module, by
	\begin{equation}
	S = R\{e, m\} / \langle 2 m - (t + \bar t) e, 2 e \rangle \text{,} \nonumber
	\end{equation}
	where $\langle x_1, \dots, x_k \rangle$ denotes the $R$-module generated by $x_1, \dots, x_k$.
	Moreover, we have
	\begin{align}
		\theta(e) &= \theta(m) = 0, \nonumber
		\\
		B(m, m) &= \tfrac{1}{4} (t + \bar t), \nonumber
		\\
		B(e, m) &= \tfrac{1}{2}, \nonumber
		\\
		B(e, e) &= 0. \nonumber
	\end{align}
	In Remark~\ref{rem:t1f}, we express the presentation of $S$ in the form given in the structure theorem (Theorem~\ref{thm:f-presentations}) proved in Section~\ref{sec:structure}.
\end{example}

\subsection{Physical interpretation of perfect p-theories}
\label{subsec:perfect}

The main purpose of this section is to give a description of perfectness that may provide some physical intuition.  We introduce a $\Ztpm$-module of local detection maps, whose elements detect excitations in a slab of finite extent in the direction normal to the planes of mobility.  Perfectness translates to the statement that every excitation in $S$ realizes a distinct local detection map, and every local detection map is realized by some excitation.

Above we discussed packaging the mutual statistics of a p-theory $(S,b)$ in the map $\Phi : S \to S^*$ associated to the $\Ztpm$-Hermitian form $B$.  The same information is also encoded in the map associated to $b$, namely ${\phi \colon S \to \Hom_\Z(S, \Q/\Z)}$ defined by $\phi(x) = b(x, \cdot)$.  Injectivity of $\phi$ is again the same as p-modularity.  However, we should not expect $\phi$ to be an isomorphism even for a perfect p-theory, because $\HomZ(S, \Q / \Z)$ is expected to be uncountably large, as is certainly the case when $S$ is a free $R$-module.  It turns out that $\phi$ actually takes values in a better-behaved submodule.

First we make $\HomZ(S, \Q / \Z)$ into a $\Ztpm$-module by defining $(z \varphi)(x) = \varphi(\overline{z} x)$ for any $\varphi \in \HomZ(S, \Q / \Z)$, $x \in S$ and $z \in \Ztpm$.  With this choice, the map $\phi : S \to \HomZ(S, \Q / \Z)$ is a $\Ztpm$-module homomorphism, because $\phi$ is $\Z$-linear and $\phi(t^k x)(y) = b(t^k x, y) = b(x, t^{-k} y) = \phi(x)(t^{-k} y) = (t^k \phi(x) )(y)$.  Note that this relies on translation-invariance of $b$.

\begin{defn}
Let $S$ be a $\Ztpm$-module.  We say that $\varphi \in \HomZ(S, \Q / \Z)$ is {\bf local} if for each $x \in S$ there exists $\ell_x \in \N$ such that $\varphi(t^k x) = 0$ for $|k| \geq \ell_x$.
\end{defn}

\begin{proposition}
Let $S$ be a $\Ztpm$-module.   We denote by $\HomlocZ(S, \Q / \Z)$ the subset of local elements of $\HomZ (S, \Q / \Z)$.  Then $\HomlocZ(S, \Q / \Z)$ is a $\Ztpm$-submodule of $\HomZ(S, \Q / \Z)$.
\end{proposition}

\begin{proof}
The zero homomorphism is obviously local. 
Suppose $\varphi_1, \varphi_2 \in \HomlocZ(S, \Q / \Z)$ and consider $\varphi_1 + \varphi_2$.  We have $(\varphi_1 + \varphi_2)(t^k x) = 0$ for $|k| \geq \operatorname{max}(\ell^1_x, \ell^2_x)$, so for the new homomorphism $\varphi_1 + \varphi_2$ we can choose $\ell_x = \operatorname{max}(\ell^1_x, \ell^2_x)$.  Further, we need to check that if $\varphi \in \HomlocZ(S, \Q / \Z)$, then $z \varphi \in \HomlocZ(S, \Q / \Z)$ for any $z \in \Ztpm$.  It is enough to take $z = t^m$, then $(t^m \varphi)(t^k x) = \varphi(t^{k-m} x)$, which clearly vanishes for sufficiently large $|k|$.  
\end{proof}

The module $\HomlocZ(S,\Q / \Z)$ is the desired module of local detection maps.  Given $b \in \mathfrak b_{{\rm loc}}(S)$ (see Proposition~\ref{prop:bijection}), we clearly have $\phi(x) = b(x, \cdot) \in \HomlocZ(S, \Q / \Z)$, so we view $\phi$ as a $\Ztpm$-module map $\phi : S \to \HomlocZ(S, \Q / \Z)$.  The analog of Proposition~\ref{prop:b-Phi} holds in the local setting:

\begin{proposition}
Let $S$ be a $\Ztpm$-module.  Let $\mathfrak b_{{\rm loc}}(S)$ be the set of symmetric, translation-invariant local $\Z$-bilinear forms $S \times S \to \Q / \Z$, and let $\mathfrak f(S) \subset \Hom_\Ztpm(S, \HomlocZ(S, \Q/\Z))$ be the subset of maps $\phi$ such that $\phi(x)(y) = \phi(y)(x)$ for all $x,y \in S$.  Then there is a bijection $\mathfrak b_{{\rm loc}}(S) \to \mathfrak f (S)$ given by $b \mapsto \phi_b$ where $\phi_b(x) = b(x, \cdot)$.
\end{proposition}

\begin{proof}
Given $b \in \mathfrak b_{{\rm loc}}(S)$, we have seen above that $\phi_b$ is a $\Ztpm$-module map where $\phi_b(x) \in \HomlocZ(S, \Q / \Z)$.  Moreover $\phi_b(x)(y) = b(x,y) = \phi_b(y)(x)$, so $\phi_b \in \mathfrak f(S)$.  

We define the inverse map by $\phi \mapsto b_\phi$, where $b_{\phi}(x,y) = \phi(x)(y)$.  This is $\Z$-linear in both arguments and symmetric because $\phi \in \mathfrak f(S)$.  Using the fact that $\phi$ is a $\Ztpm$-module map, and the $\Ztpm$-module structure on $\HomlocZ(S,\Q / \Z)$,
\begin{equation}
b_\phi(t^k x , t^k y) = \phi(t^k x)(t^k y) = (t^k \phi(x) )(t^k y) = \phi(x)(y) = b_\phi(x,y) \text{,}
\end{equation}
so $b_\phi$ is translation-invariant.    Locality of $b_\phi$ follows from $b_\phi(t^k x, y) = \phi(y)(t^k x)$, which vanishes for sufficiently large $|k|$ because $\phi(y)$ is local.  

It is routine to check the functions $b \mapsto \phi_b$ and $\phi \mapsto b_\phi$ are inverses.
\end{proof}

We have seen that data of $b$ in a p-theory $(S,b)$ can be packaged into either $\Phi : S \to S^* = {\Hom_\Ztpm}(S,\Qlp)$ or $\phi : S \to \HomlocZ(S, \Q / \Z)$, which are both $\Ztpm$-module maps.  The following result tells us that it is entirely a matter of taste whether we work with $\Phi$ or $\phi$.  In particular, perfectness of $(S,b)$ is equivalent to $\Phi$ or $\phi$ being an isomorphism.

\begin{proposition}
 \label{prop:RZnloc}
Let $S$ be a $\Ztpm$-module.  The map $\alpha : S^* \to \HomlocZ(S,\Q / \Z)$ given by $\alpha(\lambda)(x) = (\lambda(x))_0$ for $\lambda \in S^*=\Hom_\Ztpm(S,\Qlp)$ is an $\Ztpm$-module isomorphism.

In addition, if $b$ is a symmetric, translation-invariant local $\Z$-bilinear form $b : S \times S \to \Q / \Z$ with associated map $\phi : S \to \HomlocZ(S,\Q / \Z)$, and $\Phi : S \to {\Hom_\Ztpm}(S,\Qlp)$ is the map associated to $B : S \times S \to \Qlp$, then $\phi = \alpha(\Phi)$. 
\end{proposition}

\begin{proof}
First we observe that $\alpha(\lambda) \in \HomlocZ(S,\Q / \Z)$, because $\alpha(\lambda)(t^k x) = (\lambda(t^k x) )_0 = (t^{k} \lambda(x))_0 = (\lambda(x))_{-k}$, which vanishes for $|k|$ sufficiently large.  Clearly $\alpha$ is a $\Z$-module map.  To show it is a $\Ztpm$-module map, it enough to observe $\alpha(t^k \lambda)(x) = ( (t^k \lambda)(x) )_0 = ( \lambda(t^{-k} x) )_0 = (t^k \alpha(\lambda))(x)$.  

We claim the map $\beta : \HomlocZ(S,\Q / \Z) \to S^*$ defined by $\beta(\varphi)(x) = \sum_{m \in \Z} \varphi(t^{-m} x) t^m$ for $\varphi \in \HomlocZ(S,\Q / \Z)$ is the inverse of $\alpha$.  It is easy to check that $\beta(\varphi)(t^k x) = t^{k} \beta(\varphi)(x)$, so indeed $\beta(\varphi) \in \Hom_\Ztpm(S,\Qlp)$.  Moreover we have $\beta(t^k \varphi)(x) = t^{-k} \beta(\varphi)(x) = (t^k \beta(\varphi))(x)$, so $\beta$ is $\Ztpm$-linear.  Now
\begin{equation}
(\beta \alpha)(\lambda)(x) = \sum_{m \in \Z} ( \lambda(t^{-m} x) )_0 t^m = \sum_{m \in \Z} (\lambda(x))_m t^m = \lambda(x) \nonumber \text{.}
\end{equation}
It is easy to see that $(\alpha \beta)(\varphi)  = \varphi$.

Finally, let $b : S \times S \to \Q / \Z$ be a $\Z$-bilinear form as specified.  Then the associated maps are $\phi(x) = b(x, - )$, and $\Phi(x) = \sum_{m \in \Z} b(t^{m} x, - ) t^m$.  We have $\alpha( \Phi(x) )(y) = b(x,y) = \phi(x)(y)$, so $\phi = \alpha(\Phi)$.  
\end{proof}

\section{Compactified theories}
\label{sec:compact}

Given a finite order planon-only theory of excitations,  we would like to understand how to obtain a theory of abelian anyons in two spatial dimensions by compactifying a spatial direction transverse to the planes of mobility.  At the level of translation symmetry, compactification corresponds to taking the quotient $\T_\perp / N \T_\perp = \Z / N \Z = \Z_N$, which identifies $t^N$ with the identity; $N \in \N$ is thus the length of the compactified dimension.  At the level of the Hamiltonian and operator algebra, compactification means we identify lattice sites and single-site operators related by $t^N$.

For a given system, compactification is only guaranteed to behave nicely for sufficiently large $N$.  Suppose we start with a system realizing a planon-only fracton order with a correlation length $\xi$ in the transverse direction.  Then if $N \gg \xi$, and if $N$ is also much larger than the range of terms in the Hamiltonian, we expect the local density matrices of the compactified ground state in three-dimensional regions of linear size much less than $N$ will be close to those of the original ground state, with equality holding asymptotically as $N \to \infty$ if we keep the size of a region fixed.  In contrast, if $N \sim \xi$, then the properties of the compactified system may have nothing to do with the original three-dimensional system.  

For sufficiently large $N$, we also expect the ground state to remain unique for a compactified system of infinite extent in the two remaining spatial directions, because no zero-form or subsystem symmetries arise upon compactification that could be spontaneously broken.  (We note this does not hold for more general fracton orders, where compactification can lead to broken zero-form symmetries.  It also does not hold for compactifications of planon-only fracton orders oriented so that there are finite closed planon string operators, which lead to spontaneously broken subsystem symmetries.)  Finally, we expect nontrivial point-like excitations of size much less than $N$ to survive compactification as nontrivial point-like excitations.

Now we give a formal description of compactification, starting from a theory of excitations $(S,\theta)$.  We give a definition 
for all values of $N \in \N$, which is guaranteed to be physically meaningful for $N \geq N_0$ for some $N_0 \in \N$.\footnote{It is not \emph{a priori} obvious that a nice mathematical definition should be possible for all $N \in \N$ for a general theory of excitations $(S,\theta)$. However, we find that our definition is well behaved even down to $N =1$. }   Let $I_N \subset \Ztpm$ be the ideal $I_N = (1 - t^N)$, then $\Ztpm_N \equiv \Ztpm / I_N \cong \Z[ \Z_N ]$.  Elements of $\Ztpm_N$ are ``periodic polynomials'' of the form $\sum_{m =0}^{N-1} a_m t^m$, where $a_m \in \Z$.  Here, $t$ is the generator of $\Z_N$, so $t^{m} t^{m'} = t^{(m + m') \operatorname{mod} N}$.  
 The involution on $\Ztpm_N$ is given by $\overline{t} = t^{-1} = t^{N-1}$, and this agrees with the usual involution on the group ring $\Z[\Z_N]$ coming from the inverse operation on $\Z_N$.  We have the quotient ring homomorphism $\rho_N : \Ztpm \to \Ztpm_N$, where $\rho_N(z) = z + I_N$ for $z \in \Ztpm$.  When it is clear from context, we suppress the ``$N$'' subscript and write $\rho = \rho_N$.  The homomorphism $\rho$ is compatible with the involution, namely $\rho(\overline{z}) = \overline{ \rho(z) }$.

In the compactified theory, we want to identify $x \in S$ with $t^{\pm N} x$.  This is accomplished by the following:
\begin{defn}
Given a finite order planon-only fusion theory $S$ and $N \in \N$, the {\bf length-$N$ compactification} is the $\Ztpm$-module $S_N = S / I_N S$.
\end{defn}
\noindent  We have the quotient map $\rho_N : S \to S_N$ defined by $\rho_N(x) = x + I_N S$; we slightly abuse notation by using the same symbol as for the ring homomorphism above.  We have $t^{\pm N} \rho(x) = \rho( t^{\pm N} x) = \rho(x)$.  While $S_N$ is defined as a $\Ztpm$-module, it can be viewed as a $\Ztpm_N$-module by defining $\rho(z) \rho(x) = z \rho(x)$.  

Letting $n$ be the fusion order of $S$, and recalling $R = \Z_n[t^\pm]$, we also denote $I_N = (1-t^N) \subset R$, and define $R_N = R / I_N = \Ztpm / (n, 1-t^N) = \Ztpm_N / (n)$.  We can clearly view $S_N$ as an $R_N$-module in the same way we view $S$ as an $R$-module, since $n S_N = 0$.  We also denote the natural ring homomorphism $R \to R_N$ by $\rho_N$.  

We still need to specify the statistics of excitations in the compactified theory; this is accomplished by the following two definitions.

\begin{defn}
 \label{defn:compact-ptheory}
Let $(S,b)$ be a p-theory.  For $N \in \N$, the {\bf length-$N$ compactification} is the p-theory of abelian anyons $(S_N, b_N)$, where $b_N : S_N \times S_N \to \Q / \Z$ is defined by
\begin{equation}
b_N( \rho(x), \rho(y)) \equiv \sum_{m \in \Z} b(t^{m N} x, y)  =  \sum_{m \in \Z} b( x, t^{m N}y)  \nonumber 
\end{equation}
for all $x,y \in S$.
\end{defn}

\begin{defn}\label{defn:compact-theory}
Given a theory of excitations $(S, \theta)$ and $N \in \N$, the {\bf length-$N$ compactification} is the theory of abelian anyons $(S_N, \theta_N)$, where $\theta_N : S_N \to \Q / \Z$ is defined by
\begin{equation}
\theta_N( \rho(x) ) \equiv \theta(x) + \sum_{m > 0} \big( \theta(x + t^{m N} x) - 2 \theta(x) \big) = \theta(x) + \sum_{m > 0} b(x, t^{m N} x)  \nonumber
\end{equation}
for all $x \in S$. 
\end{defn}

Intuitively, the expression for $b_N$ says that, upon compactification, an excitation $x$ has mutual statistics with $y$ and all its translates by multiples of $t^N$.  While the formula for $\theta_N$ is less intuitive, we give a more careful physical justification for both formulas below.
It may seem strange that the formula for $\theta_N$ involves a sum over positive integers $m$, but we observe that $\sum_{m > 0} b(x, t^{m N} x) = \sum_{m < 0} b(x, t^{m N} x)$ by symmetry and translation-invariance of $b$. Before proceeding to physical justification, we check that the definitions are mathematically sensible.  The sums defining $b_N$ and $\theta_N$ have only a finite number of nonzero terms by locality of $b$ and $\theta$.  In addition, we have the following:

\begin{proposition}
The functions $b_N$ and $\theta_N$ given in Definitions~\ref{defn:compact-ptheory} and~\ref{defn:compact-theory} are well-defined, with $b_N(\rho(x), \rho(y))$ and $\theta_N(\rho(x))$  depending only on $\rho(x), \rho(y) \in S_N$.  Moreover, $b_N$ is $\Z$-bilinear, symmetric and translation-invariant, while $\theta_N$ is a translation-invariant quadratic form.  Further, Definitions~\ref{defn:compact-ptheory} and~\ref{defn:compact-theory} are compatible, in the sense that $(S_N, b_N)$ is the p-theory of abelian anyons associated to the theory of abelian anyons $(S_N, \theta_N)$.
\end{proposition}

\begin{proof}
We first check that $b_N$ is well-defined.  Suppose $x'$ and $y'$ are different representatives for $\rho(x)$ and $\rho(y)$, then we have $x' = x + (1-t^N)x''$ and $y' = y + (1-t^N)y''$, and
\begin{eqnarray}
\sum_{m \in \Z} b(t^{m N} x', y') &=& \sum_{m \in \Z} b(t^{m N} x, y') + \sum_{m \in \Z} b(t^{m N} (1 - t^N) x'', y')  \nonumber \\
&=& \sum_{m \in \Z} b(t^{m N} x, y') + \sum_{m \in \Z} b(t^{m N}  x'', y') - \sum_{m \in \Z} b(t^{(m+1) N} x'', y')  \nonumber \\
&=& \sum_{m \in \Z} b(t^{m N} x, y') \nonumber \\
&=& \sum_{m \in \Z} b(t^{m N} x, y) + \sum_{m \in \Z} b(t^{m N} x, y'') - \sum_{m \in \Z} b(t^{m N} x, t^N y'') \nonumber \\
&=& \sum_{m \in \Z} b(t^{m N} x, y)  \text{.}  \nonumber
\end{eqnarray}
Moreover, $b_N$ is clearly $\Z$-bilinear, symmetric and translation-invariant.

Next we check that $\theta_N$ is well-defined.  Putting $x \to x'$ we have
\begin{eqnarray}
\theta_N(\rho(x)) &=& \theta(x') + \sum_{m > 0} b(x', t^{mN} x')  \nonumber \\
&=& \theta(x) + \theta( (1-t^N) x'') + b(x, (1-t^N)x'')  \nonumber \\
&+& \sum_{m > 0} \Big( b(x, t^{m N} x) + b( (1-t^N) x'', t^{m N} x)
+  b(x, t^{m N} (1-t^N) x'') + b( (1-t^N)x'' , t^{m N} (1-t^N) x'') \Big) \nonumber \\
&=& \theta_N(\rho(x)) + \theta( (1-t^N) x'') + \sum_{m > 0}  b( (1-t^N)x'' , t^{m N} (1-t^N) x'') +
\sum_{m \in \Z} b(t^{m N} x, (1-t^N) x'') \label{eqn:qN1} \\
&=& \theta_N(\rho(x)) + \theta( (1-t^N) x'') + \sum_{m > 0}  b( (1-t^N)x'' , t^{m N} (1-t^N) x'') \nonumber \\
&=& \theta_N(\rho(x)) + \theta(x'') + \theta(-x'') - b(x'', t^N x'') + b( (1-t^N)x'', t^N x'') \nonumber \\
&=& \theta_N(\rho(x)) + \theta(x'') + \theta(-x'') - b(x'', x'') \nonumber \\
&=& \theta_N(\rho(x)) + \theta(x'' - x'') = \theta_N(\rho(x)) \nonumber \text{.}
\end{eqnarray}
Here, the last term in Equation~\ref{eqn:qN1} is equal to $b_N(\rho(x), \rho(0)) = 0$.  Translation-invariance of $\theta_N$ follows from translation-invariance of $\theta$.  

Finally we compute
\begin{align}
\theta_N(\rho(x+y)) - \theta_N(\rho(x)) - \theta_N(\rho(y)) &= 
\theta(x+y) - \theta(x) - \theta(y)  \nonumber \\ 
&+ \sum_{m > 0} \Big( b(x+y, t^{m N}(x+y) ) - b(x, t^{m N} x) - b(y, t^{m N}y ) \Big) \nonumber \\
&= b(x,y) + \sum_{m > 0} \Big( b(x, t^{m N} y) + b(y, t^{m N} x) \Big ) \nonumber \\
&= b_N(\rho(x), \rho(y)) \text{.} \nonumber 
\end{align} 
Moreover, $\theta_N (m \rho(x)) = m^2 \theta_N( \rho(x) )$ for all $x \in S$ and $m \in \N$ because the same property holds for $\theta$, so $\theta_N$ is a quadratic form with $b_N$ the associated bilinear form.
\end{proof}

\begin{proposition}
Fixing $N \in \N$, modularity of $(S_N, b_N)$ respects stacking.   That is, if $(S,b) = \bigominus_i (S_i, b_i)$, then $(S_N, b_N)$ is modular if and only if all the $(S_{iN}, b_{iN})$ are modular.   \label{prop:mod-respects-stacking}
\end{proposition}

\begin{proof}
Recall that modularity is the same as nondegeneracy of the compactified bilinear form.  We have $I_N S = \bigoplus_i I_N S_i$, so $S_N \cong \bigoplus_i S_{i N}$.  Moreover, for any $x,y \in S$ with $x = \sum_i x_i$ and $y = \sum_i y_i$, we have 
$b(x,y) = \sum_i b_i(x_i, y_i)$, and from the definition of $b_N$ it follows that $b_N(\rho(x), \rho(y)) = \sum_i b_{i N}(\rho_i (x_i), \rho_i(y_i) )$, where $\rho_i : S_i \to S_{i N}$ is the natural map, which satisfies $\rho = \bigoplus_i \rho_i$.  Therefore $b_N = \perp_i b_{i N}$, and the result follows from Proposition~\ref{prop:nd-respects-stacking}.
\end{proof}

Modularity of $(S_N, \theta_N)$ is defined as modularity of $(S_N, b_N)$. Therefore, we have immediately:

\begin{corollary}
Fixing $N \in \N$, modularity of $(S_N, \theta_N)$ respects stacking. 
\end{corollary}

Now we give a physical justification of Definition~\ref{defn:compact-theory}.  The idea is to choose a finite set of generators for $S_N$ that correspond to planons whose transverse support is small compared to $N$.  Then we check that $\theta_N$ and $b_N$ behave as expected physically for these ``small'' excitations.  It is not clear what to expect for the statistics of ``large'' excitations of transverse size on the order of $N$.  However, the identity $\theta_N(x + y) = \theta_N(x) + \theta_N(y) + b_N(x,y)$ tells us that $\theta_N$ (and thus $b_N$) is completely determined by the values of $\theta_N$ and $b_N$ on generators.

Let $x_1, \dots, x_k \in S$ be $\Ztpm$-module generators for $S$.  Each generator $x_i$ corresponds to a point-like excitation, which in particular has finite extent in the transverse direction.  By locality of $b$, there exist $\ell_{i j} \in \N$ such that $b(x_i, t^m x_j) = 0$ for $|m| \geq \ell_{i j}$.  We set $\ell = \operatorname{max} \{ \ell_{i j} \}$, and consider a compactification of size $N \geq 2 \ell$.  We observe that $\{ t^m \rho( x_i)  \mid 0 \leq m < N , 1 \leq i \leq k \}$ is a set of $\Z$-module generators for $S_N$ -- these will be the small excitations of the compactified theory.

For each small excitation, we have $\theta_N ( t^m \rho(x) ) = \theta_N ( \rho(x)) = \theta(x)$, using $N > \ell$.  That is, the small excitations have the same topological spin in the compactified theory that they had prior to compactification, which is precisely what we expect for a sufficiently large compactification.  Now we consider mutual statistics between the small excitations; physically, we expect  $b_N ( t^m \rho(x_i), t^{m'} \rho(x_j) ) = b(t^{m + a N} x_i, t^{m'} x_j )$ for a unique integer $a$ (which depends on $m$ and $m'$) chosen to bring $t^m x_i$ and $t^{m'} x_j$ sufficiently close that the mutual statistics may be nonzero.  That is, mutual statistics between small excitations in the compactified theory is given by the mutual statistics of the unique pair of representative excitations prior to compactification that are close enough to have nontrivial braiding.  
 We have
\begin{equation}
b_N ( t^m \rho(x_i), t^{m'} \rho(x_j) ) = \sum_{a \in \Z} b( t^{m + a N} x_i, t^{m'} x_j ) \text{.} \label{eqn:ms}
\end{equation}
If $|m' - m| < \ell$, then only the $a = 0$ term in Equation~\ref{eqn:ms} may be nonzero.  For if $a \neq 0$, we have
\begin{equation}
| m' - m - a N | \geq \big| |a | N - |m' - m| \big| \geq  2 |a | \ell - |m' - m| > (2 |a| - 1) \ell \geq \ell \text{.}
\end{equation}
If $|m' - m| = \ell$, then we have $| m' - m - a N | \geq \ell$ for all $a \in \Z$, and $b_N( t^m \rho(x_i), t^{m'} \rho(x_j) ) = 0$.  Finally if $|m'- m| > \ell$,  we may have either $\ell < m' - m < N$, or $-N < m' - m < - \ell$.  In the first case, only the $a = 1$ term may be nonzero, while in the second, only the $a = -1$ term may be nonzero.  Therefore in all cases, $b_N( t^m \rho(x_i), t^{m'} \rho(x_j) )$ is equal to the mutual statistics of a single pair of small excitations prior to compactification, as expected.

Now we describe the compactified theories in a more formal manner.  We define $\Qlp_N = \Qlp / I_N \Qlp$, which is a $\Ztpm_N$-module consisting of periodic polynomials with $\Q / \Z$ coefficients and with terms of degree $0$ through $N-1$.  As in $\Ztpm_N$, $t^m t^{m'} = t^{(m + m') \mod N}$.  We again denote the natural map by $\rho_N : \Qlp \to \Qlp_N$. 

\begin{proposition}
There is a bijection between the set of symmetric, translation-invariant $\Z$-bilinear forms $S_N \times S_N \to \Q / \Z$, and the set of $\Ztpm_N$-Hermitian forms $S_N \times S_N \to \Qlp_N$.  The bijection is given by $b_N \mapsto B_N$, where $B_N : S_N \times S_N \to \Qlp_N$ is defined by
\begin{equation}
B_N\big(\rho(x),\rho(y)\big) = \sum_{g \in \Z_N} b_N \big(g \rho(x), \rho(y)\big) g  \label{eqn:BN} \text{,}
\end{equation}
for $x,y \in S$.
The associated $\Ztpm_N$-linear map $\Phi_N : S_N \to {\Hom_{\Ztpm_N}}(S_N, \Qlp_N)$ is given by $\Phi_N(x) = B_N(x, \cdot)$.  \label{prop:cpct-bijection}
\end{proposition}

\begin{proof}
The proof is analogous to that of Proposition~\ref{prop:bijection}.  The last statement is simply Definition~\ref{defn:amap}.
\end{proof}

Next we will see how to obtain $\Phi_N$ directly from $\Phi$.

\begin{proposition}
There is a $\Ztpm$-module map $\rho_{N*} : S^* \to \Hom_{\Ztpm_N}(S_N, \Qlp_N)$ characterized by the property that, for any $\lambda \in S^*$,
$\rho_{N*} (\lambda)$ is the unique element of $\Hom_{\Ztpm_N}(S_N, \Qlp_N)$ for which $\rho_{N*} (\lambda) ( \rho(x) ) = (\rho  \lambda)(x)$ for all $x \in S$.
\label{prop:rhostar}
\end{proposition}

\begin{proof}
In this proof, and elsewhere when it is clear from the context, we omit the ``$N$'' subscript and write $\rho_*$.  The construction of $\rho_*$ is illustrated by the following commutative diagram, where $\lambda$ is an arbitrary element of $S^*$, and explained below:
\begin{equation}
\begin{tikzcd}
S \arrow[r, "\lambda"] \arrow[d, "\rho"] & \Qlp \arrow[r, "\rho"] & \Qlp_N \\
S_N \arrow[rru,dashed, swap, "\rho_* (\lambda)"] & & 
\end{tikzcd} \nonumber
\end{equation}
We get a map  $S^* \to {\Hom_\Ztpm}(S, \Qlp_N)$ by post-composition with $\rho$, \emph{i.e.} $\lambda \mapsto \rho \lambda$, corresponding to the top row of the diagram.  Then we can identify ${\Hom_\Ztpm}(S, \Qlp_N) \cong  {\Hom_\Ztpm}(S_N, \Qlp_N)$; any $\varphi \in {\Hom_\Ztpm}(S, \Qlp_N)$ factors uniquely through $\rho : S \to S_N$, because $\varphi( I_N S) = 0$.  This gives the unique map $\rho_* (\lambda) \in {\Hom_\Ztpm}(S_N, \Qlp_N)$ shown as the dashed arrow of the diagram.  Finally, it is clear that ${\Hom_\Ztpm}(S_N, \Qlp_N) = {\Hom_{\Ztpm_N}}(S_N, \Qlp_N)$, as the action of $\Ztpm$ on both $S_N$ and $\Qlp_N$ factors through $\rho : \Qlp \to \Qlp_N$.  Therefore given $\lambda \in S^*$, we get $\rho_* (\lambda) \in {\Hom_{\Ztpm_N}}(S_N, \Qlp_N)$.  The claimed property of $\rho_*$ is just commutativity of the diagram, and uniqueness is clear because $\rho$ is surjective.
\end{proof}

Making use of $\rho_*$, we repeat a similar argument to get $\Phi_N$ from $\Phi$.
\begin{proposition}
There is a unique $\Ztpm$-module map $\Phi_N : S_N \to {\Hom_{\Ztpm_N}}(S_N, \Qlp_N)$ satisfying 
$\Phi_N( \rho(x)) = \rho_* \Phi ( x)$ for all $x \in S$.  \label{prop:PhiN}
\end{proposition}

\begin{proof}
We consider the commutative diagram
\begin{equation}
\begin{tikzcd}
S \arrow[r, "\Phi"] \arrow[d, "\rho"] & S^* \arrow[r, "\rho_*"] & {\Hom_{\Ztpm_N}}(S_N, \Qlp_N) \\
S_N \arrow[rru,dashed, swap, "\Phi_N"] & & 
\end{tikzcd} \nonumber
\end{equation}
As above, the map $\rho_*  \Phi$, shown in the top row of the diagram, factors uniquely through $\rho : S \to S_N$.  To see this, take $x \in I_N S$, so $x = (1-t^N)x'$.  Then $\rho_* \Phi (x) = \rho_* \Phi ((1-t^N) x') = (1-t^{N}) \rho_* \Phi(x') = 0$, so $\rho_* \Phi(I_N S) = 0$.    Again, the claimed property of $\Phi_N$ is  commutativity of the diagram, and uniqueness is clear.
\end{proof}

When working with a theory of fixed fusion order $n$, $B_N$ takes values in $R_N \subset \Qlp_N$, and we can think of $B_N$ as $R_N$-sesquilinear because $n S_N = 0$.  Similarly we can view $S_N$ as an $R_N$-module and identify ${\Hom_{\Ztpm_N}}(S_N, \Qlp_N) = {\Hom_{R_N}}(S_N, R_N)$, with $\Phi_N$ an $R_N$-linear map
$\Phi_N : S_N \to  {\Hom_{R_N}}(S_N, R_N)$.  Within this viewpoint, we have $\rho_* : {\Hom_R}(S, R) \to {\Hom_{R_N}}(S_N, R_N)$.

We have given two different constructions of $\Phi_N$; one directly from $\Phi$ as above, and one in Proposition~\ref{prop:cpct-bijection}, \emph{i.e.} $\Phi_N(x) = B_N(x, \cdot)$, where $B_N$ is given by Equation~\ref{eqn:BN}.  We need to check that these two constructions are the same:

\begin{proposition}\label{prop:PhiPhiNBBN}
Let $(S,b)$ be a p-theory with associated map $\Phi : S \to S^*$, and let $(S_N, b_N)$ be the length-$N$ compactification. We let $\Phi_N : S_N \to \Hom_{\Ztpm_N}(S_N, \Qlp_N)$ be the map defined in Proposition~\ref{prop:PhiN}.  Then $\Phi_N$ is the associated map of $B_N : S_N \times S_N \to \Qlp_N$, where $B_N$ is given in terms of $b_N$ in Equation~\ref{eqn:BN}. That is,  $\Phi_N(\rho(x))(\rho(y)) = B_N(\rho(x), \rho(y))$ for all $x,y \in S$.
\end{proposition}
\begin{proof}
We have
\begin{equation}
\Phi_N( \rho(x) )(\rho(y)) = (\rho_* \Phi) (x) (\rho(y)) = \rho_* ( \Phi(x)) ( \rho(y)) = \rho ( \Phi(x)(y) ) \text{.} \nonumber
\end{equation}
Now $\Phi(x)(y) = \sum_{m \in \Z} b(t^{m} x, y) t^m$, so
\[
\rho ( \Phi(x)(y) ) = \sum_{m = 0}^{N-1} \Big[ \sum_{k \in \Z} b(t^{k N} t^{m} x, y )  \Big] t^m
= \sum_{m = 0}^{N-1} b_N(\rho(t^{m} x), \rho(y) ) t^m = B_N(\rho(x), \rho(y)) \text{.}  \qedhere
\]
\end{proof}

We give $\Hom_{\Z}(S_N, \Q / \Z)$ a $\Ztpm_N$-module structure by $(r \psi)(\chi) = \psi (\overline{r} \chi)$ for $r \in \Ztpm_N$, $\psi \in \Hom_{\Z}(S_N, \Q / \Z)$ and $\chi \in S_N$. 
\begin{proposition}
Let $S$ be a fusion theory and $S_N$ the length-$N$ compactification.  Then $S_N$ is finite.  Moreover, ${\Hom_{\Ztpm_N}}(S_N, \Qlp_N) \cong \Hom_{\Z}(S_N, \Q / \Z)$ as $\Ztpm_N$-modules, and $S_N \cong {\Hom_{\Ztpm_N}}(S_N, \Qlp_N)$ as abelian groups.  \label{prop:finite}
\end{proposition}

\begin{proof}
Viewing $S$ as an $R$-module, let $x_1, \dots, x_k \in S$ be a set of generators for $S$.  A general element $\chi \in S_N$ can be written $\chi = \rho_N (\sum_i r_i x_i ) = \sum_i \rho_N(r_i) \rho_N(x_i)$, and there are only a finite number of expressions of this form because $R_N$ is finite.

Along the same lines as Proposition~\ref{prop:RZnloc}, we define maps $\alpha : {\Hom_{\Ztpm_N}}(S_N, \Qlp_N) \to \Hom_{\Z}(S_N, \Q / \Z)$ and  $\beta : \Hom_{\Z}(S_N, \Q / \Z) \to {\Hom_{\Ztpm_N}}(S_N, \Qlp_N)$ by $\alpha(\varphi)(\chi) = ( \varphi(\chi) )_0$ and $\beta(\psi)(\chi) = \sum_{m = 0}^{N-1} \psi(t^{-m} \chi) t^m$.  It is straightforward to check that these are $\Ztpm_N$-module maps and inverses of one another.  

Because $S_N$ is a finite abelian group, we have $S_N \cong \Hom_{\Z}(S_N, \Q / \Z)$.
\end{proof}

\begin{corollary}
Let $(S, b)$ be a p-theory and $(S_N, b_N)$ the length-$N$ compactification.  Then the following are equivalent properties of the associated map $\Phi_N : S_N \to {\Hom_{\Ztpm_N}}(S_N, \Qlp_N)$:
\begin{enumerate}
\item $\Phi_N$ is an isomorphism.
\item $\Phi_N$ is injective.
\item $\Phi_N$ is surjective.
\end{enumerate}
\label{cor:finite}
\end{corollary}

\begin{proof}
The proposition tells us that $\Phi_N$ is a homomorphism between isomorphic finite abelian groups.  We have $| S_N | / | \operatorname{ker} \Phi_N | = | \operatorname{im} \Phi_N |$.  Moreover, $\Phi_N$ injective is equivalent to $| \operatorname{ker} \Phi_N | = 1$, and $\Phi_N$ surjective is equivalent to $| \operatorname{im} \Phi_N | = | S_N|$. The result follows.
\end{proof}

\section{Planons of infinite support}
\label{sec:infinite}

In this section we discuss how to work with planon excitations of infinite spatial support.  Given p-theory $(S, b)$, we define the module $\widetilde{S}$ of infinitely supported planon excitations.  We then introduce the bilinear form $\tilde{b} : S \times \widetilde{S} \to \Q / \Z$ that gives the mutual statistics between infinite and finite planons, and plays the role of the remote detection pairing in planon-only fracton orders.  In Sec.~\ref{subsec:infinite-physical}, we provide a physical justification of our definition of $\widetilde{S}$.

\subsection{Formal treatment}
\label{subsec:infinite-formal}

\begin{defn}
Let $\cA$ be an abelian group.  The set of {\bf Laurent formal sums} with $\cA$ coefficients is denoted $\cA[[t^\pm]]$ and consists of elements $\sum_{m \in \Z} c_m t^m$, where $c_m \in \cA$.  In particular, we denote $\tZ = \Z[[t^\pm]]$, while $\tQlp$ denotes the set of Laurent formal sums with $\Q / \Z$ coefficients and $\tilde{R} = \Z_n[[t^\pm]]$ denotes the set of Laurent formal sums with $\Z_n$ coefficients.
\end{defn}
\noindent It is clear that $\cA[[t^\pm]]$ is an abelian group and moreover a $\Ztpm$-module.  We note that even if $\cA$ is a ring, $\cA[[t^\pm]]$ is not a ring because there is no sensible way to multiply arbitrary elements, in contrast to the case of formal sums with only positive degree terms.

\begin{remark}
As a $\Ztpm$-module, $\tZ$ is isomorphic to the $\Z$-linear dual $\Hom_{\Z}(\Ztpm, \Z)$, where the module structure on the dual is given by $(z \psi)(x) = \psi (\overline{z} x)$ for $z \in \Ztpm$. Similarly, $\tQlp$ is isomorphic to the Pontryagin dual  $\Hom_{\Z}(\Ztpm, \Q/\Z)$ and $\tilde{R}$ to the dual  $\Hom_{\Z}(\Ztpm, \Z_n)$. 
\end{remark}

\begin{defn}
Given finite order planon-only fusion theory $S$, the module of {\bf infinite excitations} is $\widetilde{S} =S\otimes_{\Ztpm} \tZ $.
\label{defn:infinite-excitation-module}
\end{defn}
\noindent Elements of $\widetilde{S}$ are superselection sectors of certain excitations of infinite spatial support, as discussed in more detail below. This definition gives us a way to formalize ill-defined but physically intuitive expressions for infinitely supported excitations like ``$\tilde{x} = \sum_{m \in \Z} c_m t^m x$''; instead, we write $\tilde{x} = ( \sum_{\m \in \Z} c_m t^m ) \otimes x$.  

Recall that we often view $S$ as an $R$-module, where $S$ is of fusion order $n$ and $R = \Z_n[t^\pm]$.  In this viewpoint, the role of $\tZ$ in defining the module of infinite excitations is played by the $R$-module $\tilde{R} = \Z_n[[t^\pm]]$.  We can identify $\tilde{R} = \tZ / (n) \tZ$, and thus we have a natural $\Ztpm$-module map $\tZ \to \tilde{R}$.  To identify $\widetilde{S}$ within this viewpoint, the following standard general result is helpful:

\begin{proposition}
Let $\cR$ be a commutative ring with $\cR$-modules $M$ and $N$, and let $I \subset \cR$ be an ideal with $I N = 0$.  Then the natural map $ N  \otimes_{\cR} M \to N  \otimes_{\cR} (M / I M) $ is an isomorphism.  \label{prop:qti}
\end{proposition}

\noindent We therefore have $S \otimes_{\Ztpm}  \tZ \cong S  \otimes_{\Ztpm}  \tilde{R}  \cong S \otimes_R  \tilde{R}$, and it is consistent to define $\widetilde{S} \cong S \otimes_R \tilde{R}  $.

There is a well-defined notion of mutual statistics between finite and infinite excitations.
\begin{proposition}
Given a p-theory $(S,B)$, there is a sesquilinear form $\tilde{B} : S \times \widetilde{S} \to \tQlp$, defined by $\tilde{B}( x,   y \otimes \tilde{z} ) = B(x,y) \tilde{z}$, for any $\tilde{z} \in \tZ$ and $x,y \in S$.  If $S$ has fusion order $n$, then $\tilde{B}$ takes values in $\tilde{R} \subset \tQlp$, and we can view $\tilde{B}$ as an $R$-sesquilinear form $\tilde{B} : S \times \widetilde{S} \to \tilde{R}$.  
\end{proposition}

\begin{proof}
The form $\tilde{B}$ arises from a function $S   \times S \times \tZ \to \tQlp$, which is anti-linear argument and linear in the last two, given by $(x, y,  \tilde{z}) \mapsto B(x,y)  \tilde{z}$.  The last statement is clear.
\end{proof}
\noindent The form $\tilde{B}$ has essentially the same physical interpretation as that of $B$; the value $\tilde{B}( x, y \otimes \tilde{z} )$ gives the mutual statistics between the excitations $x$ and $y \otimes \tilde{z}  $ and all of their relative translates by arbitrary powers of $t$.  This takes values in $\tQlp$ because there may be an infinite number of relative translates with nonzero mutual statistics.  

\begin{remark}\label{rem:Qcomments}
It is natural to ask about the relationship between $\Qlp  \otimes_{\Ztpm} \tZ $ and $\tQlp$.  There is an injective -- but not surjective -- $\Ztpm$-linear map $\zeta :\Qlp \otimes_{\Ztpm}  \tZ   \to \tQlp$ defined by $\zeta ( q \otimes  \tilde{z}  ) = q \tilde{z}$ for any $\tilde{z} \in \tZ$ and $q \in \Qlp$. The form $\tilde{B}$ can be expressed using $\zeta$ by $\tilde{B}( x, y  \otimes \tilde{z} ) = \zeta(  B(x,y) \otimes \tilde{z}  ) = B(x,y) \tilde{z}$. 

Note that the $\Ztpm$-modules $\Qlp  \otimes_{\Ztpm} \tZ $ and $\tQlp$ are not isomorphic. Any element of $\Qlp \otimes_{\Ztpm} \tZ  $ can be written in the form $ (1/d) \otimes\tilde{z} $ for some $\tilde{z} \in \tZ$ and $d \in \N$. It follows that for any $w \in \Qlp \otimes_{\Ztpm} \tZ  $, there exists $d \in \N$ with $d w = 0$. However, for $\tilde{q} = \sum_{m = 1}^\infty (1/m) t^m \in \tQlp$, there is no $d \in \N$ with $d \tilde{q} = 0$. 
\end{remark}

It is useful to unwind the definition of $\tilde{B}$ to extract a bilinear form that simply gives the mutual statistics between its two arguments.
\begin{defn}
Given a p-theory $(S, B)$, there is a $\Z$-bilinear form $\tilde{b} : S \times \widetilde{S} \to \Q / \Z$ given by
\begin{equation}
\tilde{b}(x, y\otimes \tilde{z}  ) = \Big( \tilde{B}(x, y\otimes \tilde{z}) \Big)_0 = \Big( B(x,y) \tilde{z} \Big)_0 \text{.} \nonumber
\end{equation} \label{defn:wtb}
\end{defn}
\noindent The physical interpretation is that $\tilde{b}(x, y\otimes \tilde{z})$, as the $0$-component of $\tilde{B}(x, y\otimes \tilde{z})$, is precisely the mutual statistics between $x$ and $y\otimes  \tilde{z}$, with no relative translation.  This is further substantiated by the following:

\begin{proposition}
Given a p-theory $(S,b)$, for any $x, y \in S$ and any $\sum_{m \in \Z} c_m t^m \in \tZ$, we have
\begin{equation}
\tilde{b} \Big(  x, y \otimes  \big( \sum_{m\in \Z} c_m t^m \big)  \Big) = \sum_{m \in \Z} c_m b( x, t^my) \text{.} \label{eqn:Theta}
\end{equation}
In addition, for any $\tilde{z} \in \tZ$, $x,y \in S$ and $m \in \Z$, we we have $\tilde{b}(t^m x, t^m y \otimes  \tilde{z}  ) = \tilde{b}( x, y \otimes \tilde{z} )$ and also
\begin{equation}
\tilde{b}(  t^m x, y \otimes \tilde{z} ) = \Big( t^{-m} \tilde{B}( x,  y \otimes \tilde{z} ) \Big)_0 = \Big( \tilde{B}( x,  y \otimes \tilde{z}) \Big)_m  \text{.} \label{eqn:tb-translates}
\end{equation}
\end{proposition}
\begin{proof}
Apply Definition~\ref{defn:wtb} and use $B(x,y) = \sum_{m \in \Z} b(t^{m} x, y) t^m$.  We note that right-hand side of Equation~\ref{eqn:Theta} is a sensible expression by locality of $b$.  Verification of the latter two statements is routine.
\end{proof}
\noindent Equation~\ref{eqn:Theta} is recognizable as the desired mutual statistics between $x$ and the infinite excitation ``${\tilde{y} = \sum_{m \in \Z} c_m t^m y}$.''

The right-hand side of Equation~\ref{eqn:Theta} looks similar to the definition of $b_N$ in the compactified p-theory $(S_N, b_N)$ (see Definition~\ref{defn:compact-ptheory}).  Indeed, the following result, which relates infinitely supported planons with the compactified theories, is an immediate application of Equation~\ref{eqn:Theta}.

\begin{proposition}
Given a p-theory $(S,b)$, for any $N \in \N$ we have
\begin{equation}
\tilde{b}( x, y \otimes \gamma_N) = b_N( \rho_N(x), \rho_N(y)) \label{eqn:tb-bN} \text{,}
\end{equation}
where $\gamma_N = \sum_{m \in \Z} t^{m N} \in \Ztpm$.
\end{proposition}

It is also interesting to consider a restricted subset of infinite excitations, namely those which are invariant under some translation $t^N$:
\begin{defn}
Given finite order planon-only fusion theory $S$, the module of {\bf periodic excitations} is
\begin{equation}
\widetilde{S}^{per} = \{ \tilde{x} \in \widetilde{S} \mid \exists N \in \N \text{ such that } t^N \tilde{x} = \tilde{x} \} \nonumber \text{.}
\end{equation}
\end{defn}

It is easy to see that $\widetilde{S}^{per}$ is indeed a submodule of $\widetilde{S}$.  Given a p-theory, mutual statistics between periodic and finite excitations is encoded in the $\Ztpm$-sesquilinear form $\tilde{B}^{per} : S \times \widetilde{S}^{per} \to \tQlp$ defined by restriction of $\tilde{B}$, \emph{i.e.} $\tilde{B}^{per} = \tilde{B} |_{S \times \widetilde{S}^{per}}$.  It is clear that $\tilde{B}^{per}$ takes values in $\tQlp^{per} = \{ \tilde{q} \in \tQlp \mid \exists N \in \N \text{ such that } t^N \tilde{q} = \tilde{q} \}$, the module of periodic Laurent formal sums with $\Q / \Z$ coefficients. Therefore we can also write $\tilde{B}^{per} : S \times \widetilde{S}^{per} \to \tQlp^{per}$.  We also have $\tilde{b}^{per} : S \times \widetilde{S}^{per} \to \Q / \Z$ given by the restriction $\tilde{b}^{per} = \tilde{b}|_{S \times \widetilde{S}^{per}}$.

\begin{proposition}
Given a p-theory $(S, B)$, $\tilde{B}$ is left-nondegenerate (resp. right-nondegenerate) if and only if $\tilde{b}$ is left-nondegenerate (resp. right-nondegenerate).  The same statement holds for $\tilde{B}^{per}$ and $\tilde{b}^{per}$.  \label{prop:Theta-Bt-nd}
\end{proposition}

\begin{proof}
It is obvious that $\tilde{b}$ left-nondegenerate (resp. right-nondegenerate) implies $\tilde{B}$ left-nondegenerate (resp. right-nondegenerate).  Assume that $\tilde{B}$ is right-nondegenerate, then for any nonzero $\tilde{y} \in \widetilde{S}$ there exists $x \in S$ with $\tilde{B}(x, \tilde{y}) \neq 0$.  Therefore for some $m \in \Z$, we have $\big( \tilde{B}(x, \tilde{y}) \big)_m \neq 0$.  Moreover $\big( \tilde{B}(x, \tilde{y}) \big)_m = \big( \tilde{B}( t^m  x, \tilde{y}) \big)_0$, so $\tilde{b}(t^m x, \tilde{y}) \neq 0$.  Assuming $\tilde{B}$ is left-nondegenerate, an analogous argument shows that for any nonzero $x \in S$, there exists $\tilde{y} \in \widetilde{S}$ with $\tilde{b}(x, \tilde{y}) \neq 0$.  The proof in the periodic case is identical.
\end{proof}

\begin{proposition}
If $(S,b) = \bigominus_i (S_i, b_i)$, then $\widetilde{S} \cong \bigoplus \widetilde{S}_i$ and $\tilde{b} = \perp_i \tilde{b}_i$.  Moreover, nondegeneracy of $\tilde{b}$, and also of $\tilde{b}^{per}$, respects stacking.  \label{prop:ird-respects-stacking}
\end{proposition}

\begin{proof}
It is clear that $\widetilde{S} \cong \bigoplus_i \widetilde{S}_i$.  We have $b = \perp_i b_i$ and $B = \perp_i B_i$.  For any $x \in S$ and $\tilde{y} \in \widetilde{S}$ we have unique expressions ${x} = \sum_i {x}_i$ and $\tilde y = \sum_i  \tilde y_i$, with ${x}_i \in {S}_i$ and $\tilde y_i \in \tilde S_i$. For each $\tilde{y}_i$ we write $\tilde{y}_i = \sum_j  y^i_j\otimes \tilde{z}^i_j $, where $\tilde{z}^i_j \in \Ztpm$ and $y^i_j \in S_i$.  Now we consider
\begin{eqnarray}
\tilde{b}({x}, \tilde y) &=& \sum_{ i', i, j} \tilde{b}( x_{i'},  y^i_j\otimes  \tilde{z}^i_j ) =
\sum_{i', i,  j}\big( B( x_{i'}, y^i_j ) \tilde{z}^i_j \big)_0  \nonumber \\
&=& \sum_{i, j} \big( B_i (x_i, y^i_j) \tilde{z}^i_j \big)_0
= \sum_i \big( \tilde{B}_i ( x_i, \tilde{y}_i) \big)_0 = \sum_i \tilde{b}_i ( x_i, \tilde{y}_i) \text{.} \nonumber
\end{eqnarray}
Therefore $\tilde{b} = \perp_i \tilde{b}_i$, and nondegeneracy of $\tilde{b}$ thus respects stacking by Proposition~\ref{prop:nd-respects-stacking}.  The same argument goes through in the periodic case because $\widetilde{S}^{per} \cong \bigoplus_i \widetilde{S}^{per}_i$.  
\end{proof}

\subsection{Physical justification}
\label{subsec:infinite-physical}

To give a physical justification of Definition~\ref{defn:infinite-excitation-module}, we show that $\widetilde{S}$ is isomorphic to a quotient of certain infinitely supported excitations.  For this discussion we view $S$ and $\widetilde{S}$ as $R$-modules.  For the moment we restore the three-dimensional translation symmetry, and let $\bR = \Z_n[ u^\pm, v^\pm, t^\pm ]$.  Planons are invariant under translations in the $u$-$v$ plane, but are acted on nontrivially by $t$ which represents transverse translations.  Thus $S$ is an $\bR$-module for which $u x = v x = x$ for all $x \in S$, so the action of $\bR$ on $S$ factors through $\bR / I_{uv}$, where $I_{uv} = (1-u, 1-v)$.  We can view $S$ as an $R$-module because $R = \Z_n[t^\pm] \cong \bR / I_{uv}$.   To see this, the ring homomorphism $\bR \to R$ defined by $u, v \mapsto 1$, $t \mapsto t$ vanishes on $I_{uv}$ and thus factors through $\bR / I_{uv}$, giving a map $\bR / I_{uv} \to R$.  It is straightforward to show this map is an isomorphism.

Now we consider a free resolution of $S$ as an $\bR$-module, suggestively written as $\cdots \to P \to E \to S \to 0$, with finitely generated free $\bR$-modules $P$ and $E$.  (We know $S$ has a free resolution with all terms finitely generated because it is finitely generated and $\bR$ is Noetherian.)    For many fracton orders realized in commuting Pauli Hamiltonians, the description as a Pauli Hamiltonian provides the first terms of a free resolution.  Elements of $E$ are finitely supported configurations of excitations, while elements of $P$ are finitely supported Pauli operators.  More generally, for any free resolution of $S$, it is reasonable to interpret the first two terms in this manner, so that $S$ is a quotient of excitations by the submodule of excitations createable by operators of finite support.

We introduce the module $\tilde{\bR} = \Z_n[u^\pm, v^\pm][[t^\pm]]$, whose elements are formal sums $\sum_{m \in \Z} c_m t^m$, where each $c_m \in \Z_n[u^\pm, v^\pm]$, \emph{i.e.} it is a finite Laurent polynomial in $u$ and $v$.  We have $\tilde{R} \cong \tilde{\bR} / I_{uv} \tilde{\bR}$, which can be shown along the same lines as the isomorphism $R \cong \bR / I_{uv}$ discussed above.  
Starting with the exact sequence $P \to E \to S \to 0$, tensoring with $\tilde{\bR}$ gives
$P \otimes_{\bR} \tilde{\bR}  \to  E \otimes_{\bR} \tilde{\bR}  \to  S \otimes_{\bR}  \tilde{\bR}  \to 0$, which is exact by right-exactness of the tensor product.

Proposition~\ref{prop:qti} gives $S \otimes_{\bR}  \tilde{\bR}  \cong S \otimes_{\bR} \tilde{R} $, which in turn is identical to $S \otimes_R  \tilde{R}  \cong \widetilde{S}$, so we have $\widetilde{S} \cong (E \otimes_{\bR} \tilde{\bR} ) / \epsilon(P \otimes_{\bR}  \tilde{\bR} )$.  Now because $E \cong \bR^m$, we have $E \otimes_{\bR}  \tilde{\bR} \cong \tilde{\bR}^m$, and similarly for $P \otimes_{\bR}  \tilde{\bR} $.  We can thus interpret elements of $E  \otimes_{\bR}  \tilde{\bR}$ as configurations of excitations whose support is finite within every $uv$-plane, but may be infinite along the $t$-direction.  Elements of $P  \otimes_{\bR}  \tilde{\bR}$ are infinite products of Pauli operators, with the same restriction on their support.  Note that while the support in each $uv$-plane is finite, it is not necessarily the case that the infinite excitations and Pauli operators are supported within a cylinder of finite radius extending along the $t$-direction.

\section{Superselection sectors of p-modular p-theories}
\label{sec:structure}

We turn to the study of the modules $S$ of superselection sectors of p-modular p-theories of finite fusion order. In this case, $S$ is finitely generated and torsion-free as an $R$-module, and we give a detailed study of such modules. By Theorem~\ref{thm:prime-decomp}, it is sufficient to consider the case when the fusion order $n$ is a prime power $p^k$. So we assume $n=p^k$ and focus on $R=\mathbb{Z}_{p^k}[t^\pm]$. Recalling the associated map $\Phi : S \to S^*$, we are also interested in the duals of finitely generated torsion-free modules. The main results of our study are a concrete description of finitely generated torsion-free $\mathbb{Z}_{p^k}[t^\pm]$-modules and of their duals. In the special case when $k=1$, we denote $A=\mathbb{Z}_{p}[t^\pm]$. In this case, $A$ is a PID and finitely generated torsion-free $A$-modules are simply the free $A$-modules of finite rank. The key tool for $\mathbb{Z}_{p^k}[t^\pm]$ is the notion of a free $A$-filtration, which provides an inductive approach that allows us to bootstrap results for $A$-modules to results for $R$-modules.

\subsection{Preliminaries}
\label{subsec:st-prelims}

Here we collect some results that will be useful in what follows.  In particular, we will need to understand the zero divisors of $\Z_{p^k}[t^\pm]$.

\begin{proposition}
In the ring $\Z_{p^k} = \Z / p^k \Z$ where $p$ is prime, the set of zero divisors is the ideal $(p)$.
\end{proposition}

\begin{proof}
We write $n \mapsto [n]$ for the natural map $\Z \to \Z_{p^k}$.  For all $[n] \in \Z_{p^k}$ we choose a representative $n \in \{ 0, 1, \dots, p^k - 1\}$.  Suppose $[n] \neq 0$ and $p \mid n$, then $n = p^s m$ for some $s < k$, and clearly $[p^{k-s} n] =0$.  Thus all elements of $(p)$ are zero divisors.

Now suppose $p \nmid n$.  Then the only $m \in \N$ for which $[m n] = 0$ must have $p^k \mid m$, \emph{i.e.} $[m] = 0$.  So $[n]$ is not a zero divisor.
\end{proof}

\begin{proposition}[{\cite[Theorem 2]{McCoy}}]
Let $R$ be a commutative ring and let $f \in R[t]$ be a zero divisor.  Then there is some non-zero $a \in R$ with $a f = 0$.  \label{prop:poly-zero-div}
\end{proposition}

\begin{proof}
Assume the claim is false, so there a zero divisor $f \in R[t]$ where $a f \neq 0$ for all non-zero $a \in R$. Let $g \in R[t]$ be a non-zero polynomial of minimal degree such that $f g = 0$; by assumption $\operatorname{deg} g > 0$.  We write $f = f_0 + f_1 t + \cdots + f_m t^m$.  Let $f_n$ be the highest degree coefficient satisfying $f_i g \neq 0$.  (Such a coefficient must exist, for if not then $f_i g = 0$ for all $i$, which implies $f_i g_j = 0$ for all $i$ and $j$. But then $g_j f = 0$ contradicting the assumption.)  Write $g = g_0 + g_1 t + \cdots + g_\ell t^\ell$.  We have 
$f g = (f_0 + \cdots + f_n t^n)(g_0 + \cdots + g_\ell t^\ell) = 0$, which implies $f_n g_\ell = 0$ and thus $\operatorname{deg} f_n g < \ell$.  But $(f_n g) f = f_n (g f) = 0$, contradicting the minimality assumption on $g$. 
\end{proof}

This is also true for Laurent polynomial rings:

\begin{proposition}
Let $R$ be a commutative ring and let $f \in R[t^\pm]$ be a zero divisor.  Then there is some non-zero $a \in R$ with $a f = 0$.  \label{prop:laurent-zero-div}
\end{proposition}

\begin{proof}
Let $f \in R[t^\pm]$ be a zero divisor, so there is non-zero $g \in R[t^\pm]$ with $f g = 0$.  We can multiply both sides of this equation by a sufficiently high power of $t$ to obtain $f' g' = 0$ where $f'$ and $g'$ have non-negative degree.  Then $f'$ is a zero divisor in $R[t]$ and by Proposition~\ref{prop:poly-zero-div} there is non-zero $a \in R$ with $a f' = 0$.  Clearly also $a f = 0$.
\end{proof}

\begin{proposition}
The set of zero divisors in $\Z_{p^k}[t^\pm]$ is the ideal $(p)$.  \label{prop:R-zero-div}
\end{proposition}

\begin{proof}
Take $f \in \Z_{p^k}[t^\pm]$ a zero divisor.  By Proposition~\ref{prop:laurent-zero-div} there is non-zero $a \in \Z_{p^k}$ with $a f = 0$.  In particular, each coefficient $f_i$ of $f$ is a zero divisor in $\Z_{p^k}$ and thus $p \mid f_i$ for all the coefficients, so $f = p g$ for some $g \in \Z_{p^k}[t^\pm]$ and $f \in (p)$.    Now if $f \in (p)$, then clearly $p^{k-1} f = 0$ and $f$ is a zero divisor.  
\end{proof}

\begin{theorem}[Smith normal form]
	Let $A$ be a principal ideal domain, and $X \in M_{r, s}(A)$. Then there exist
	invertible matrices $U \in M_{r}(A)$ and $V \in M_{s}(A)$ such that 
	$U X V = \begin{pmatrix} \textup{diag}(\alpha_1, \dots, \alpha_t) & 0 \\ 0 & 0 \end{pmatrix}$ for some $0 \leq t
	\leq \min(r, s)$, where the $\alpha_i \neq 0$ and $\alpha_i | \alpha_{j}$ if $i \leq j$. 
\end{theorem}

\begin{proof}
See theorems 3.8 and 3.9 of \cite{Jacobson_1985}. See also the discussion at the start of \S6 in
\cite{Haah_2013}.
\end{proof}

\begin{theorem}[Structure theorem for finitely generated modules over a PID]
Every finitely generated module $M$ over a PID $A$ is isomorphic to a module of the form
\begin{equation}
A^k \oplus A / (f_1) \oplus \cdots \oplus A / (f_\ell) \nonumber \text{,}
\end{equation}
where $k$ is a nonnegative integer,  each $f_i \in A$ is nonzero and not a unit, and $f_i \mid f_{i+1}$.  In particular, a finitely generated module over a PID is free if and only if it is torsion-free.  \label{thm:PID-structure}
\end{theorem}

\begin{proof}
	See \S3.8 of \cite{Jacobson_1985}.
\end{proof}

\begin{defn}
Let $M$ be a finitely generated module over a PID $A$.  The {\bf rank} of $M$ is the nonnegative integer $k$ of Theorem~\ref{thm:PID-structure}. We write  $\operatorname{rank} M = k$.
\end{defn}

\subsection{Structure theorems}
\label{subsec:st}

We let $p$ be a prime and take an integer $k \geq 1$, and set $R = \Z_{p^k}[t^\pm]$ throughout the rest of this subsection.  We define $A = R / (p) \cong \Z_p [t^\pm]$ and denote the natural map $R \to A$ by $r \mapsto [r]$.  Note that $A$ is a PID.  

\begin{defn}
An $A${\bf-filtration} for a finitely generated $R$-module $M$ is a filtration
\begin{equation}
0 = C^0 \subset C^1 \subset \cdots \subset C^\ell = M \nonumber \text{,}
\end{equation}
such that each quotient $C_s \equiv C^s / C^{s-1}$ is an $A$-module.  (This means that $(p) C_s = 0$, so $C_s$ can naturally be viewed as a module over $A$.) An $A$-filtration is {\bf free} if each $C_s$ is a free $A$-module.  We say that $\ell$ is the length of the filtration.  The {\bf height} of $x \in M$, denoted $h(x)$, is the smallest $s \in \{0 ,\dots, \ell\}$ for which $x \in C^s$. 
\end{defn}

Note that we use $\subset$ and $\subseteq$ interchangeably, so it is allowed to have $C^s = C^{s+1}$ and thus $C_s = 0$; this is compatible with an $A$-filtration being free because $\{0\}$ is  free -- it has the empty basis.  A filtration is an $A$-filtration if and only if $p C^s \subset C^{s-1}$ for all $s$.   Because $R$ is Noetherian and $M$ is finitely generated, each $C^s$ is finitely generated, and thus also the quotients $C_s$ are finitely generated.
We denote the natural maps $C^s \to C_s$ by $x \mapsto [x]_s$.  For any $x \in C^s$ and $r \in R$, note that $[r x]_s = r [x]_s = [r] [x]_s$.

We will see below that every finitely generated $R$-module has an $A$-filtration. Although the length of an $A$-filtration can vary, there is an invariant of $M$ which can be extracted from an $A$-filtration.

\begin{defn}Given a finitely generated $R$-module $M$ with an $A$-filtration by $C^s$ ($s = 0,\dots, \ell$), let $r_s$ denote the rank of $C_s=C^s/C^{s-1}$ as an $A$-module. We call the sum $\sum_{s=1}^\ell r_s$ the \textbf{total rank} of the filtration $C^s$. \end{defn}

We will explain the following result in Appendix~\ref{app:ktheory}.
\begin{theorem}\label{thm:totalrankinvariance}
The total rank is an additive invariant of $M$, which we denote by $[M]$. More specifically,  all $A$-filtrations of a finitely generated $R$-module $M$ have the same total rank. 
\end{theorem}
\begin{remark}
The invariant $[M]$ is an additive invariant which belongs to a certain $K$-theory group, as will be explained in Appendix~\ref{app:ktheory}.
\end{remark}

Every finitely generated torsion-free $R$-module comes with a free $A$-filtration:

\begin{defn}
Given an $R$-module $M$, the $p${\bf-torsion filtration} of $M$ is
\begin{equation}
0 = M^0 \subset M^1 \subset \cdots \subset M^k = M \nonumber \text{,}
\end{equation}
where $M^s = \{ x \in M \mid p^s x = 0 \}$.  That is, $M^s$ is the submodule of elements with order at most $p^s$.  
\end{defn}

\begin{proposition}
The $p$-torsion filtration of a finitely generated $R$-module $M$ is an $A$-filtration.  It is a free $A$-filtration if $M$ is torsion-free.  \label{prop:pt-is-A}
\end{proposition}

\begin{proof}
First observe that $p M^s \subset M^{s-1}$, so $(p)$ annihilates $M_s$, which is thus naturally a module over $A$.  We assume $M$ is torsion-free and show that $M_s$ is torsion-free and thus free by Theorem~\ref{thm:PID-structure}.  Suppose $x \in M^s$ has $[x]_s$ torsion and nonzero, so there is some nonzero $a \in A$ with $a [x]_s = 0$.  Let $\sigma : A \to R$ be a section of the natural map $R \to A$, and note that $\sigma(a)$ is not a zero divisor in $R$ by Proposition~\ref{prop:R-zero-div}, because $[\sigma(a)] = a \neq 0$ so $\sigma(a)$ does not lie in $(p)$.  We have $[ \sigma(a) x ]_s = a [x]_s = 0$, so $\sigma(a) x \in M^{s-1}$ and thus $\sigma(a) (p^{s-1} x ) = p^{s-1} \sigma(a) x = 0$.  We also have $p^{s-1} x \neq 0$ because $x \notin M^{s-1}$.  Therefore $p^{s-1} x$ is a nonzero torsion element of $M$, a contradiction.  
\end{proof}

\begin{proposition}
Let $M$ be a finitely generated $R$-module.  Then $M$ has a free $A$-filtration if and only if $M$ is torsion-free.
\end{proposition}

\begin{proof}
By Proposition~\ref{prop:pt-is-A}, if $M$ is torsion-free then it has a free $A$-filtration, namely the $p$-torsion filtration.

Suppose that $M$ has a free $A$-filtration by $C^s \subset M$.  Let $x \in M$ be a nonzero torsion element of height $s = h(x)$, then there is a regular $r \in R$ with $r x = 0$.  We have $[r] \neq 0$ because as a regular element, $r \notin (p)$.  Then $[r x]_s = [r] [x]_s = 0$, implying $[x]_s = 0$ because $C_s$ is torsion-free.  Therefore $x \in C^{s - 1}$, a contradiction, so $x = 0$.
\end{proof}

\begin{proposition}
Let $M$ and $N$ be finitely generated torsion-free $R$-modules with $p$-torsion filtrations by $M^s \subset M$ and $N^s \subset N$, and let $f : M \to N$ be an $R$-linear map.  Then there are $R$-linear maps $f_s : M_s \to N_s$ defined by $f_s([x]_s) = [ f(x) ]_s$, and
\begin{enumerate}
\item $f$ is injective if and only if $f_s$ is injective for all $s$ .
\item $f_s$ surjective for all $s$ implies that $f$ is surjective.
\item $f$ is an isomorphism if and only if $f_s$ is an isomorphism for all $s$.
\end{enumerate}
\label{prop:filtration-maps}
\end{proposition}

\begin{proof}
First we have to show the maps $f_s$ are well-defined; it is enough to show $f(M^s) \subset N^s$.  If $x \in M^s$, then $p^s x = 0$, so $p^s f(x) = f(p^s x) = 0$, and $f(x) \in N^s$.  It is clear the maps $f_s$ are $R$-linear.

In what follows we define $f^s : M^s \to N^s$ by $f^s = f|_{M^s}$.  Note that $M^1 = M_1$, $N^1 = N_1$ and $f^1 = f_1$.  

For $s = 1,\dots, k-1$, we have the following commutative diagrams with exact rows:
\begin{equation}
\begin{tikzcd}
0 \arrow[r] & M^s \arrow[r] \arrow[d, "f^s"] & M^{s+1} \arrow[r] \arrow[d, "f^{s+1}"] & M_{s+1} \arrow[r] \arrow[d, "f_{s+1}"] & 0 \\
0 \arrow[r] & N^s \arrow[r] & N^{s+1} \arrow[r] & N_{s+1} \arrow[r] & 0 
\end{tikzcd}
\end{equation}
Suppose that $f_s$ is injective for all $s$.  Considering the above diagram for $s=1$ and noting $f^1 = f_1$, the five lemma implies that $f^2$ is injective.  Similarly, the $s=2$ diagram gives that $f^3$ is injective.  Continuing in this manner, we obtain $f^k = f$ is injective.  The same argument shows that $f_s$ surjective for all $s$ implies $f$ surjective, and that $f_s$ an isomorphism for all $s$ implies $f$ is an isomorphism.

Now assume that $f$ is injective.  Take $x \in M^s$ and suppose $f_s( [x]_s ) = 0$, then $[f(x)]_s = 0$, so $f(x) \in N^{s-1}$.  Therefore $0 = p^{s-1} f(x) = f( p^{s-1} x)$, so $p^{s-1} x = 0$ and hence $[x]_s = 0$.  We have shown $f_s$ is injective.

Finally assume $f$ is an isomorphism.  We have already shown that $f_s$ is injective; we need to show surjectivity.  Choose $y \in N^s$ and $x \in M$ with $y = f(x)$.  We have $p^s f(x) = f(p^s x) = 0$, so $p^s x = 0$, and thus $x \in M^s$.  Therefore $f_s([x]_s) = [y]_s$.
\end{proof}

\begin{defn}
Given a finitely generated $R$-module $M$ with a free $A$-filtration by $C^s$ ($s = 0,\dots, \ell$), a {\bf filtered basis} is a choice of elements $m^s_1, \dots, m^s_{r_s} \in C^s$ for each $s$.  Here $r_s$ is the rank of $C_s$ as a free $A$-module, and we require that $[m^s_1]_s, \dots, [m^s_{r_s}]_s$ is a basis for $C_s$.  If there is some $s$ for which $C_s = 0$, then $r_s = 0$ and the filtered basis contains no elements $m^s_i$.
\end{defn}

It is clear that a filtered basis always exists.  When working with filtered bases, we use implied summation notation for the lower index of $m^s_i$, but not for the upper index.  We will often write expressions like $\sum_{t = 1}^\ell r^t_i m^t_i$; if there is some $r_t = 0$, the convention is that we omit that value of $t$ from the sum over $t$.  The following proposition explains the name ``filtered basis'':

\begin{proposition}
Let $M$ be a finitely generated $R$-module with a free $A$-filtration by $C^s \subset M$, and let $\{ m^s_i \}$ be a filtered basis.  
Let $\sigma : A \to R$ be a set-theoretic section of the natural map $R \to A$ for which $\sigma(0) = 0$.  Then for any $x \in M$ there exist unique elements $a^t_1, \dots, a^t_{r_t} \in A$ such that $x = \sum_{t=1}^\ell \sigma(a^t_i) m^t_i$, where the sum on $i$ is implied, and where $a^t_i = 0$ for $t > h(x)$.   \label{prop:unique-expression}
\end{proposition}

\begin{proof}
Take $x \in M$.  If $x = 0$ then we choose all the $a^t_i = 0$.  Now assume $x \neq 0$.  Let $s = h(x)$; clearly $s \geq 1$ and $[x]_s \neq 0$.  If $s < \ell$, then we set $a^t_i = 0$ for $t > s$ (if any were nonzero, then it would imply $x \notin C^s$).  We choose $a^s_i$ by expanding $[x]_s = a^s_i [m^s_i]_s$, and then observe that
$[x - \sigma(a^s_i) m^s_i]_s = 0$, so $x - \sigma(a^s_i) m^s_i \in C^{s-1}$.  Repeat this procedure with $x' =  x - \sigma(a^s_i) m^s_i \in C^{s-1}$, and continue inductively to get the desired expression.

To check uniqueness, suppose that two choices $a^t_i$ and $\tilde{a}^t_i$ give the same $x \in M$.  Considering $[x]_\ell$, we have the expression $a^\ell_i [m^\ell_i]_\ell =  \tilde{a}^\ell_i [m^\ell_i]_\ell$, which implies $a^\ell_i = \tilde{a}^\ell_i$.  It follows that
\begin{equation}
x' = x - \sigma(a^\ell_i) m^\ell_i = \sum_{t =1}^{\ell-1} \sigma(a^t_i) m^t_i = \sum_{t =1}^{\ell-1} \sigma(\tilde{a}^t_i) m^t_i \nonumber \text{.}
\end{equation}
Considering $[x']_{\ell-1}$ we repeat the same argument and find $a^{\ell-1}_i = \tilde{a}^{\ell-1}_i$.  Clearly we can continue repeating this procedure to find $a^t_i = \tilde{a}^t_i$ for all $i$ and $t$.  
\end{proof}

\begin{theorem}\label{thm:f-presentations}
Let $R=\Z_{p^k}[t^\pm]$ for $k\geq 2$.
\begin{enumerate}[(i)]
\item Let $M$ be a finitely generated $R$ module with a free $A$-filtration by $C^s \subset M$ and a filtered basis $\{ m^s_i \}$, and let $\sigma : A \to R$ be a section of the natural map $R \to A$, with $\sigma(0) = 0$.  Then $M$ has a presentation
\begin{equation}
M \cong M(f) \equiv \frac{ R \{ x^s_i \} }{\langle p x^\ell_i - \sum_{t=1}^{\ell-1} f^{\ell t}_{i j} x^t_j , \dots,
p x^s_i - \sum_{t=1}^{s-1} f^{s t}_{i j} x^t_j ,
 \dots, p x^1_i \rangle}  \text{,}  \label{eqn:f-form}
\end{equation}
specified by matrices $f^{s t} \in M_{r_s, r_t}(R)$, defined for $s > t$ when both $r_s$ and $r_t$ are nonzero. The elements of each $f^{st}$ lie in the image of $\sigma$.  Here, $\{x^s_i \}$ is in bijective correspondence $x^s_i \leftrightarrow m^s_i$ with the filtered basis elements $\{ m^s_i\}$, and  $R \{ x^s_i \}$ denotes the free $R$-module with basis $\{ x^s_i \}$.  Letting $F = R\{x^s_i\}$, the isomorphism is induced by the surjective $R$-module map $\eta : F \to M$ defined by $\eta(x^s_i) = m^s_i$.  

\item If $M(f)$ is a finitely generated module of the form given in Equation~\ref{eqn:f-form} with $\ell \leq k$, then, denoting by $m^s_i$ the image of $x^s_i$ under the quotient map, $M(f)$ has a free $A$-filtration defined by $C(f)^s = \langle m^1, \dots, m^1_{r_1}, \dots, m^s_1, \dots, m^s_{r_s} \rangle \subset M(f)$, and $\{ m^s_i \}$ is a filtered basis.  
\end{enumerate}
\end{theorem}

\begin{proof}
(i) Let $G = \ker \eta$, then $F / G \cong M$.  Because $p m^s_i \in C^{s-1}$, by Proposition~\ref{prop:unique-expression}, we can express $p m^s_i = \sum_{t = 1}^{s-1} \sigma(a^{s t}_{i j} ) m^t_j$ for unique $a^{s t}_{i j} \in A$.  Let $G' \subset F$ be the submodule generated by elements of the form $p x^s_i - \sum_{t = 1}^{s-1} \sigma(a^{s t}_{i j} ) x^t_j$ for all $s, i$, so $M(f) = F/G'$. Clearly $G' \subset G$ and we claim that in fact $G' = G$, so $M(f) \cong M$.  

Consider an arbitrary element $x = \sum_s r^s_i x^s_i \in F$.  Note that for any $r \in R$, there exist unique $b \in A$ and $c \in R$ such that $r = \sigma(b) + p c$, where $b = [r]$.  Thus we have $r^\ell_i = \sigma(b^\ell_i) + p c^\ell_i$ for $b^\ell_i \in A$ and $c^\ell_i \in R$, and
\begin{eqnarray}
x &=& \big( \sigma(b^\ell_i) + p c^\ell_i \big) x^\ell_i + \sum_{s = 1}^{\ell-1} r^s_i x^s_i  \nonumber \\
&\equiv& \Big( \sigma(b^\ell_i) x^\ell_i  +  \sum_{s = 1}^{\ell-1} \tilde{r}^s_i x^s_i \Big) \operatorname{mod} G' \text{,} \nonumber
\end{eqnarray} 
for some $\tilde{r}^s_i$ defined for $1 \leq s \leq \ell-1$.  Repeating this procedure we find $x \equiv ( \sum_s \sigma(b^s_i) x^s_i  )\operatorname{mod} G'$ for some $b^s_i \in A$.  Since $\eta$ vanishes on $G'$, we thus have $\eta(x) = \sum_s \sigma(b^s_i) m^s_i$.  Then taking $x \in G$, we have $\eta(x) = 0$, and by Proposition~\ref{prop:unique-expression} we have $b^s_i = 0$, so $x \equiv 0 \operatorname{mod} G'$ and $x \in G'$.  We set $f^{s t}_{i j} = \sigma(a^{s t}_{i j})$ to obtain Equation~\ref{eqn:f-form}.

(ii)  Now we suppose $M(f)$ is a module of the form given in Equation~\ref{eqn:f-form}.   
Clearly we have a filtration by $C(f)^s$, and $C(f)_s = C(f)^{s} / C(f)^{s-1}$ is naturally an $A$-module since  $p m^s_i = \sum_{t =1}^{s-1} f^{s t}_{i j} m^t_j \in C(f)^{s-1}$, so $p C(f)^s \subset C(f)^{s-1}$.  We denote the natural maps $C(f)^s \to C(f)_s$ by $m \mapsto [m]_s$.

We let $F$ be the free module in the numerator of Equation~\ref{eqn:f-form}, and $G \subset F$ the submodule in the denominator.  We have submodules $F^s = R\{ x^s_1, \dots, x^s_{r_s}, \dots, x^1_1, \dots, x^1_{r_1} \} \subset F$ and 
$G^s = \langle p x^s_i - \sum_{t=1}^{s-1} f^{s t}_{i j} x^t_j ,
 \dots, p x^1_i \rangle$, which fit into the filtrations $0 = F^0 \subset F^1 \subset \cdots \subset F^\ell = F$ and
 $0 = G^0 \subset G^1 \subset \cdots \subset G^\ell = G$.  In addition we have $G^s \subset F^s$.
 
It will be useful below to show that $p^t x^t_i \in G^t$.  We proceed by induction.  For $t=1$, clearly $p x^1_i \in G^1$.  Now suppose the claim is true for $t \leq s$, we will show it holds for $t = s+1$.  We have
\begin{equation}
p^{s+1} x^{s+1}_i = p^s \Big( p x_i^{s+1} - \sum_{t < s+1} f^{s+1,t}_{i j} x^t_j \Big) + p^s \sum_{t < s+1} f^{s+1,t}_{i j} x^t_j \nonumber \text{.}
\end{equation}
The first term lies in $G^{s+1}$ and the second term lies in $G^s$ by assumption.  
 
 We have the following diagram with short exact columns:
\begin{equation}
\begin{tikzcd}
 & 0 \arrow[d] & 0 \arrow[d] & 0 \arrow[d] &  & 0 \arrow[d] \\
\cdots \arrow[r, hook] & G^{s-1} \arrow[r, hook] \arrow[d, hook] & G^s \arrow[r ,hook] \arrow[d, hook] & G^{s+1} \arrow[r, hook] \arrow[d, hook] & \cdots \arrow[r, hook] & G^k = G \arrow[d, hook]  \\
\cdots \arrow[r, hook] & F^{s-1} \arrow[r, hook] \arrow[d] & F^s \arrow[r ,hook] \arrow[d] & F^{s+1} \arrow[r, hook] \arrow[d] & \cdots \arrow[r, hook] & F^k = F  \arrow[d] \\
\cdots \arrow[r, "i_{s-2,s-1}"] & F^{s-1} / G^{s-1} \arrow[r, "i_{s-1,s}"] \arrow[d] & F^s / G^s \arrow[r, "i_{s,s+1}"] \arrow[d] & F^{s+1} / G^{s+1} \arrow[d] \arrow[r, "i_{s+1,s+2}"] & \cdots \arrow[r, "i_{k-1,k}"] & F / G = M(f)  \arrow[d] \\
 & 0 & 0 & 0  &  & 0
\end{tikzcd}  \label{eqn:Mf-cd}
\end{equation}
The unlabeled arrows are all inclusions and natural maps.  The maps $i_{s,s+1} : F^s / G^s \to F^{s+1} / G^{s+1}$ are given by $x + G^s \mapsto x + G^{s+1}$, which is well-defined because $G^s \subset G^{s+1}$.  It is straightforward to check that the squares are all commutative.  Moreover, $i_{s, s+1}$ is injective; to see this, take $x \in F^s$ and suppose $i_{s,s+1}(x + G^s) = x + G^{s+1} = 0$, which is equivalent to $x \in G^{s+1}$.  We will show that $x \in G^s$, so $x + G^s = 0$.  We have
\begin{equation}
x = \sum_{t \leq s} r^t_i x^t_i = \sum_{t \leq s+1} q^t_i (p x^t_i - \sum_{t' < t} f^{t t'}_{i j} x^{t'}_j ) \text{,} \nonumber
\end{equation}
for some $q^t_i \in R$.  The first expression holds because $x \in F^s$ while the second holds because $x \in G^{s+1}$.  Comparing the second two expressions, we observe that $x^{s+1}_i$ only appears in the first term of the second expression for $t = s+1$.  Therefore we have $p q^{s+1}_i = 0$, and thus $q^{s+1}_i = p^{k-1} \tilde{q}^{s+1}_i$ for some $\tilde{q}^{s+1}_i \in R$.  Therefore,
\begin{equation}
x = - \tilde{q}^{s+1}_i \sum_{t' < s+1} f^{s+1, t'}_{i j} (p^{k-1} x^{t'}_j ) + \sum_{t \leq s} q^t_i (p x^t_i - \sum_{t' < t} f^{t t'}_{i j} x^{t'}_j ).
 \nonumber
 \end{equation}
Examining the first term, we can write $p^{k-1} x^{t'}_j = p^{k-1 -t'} p^{t'} x^{t'}_j \in G^{t'} \subset G^s$, which makes sense because $t' \leq s \leq \ell - 1 \leq k-1$.  Therefore the first term lies in $G^s$, and $x \in G^s$ as claimed.

Viewing the arrows in the bottom row as inclusions, we have another filtration $0 = F^0 / G^0 \subset F^1 / G^1 \subset \cdots \subset F^k / G^k = F/G = M(f)$.  Moreover, if we start at $F^s$ and follow the arrows rightward to $F$, then downward to $M(f)$, the image is $C(f)^s$ as defined above.  Following the arrows first downward and then to the right instead, we conclude that $C(f)^s = F^s / G^s$.
Furthermore, the Snake Lemma applied to the diagram
\[\xymatrix{0\ar[r] &  G^{s-1} \ar[d] \ar[r] & F^{s-1}\ar[d] \ar[r] & C(f)^{s-1} \ar[d]\ar[r] & 0 \\
0\ar[r] &  G^{s} \ar[r] & F^{s} \ar[r] & C(f)^{s} \ar[r] & 0 
}\]
and the injectivity of the vertical maps imply that there is a short exact sequence on cokernels
\[\xymatrix{ 0\ar[r] &  G^{s}/G^{s-1} \ar[r] & F^{s}/F^{s-1} \ar[r] & C(f)_s \ar[r] & 0 
}.\]
This allows us to identify $C(f)_s$ with $R \{ x^s_1, \dots, x^s_{r_s} \}/( p x^s_1 , \dots, p x^s_{r_s} ) \cong A^{r_s} $, 
the free $A$-module with basis $[x^s_1]_s, \dots, [x^s_{r_s}]_s$.  
\end{proof}

Now we characterize the $p$-torsion filtration in terms of its $f$-matrices.  First we need a few preliminary results.

\begin{defn}
A matrix $\alpha \in M_{n,m}(R)$ is $A${\bf-left-nondegenerate} ($A$-lnd for short) if for all $x \in R^n$ with $[x] \neq 0$, we have $[x \alpha] \neq 0$.
 This is equivalent to the property that for all $x \in R^n$ with $x \notin (p) R^n$, we have $x \alpha \notin (p) R^m$.
 \end{defn}

\begin{proposition}
$A$-left-nondegeneracy is a basis-independent property of a matrix.  That is, if $\alpha \in M_{n,m}(R)$ is $A$-lnd and $u \in M_n(R)$ and $v \in M_m(R)$ are invertible, then $u \alpha v$ is $A$-lnd.
\end{proposition}

\begin{proof}
Note that $[u]$ and $[v]$ are invertible matrices with elements in $A$, because \emph{e.g.} $[u] [u^{-1} ] = [u u^{-1} ] = [ \mathbbm{1} ] = \mathbbm{1}$.  So $[u]^{-1} = [u^{-1}]$.   Suppose $x \in R^n$ with $[x] \neq 0$.  Then $[x u] = [x][u] \neq 0$.  Since $\alpha$ is $A$-lnd, $[x u \alpha] = [x u][\alpha] \neq 0$.  Finally we have $[x u \alpha v] = [x u \alpha][v] \neq 0$.
\end{proof}

\begin{lemma}
Let $\alpha \in M_n(A)$ be invertible, and let $\sigma : A \to R$ be a section of the natural map $R \to A$.  Then $\sigma(\alpha) \in M_n(R)$, \emph{i.e.} the matrix obtained by applying $\sigma$ to $\alpha$ element-wise, is invertible.  \label{lem:invertible-matrix}
\end{lemma}

\begin{proof}
Note that for $a, b \in A$, $[ \sigma(a) \sigma(b) ] = [ \sigma(a b) ]$ and also $[\sigma(a) + \sigma(b)] = [\sigma(a+b)]$.  Therefore we have $[ \sigma(\alpha) \sigma(\alpha^{-1} ) ] = [ \sigma( \alpha \alpha^{-1} ) ] = [ \sigma ( \mathbbm{1} ) ] = \mathbbm{1}$.  Therefore
$\sigma(\alpha) \sigma(\alpha^{-1}) = 1 - p \beta$ for some $\beta \in M_n(R)$.  We have $\sigma(\alpha)^{-1} = \sigma(\alpha^{-1})(1 + p \beta + p^2 \beta^2 + \cdots + p^{k-1} \beta^{k-1})$.  
\end{proof}

\begin{proposition}
A matrix $\alpha \in M_{n,m}(R)$ is $A$-lnd if and only if $n \leq m$ and the Smith normal form of $[\alpha] \in M_{n,m}(A)$ has no zero rows.  \label{prop:lnd-characterization}
\end{proposition}

\begin{proof}
Because $A$-left-nondegeneracy is a basis-independent property, we might as well change to a convenient basis.  We choose invertible matrices $u \in M_n(A)$ and $v \in M_m(A)$ so that $u [\alpha] v$ is in Smith normal form.  By Lemma~\ref{lem:invertible-matrix}, $\sigma(u)$ and $\sigma(v)$ are also invertible, so we change basis by letting $\alpha' = \sigma(u) \alpha \sigma(v)$, and then $[\alpha']$ is in Smith normal form.  

If $n \leq m$ and $[\alpha']$ has no zero rows, then for any $x \in R^n$ with $[x] \neq 0$, we have $[x \alpha'] = [x] [\alpha'] \neq 0$.  Therefore $\alpha'$ and hence $\alpha$ is $A$-lnd.

Now assume $\alpha$ is $A$-lnd.  Suppose $[\alpha']$ has a zero row, then there is some $[x] \neq 0$ with $[x \alpha']  = 0$, a contradiction.  Since $[\alpha']$ is in Smith normal form, the only way to avoid having a zero row is for $n \leq m$. 
\end{proof}

\begin{theorem}
Let $M$ be a nonzero finitely generated torsion-free $R$-module, and let $\{ m^s_i \}$ be a filtered basis for the $p$-torsion filtration.  Considering a presentation $M(f) \cong M$ of the form Equation~\ref{eqn:f-form} as in Theorem~\ref{thm:f-presentations}, then
\begin{enumerate}[(i)]
\item there exists $s_0 \in \{1,\ldots, k\}$ such that  $r_s\neq 0$ if and only if $1\leq s \leq s_0$, and
\item for $s=2, \ldots, s_0$, $r_{s} \leq r_{s-1}$ and the matrix $f^{s, s-1}$ is $A$-lnd.
\end{enumerate}
The index $s_0$ is the smallest non-negative integer such that $p^{s_0}M=0$.

Conversely,  if $M(f)$ is a nonzero module of the form given in Equation~\ref{eqn:f-form} with $\ell = k$, and which satisfies the conditions (i) and (ii) above, then the $A$-filtration of $M(f)$ by 
\[C(f)^s = \langle m^s_1, \dots, m^s_{r_s}, \dots, m^1_1, \dots, m^1_{r_1} \rangle \subset M(f)\] 
is the $p$-torsion filtration. 
\end{theorem}
\begin{proof}
If $C^s$ is the $p$-torsion filtration on a finitely generated, nonzero, torsion-free $R$-module $M$, if $C^{s} = C^{s-1}$ for some $s$, then $M=C^{k}=C^{k-1}$. Indeed, let $m\in C^k$, then  $p^{k-s}m\in C^s=C^{s-1}$. Therefore, $p^{s-1}p^{k-s}m=0$, and so, $p^{k-1}m =0$ and $m\in C^{k-1}$. So, $M$ is annihilated by $p^{k-1}$, and thus can be viewed as a $R/p^{k-1}=\Z_{p^{k-1}}[t^\pm]$. It follows inductively that there always exists $1\leq s_0\leq k$ such that $C^0 \subsetneq C^1 \subsetneq C^2 \subsetneq \cdots \subsetneq C^{k_0-1}\subsetneq C^{s_0} = C^{s_0+1} = \cdots = C^k$. Note that the index $s_0$ corresponds to the smallest non-negative integer such that $p^{s_0}M=0$. Since $r_s=0$ if and only if $C^s=C^{s-1}$, it follows that $r_s \neq 0$ for $ 1\leq s\leq s_0$ and  $r_s = 0$ for $s>s_0$.

The matrices  $f^{s,s-1}$ are thus defined for $2\leq s\leq s_0$. We have $p m^s_i = \sum_{t=1}^{s-1} f^{s t}_{i j} m^t_j$, which implies $[p m^s_i]_{s-1} = [f^{s,s-1}_{i j}] [m^{s-1}_j ]_{s-1} = a^{s,s-1}_{i j} [m^{s-1}_j]_{s-1}$.  If we make basis changes on the $[m^s_i]_s$ for $M_s$ and the $[m^{s-1}_i]_{s-1}$ for $M_{s-1}$, then we transform the matrix $a^{s,s-1}$ by $a^{s,s-1} \mapsto u a^{s, s-1} v$ for invertible matrices $u$ and $v$.  Moreover $\sigma(u)$ and $\sigma(v)$ are invertible by Lemma~\ref{lem:invertible-matrix}, so these basis changes lift to a basis change on the free module $F = R\{x^s_i\}$.  Therefore, by making a suitable basis change on $F$, we can put $a^{s,s-1}$ in Smith normal form.  The Smith normal form cannot have any zero rows, because a zero row would give $[p m^s_i]_{s-1} = 0$ for some $i$, implying $p m^s_i \in M^{s-2}$; this is a contradiction because $m^s_i$ has order $p^s$ and $p m^s_i$ thus has order $p^{s-1}$.  This implies $r_{s} \leq r_{s-1}$.  By Proposition~\ref{prop:lnd-characterization}, $f^{s,s-1}$ is $A$-lnd.

For the converse, first observe that $M(f)  = C(f)^{s_0}$, so $p^{s_0} x = 0$ for all $x \in M(f)$.  Therefore elements of $M$ have order at most $p^{s_0}$.  We need to show $C(f)^s = \{ m \in M(f) \mid p^s m = 0 \}$.  This clearly holds for $s \geq s_0$ by the preceding remarks.  

If $m \in C(f)^s$, then $p^s m = 0$, because $p C(f)^s \subset C(f)^{s-1}$ and thus $p^s C(f)^s = 0$.  We will show by induction that if $m \in M(f)$ has $p^s m = 0$, then $m \in C(f)^s$.  We already noted this holds for $s \geq s_0$.  For the inductive step, take $s < s_0$, and suppose the statement is true for all $t > s$. Consider $m \in M(f)$ with $p^s m = 0$.  Then $p^{s+1} m = 0$, so $m \in C(f)^{s+1}$.  We can thus write $m = \sum_{t \leq s+1} r^t_i m^t_i$.  Using $p^s C(f)^s = 0$, we have 
\begin{equation}
p^s m = p^s r^{s+1}_i m^{s+1}_i =  p^{s-1} r^{s+1}_i (p m^{s+1}_i) 
= p^{s-1} r^{s+1}_i \sum_{t < s+1} f^{s+1,t}_{i j} m^t_j = p^{s-1} r^{s+1}_i f^{s+1,s}_{i j} m^s_j \text{,} \nonumber
\end{equation}
where in the last equality we used $p^{s-1} C(f)^{s-1} = 0$.  Using the relations to get rid of powers of $p$ one at a time, we find the expression
\begin{equation}
0 = p^s m = r^{s+1}_i f^{s+1,s}_{i j_1} f^{s, s-1}_{j_1 j_2} \cdots f^{21}_{j_{s-1} j_s} m^1_{j_s} \in C(f)^1 \text{.}
\end{equation}
Since $C(f)^1$ is a free A-module with basis $[m^1_i]_1$, this implies $[r^{s+1}_i f^{s+1,s}_{i j_1} f^{s, s-1}_{j_1 j_2} \cdots f^{21}_{j_{s-1} j_s}] = 0$.  By $A$-left-nondegeneracy it follows that $[r^{s+1}_i] = 0$, using the fact that the product of $A$-lnd matrices is $A$-lnd.  Therefore we can use the relations to eliminate the $r^{s+1}_i m^{s+1}_i$ terms in $m$, and $m \in C(f)^s$.  
\end{proof}

\begin{remark}
It is instructive to present the module $S$ of the twisted 1-foliated fracton order of Example~\ref{ex:t1f} in the form of Equation~\ref{eqn:f-form}.  We have $n = 4$ and $R = \Z_4[t^\pm]$, so $p = 2$ and $k = 2$.  We recall that $S$ is generated by elements $e, m \in S$, with relations $2 e = 0$ and $2 m = (t+\bar{t})e$.  We set $x^1_1 = e$ and $x^2_1 = m$, and thus $r_1 = r_2 = 1$, with $f^{21} = t + \bar{t}$.  \label{rem:t1f}
\end{remark}

\subsection{Dual modules and filtrations}
\label{subsec:dual-st}

As above, we take $R = \Z_{p^k}[t^\pm]$ and $A = R/(p)$ throughout this subsection, with $p$ prime and $k \geq 1$.  We take the dual of an $R$-module $M$ to be $M^* = \Hom_R(M,R)$ with $R$-module structure given by $(r\varphi)(m) = \varphi(\bar{r}m) =\bar{r} \varphi(m)$.  We record the following standard general result:

\begin{proposition}
Let $\mathcal R$ be a commutative ring with involution and $M$ an $\mathcal R$-module.  Then $\Hom_{\cR}(M, \cR)$ is torsion-free. If $\cR$ is Noetherian and $M$ is finitely generated,  then ${\Hom_{\cR}}(M, \cR)$ is finitely generated.
\label{prop:noetherian-fg-dual}
\end{proposition}

\begin{defn}
Let $M$ be an $R$-module.  Given a length-$\ell$ filtration of $M$ by $C^s \subset M$, the {\bf dual filtration} of $M^*$ is
$0=D^0 \subset D^1 \subset \cdots \subset D^\ell = M^*$, where $D^s = \{ \lambda \in M^* \mid \lambda(C^{\ell-s}) = 0 \}$.  
\end{defn}

This is indeed a filtration because $C^{\ell-s} \subset C^{\ell - s + 1}$, so if $\lambda \in D^{s-1}$, $\lambda$ vanishes on $C^{\ell - s +1}$ by definition and thus also on $C^{\ell - s}$, so $\lambda \in D^s$ and $D^{s-1} \subset D^s$.

\begin{proposition}
Let $M$ be a finitely generated $R$-module with an $A$-filtration (not necessarily free) by $C^s \subset M$.  Then the dual filtration of $M^*$ is a free $A$-filtration.
\end{proposition}

\begin{proof}
By Proposition~\ref{prop:noetherian-fg-dual}, $M^*$ is finitely generated.  Take $\lambda \in D^s$.  Then for any $x \in C^{\ell-s+1}$, 
$p \lambda(x) = \lambda(p x) = 0$, because $p x \in C^{\ell - s}$.  Therefore $p \lambda \in D^{s-1}$ and $p D^s \subset D^{s-1}$.  This proves that the dual filtration is an $A$-filtration.

To see that the dual filtration is free, take $\lambda \in D^s$ with $[r] [\lambda]_s = 0$ for $r \in R$ with $[r] \neq 0$.  Therefore $r \lambda \in D^{s-1}$.  For all $x \in C^{\ell - s + 1}$, we have $\bar{r} \lambda(x) = 0$, and thus $\lambda(x) = 0$ since $\bar{r}$ is regular.  We thus have $\lambda \in D^{s-1}$ and $[\lambda]_s = 0$, so $D_s$ is torsion-free.
\end{proof}

For $s,t \in \{1 ,\dots, \ell\}$, let $D(t,s)$ be the set of descending sequences $(\sigma_1, \dots, \sigma_n)$ from $t$ to $s$, \emph{i.e.} where $\sigma_1 = t$, $\sigma_n =s$, and $\sigma_i > \sigma_{i+1}$.  Denote by $|\sigma| = n$ the number of elements in the sequence.  Note that $D(t,s) = \varnothing$ if $t < s$.  If $t \geq s$, then $|\sigma| \leq t - s + 1$.  Using matrix notation and suppressing $i, j$ indices, given $\sigma \in D(t,s)$ we define a matrix $f^\sigma \in M_{r_t, r_s}(R)$ by
\begin{equation}
f^\sigma = \left\{ \begin{array}{ll}
\mathbbm{1} \text{,} & |\sigma| = 1 \\
f^{\sigma_1 \sigma_2} f^{\sigma_2 \sigma_3} \cdots f^{\sigma_{n-1} \sigma_n} \text{,} & |\sigma| > 1 
\end{array}\right. .
\nonumber
\end{equation}

\begin{proposition}
Let $M$ be a finitely generated $R$-module with a free $A$-filtration of length $k$ by $C^s \subset M$.  We choose a filtered basis $\{ m^s_i \}_{1\leq s \leq k, 1 \leq i \leq r_s}$ and a presentation as in Theorem~\ref{thm:f-presentations} with matrices $f^{st}$.  Then, for each $s$, there are maps $\lambda^s_1, \dots, \lambda^s_{r_s} \in M^*$ defined by their values on filtered basis elements, namely
\begin{equation}
\lambda^s_i (m^t_j) = \left\{ \begin{array}{ll}
 0 \text{,} & t < s \\
 \sum_{\sigma \in D(t,s)}  p^{k - |\sigma|} ( f^{\sigma T})_{i j}  \text{,} & t \geq s 
 \end{array}\right.  \text{,} \label{eqn:lambdas}
\end{equation}
where ``$T$'' denotes the transpose.   \label{prop:lambdas}
\end{proposition}

\begin{proof}
We let $F$ and $G$ be as in the proof of Theorem~\ref{thm:f-presentations}.  There are linear maps $\lambda^s_i : F \to R$ defined on basis elements by setting $\lambda^s_i(x^t_j)$ to be the expression on the right-hand side of Equation~\ref{eqn:lambdas}. 
We will show that the $\lambda^s_i$ vanish on $G \subset F$, thus giving the desired maps $\lambda^s_i \in M^*$.

It is enough to show $\lambda^s_i$ vanishes on generators of $G$, \emph{i.e.} we need $\lambda^s_i(p x^t_j - \sum_{t' < t} f^{t t'}_{j j'} x^{t'}_{j'} ) = 0$.  This is clearly zero if $t  < s$, and if $t = s$ it becomes $p \lambda^s_i( x^s_j) = p^k \delta_{i j} = 0$.  Assuming $t > s$,
\begin{align*}
\lambda^s_i(p x^t_j - \sum_{t' < t} f^{t t'}_{j j'} x^{t'}_{j'} ) &=
p \sum_{\sigma \in D(t,s)} p^{k - |\sigma|} (f^{\sigma {T}})_{i j} 
- \sum_{t' < t} { f^{t t'}_{j j'} } \lambda^s_i (x^{t'}_{j'} )  \\
&= \sum_{\sigma \in D(t,s)} p^{k - (|\sigma|-1)} (f^{\sigma {T}})_{i j} 
- \sum_{s \leq t' < t} { f^{t t'}_{j j'} } \sum_{\sigma \in D(t', s)} p^{k - |\sigma|} (f^{\sigma{T}}_{i j'})  \\ 
&= \sum_{\sigma \in D(t,s)} p^{k - (|\sigma|-1)} (f^{\sigma {T}})_{i j} 
- \sum_{s \leq t' < t} \sum_{\sigma \in D(t', s)} p^{k - |\sigma|} (f^{t t'} f^\sigma)^{T}{i j} = 0 \text{.}  \qedhere
\end{align*}
\end{proof}

\begin{proposition}\label{prop:zetasinjective}
Let $M$ be a finitely generated $R$-module with an $A$-filtration by $C^s \subset M$, and let $D^s \subset M^*$ be the dual filtration. Then there are injective $R$-module maps $\zeta_s : D_{\ell - s +1} \to C_s^*$ defined by 
$\zeta_s([\lambda]_{\ell - s + 1} )([x]_s) = \lambda(x)$.
\end{proposition}

\begin{proof}
First recall that $D_{\ell - s  + 1} = D^{\ell - s + 1} / D^{\ell - s}$, where $D^{\ell - s + 1}$ consists of maps in $M^*$ that vanish on $C^{s - 1}$, while $D^{\ell - s}$ contains the maps vanishing on $C^s$.  We need to check that $\zeta_s$ is well-defined.  Given $[\lambda]_{\ell - s + 1} \in D_{\ell - s + 1}$, pick two representatives $\lambda, \lambda' \in D^{\ell - s + 1}$.  We have $\lambda' = \lambda + \tilde{\lambda}$ for some $\tilde{\lambda} \in D^{\ell - s}$.  Now given $[x]_s \in C_s$, pick representatives $x, x' \in C^s$, where $x' = x + \tilde{x}$, where $\tilde{x} \in C^{s-1}$.  Observe that $\lambda(x) = \lambda(x')$, because $\lambda(\tilde{x}) = 0$.  Now $\lambda'(x) = \lambda(x) + \tilde{\lambda}(x) = \lambda(x)$, because $\tilde{\lambda}$ vanishes on $C^s$.  Therefore $\zeta_s$ is well-defined.  It is clearly an $R$-module homomorphism.

Now suppose $\zeta_s([\lambda]_{\ell - s + 1} ) = 0$, then $\lambda(x) = 0$ for all $x \in C^s$.  But then $\lambda \in D^{\ell - s}$ and $[\lambda]_{\ell - s + 1} = 0$.  So $\zeta_s$ is injective.
\end{proof}

\begin{lemma}
Let $M$ be a finitely generated free $A$-module viewed as an $R$-module.  Then the dual $M^*$ is a free $A$-module.  Moreover, if $m_1, \dots, m_n$ is an $A$-basis for $M$, then $m^*_1, \dots,m^*_n  \in M^*$ defined by $m^*_i(m_j) = p^{k-1} \delta_{i j}$ is an $A$-basis for $M^*$.  \label{lem:Adual}
\end{lemma}

\begin{proof}
Let $ \overline{A}$ be the $A$-module whose underlying abelian group is $A$ and action is given by $a \cdot x =\overline{a} x$ for $a\in A$ and $x \in \overline{A}$. This is a free $A$-module isomorphic to $A$ where an isomorphism is given by the involution $a\mapsto \bar{a}$.
We have $M \cong A^n$ for some $n \in \N$, so $M^* \cong {\Hom_R}(A^n, R) \cong ( {\Hom_R}(A, R) )^n$. But, using that $A$ is $p$-torsion and $ p^{k-1}R\cong A$ as $R$-modules,
${\Hom_R}(A, R) = {\Hom_R}(A, p^{k-1}R)\cong {\Hom_R}(A, A)\cong {\Hom_A}(A, A) \cong \overline{A}$ where the last isomorphism is given by evaluation at $1\in A$. The claim about the dual basis is standard. 
\end{proof}

\begin{proposition}
Let $M$ be a finitely generated $R$-module with a free $A$-filtration by $C^s \subset M$, and a filtered basis $\{ m^s_i \}$.  Assume we have a set of maps $\{ \gamma^s_i \}_{1 \leq s \leq \ell, 1 \leq i \leq r_s} \subset M^*$ satisfying 
\begin{equation}
\gamma^{s}_i (m^t_j) = \left\{ \begin{array}{ll}
0 \text{,} & t < s \\
p^{k-1} \delta_{i j} \text{,} & t = s 
\end{array}\right. \nonumber
\end{equation}
Then $\{ \gamma^{\ell - s + 1}_i \}$ is a filtered basis for the dual filtration on $M^*$.    \label{prop:gammas}
\end{proposition}

\begin{proof}
First observe that $\gamma^s_i(C^{s-1}) = 0$, so $\gamma^s_i \in D^{\ell - s + 1}$ and therefore $\gamma_i^{\ell -s  + 1} \in D^s$.  
We have 
\begin{equation}
\zeta_s([ \gamma^s_i ]_{\ell - s + 1} )( [m^s_j]_s ) = \gamma^s_i(m^s_j) = p^{k-1} \delta_{i j} \text{,} \nonumber
\end{equation}
so by Lemma~\ref{lem:Adual} $\{ \zeta_s ( [\gamma^s_i]_{\ell -s + 1}) \}_{1 \leq i \leq r_s}$ is a basis for the free $A$-module $C_s^*$.  This implies that $\zeta_s$ is surjective. It is injective by Proposition~\ref{prop:zetasinjective} and thus an isomorphism.  Therefore $\{ [\gamma^s_i]_{\ell - s  +1} \}_{1 \leq i \leq r_s }$ is a basis for $D_{\ell -s + 1}$ and hence $\{ [\gamma^{\ell - s + 1}_i ]_s \}_{1 \leq i \leq r_{\ell -s +1}}$ is a basis for $D_s$.
\end{proof}

\begin{theorem}
Let $M$ be a finitely generated $R$-module with a free $A$-filtration of length $k$ by $C^s \subset M$.  Then the maps $\{ \lambda^{\ell - s + 1}_i \}$ of Proposition~\ref{prop:lambdas} give a filtered basis for the dual filtration on $M^*$.  Moreover, $D_{\ell - s + 1}$ and $C_s$ have the same rank as free $A$-modules.
\end{theorem}

\begin{proof}
The $\lambda^s_i$ satisfy the condition assumed on the maps $\gamma^s_i$ in Proposition~\ref{prop:gammas}, and we simply apply the proposition.
\end{proof}

\begin{corollary}\label{cor:totalrankdualfiltration}
Let $M$ be a finitely generated $R$-module with a free $A$-filtration of length $k$ by $C^s \subset M$.  Let $r_s$ be the rank of $C_s$. Let $r^*_s$ be the ranks of the $p$-torsion filtration of the dual $M^*$. Then
$\sum_{s} r_s = \sum_{s} r^*_s$. 
\end{corollary}
\begin{proof}
By the previous result, the total rank of the dual filtration on $M^*$ is $\sum_s r_s$. But total rank is independent of the filtration by Theorem~\ref{thm:totalrankinvariance}. Hence, the $p$-torsion filtration of $M^*$ has the same total rank.
\end{proof}

\section{Excitation-detector principle, physical realizability and perfectness}
\label{sec:main-proof}

In this section we prove Theorem~\ref{thm:main-intro}.  As discussed in the Introduction, this shows that the excitation-detector principle, defined as nondegeneracy of $\tilde{b} : S \times \widetilde{S} \to \Q / \Z$, is a necessary condition for a theory of excitations to be physically realizable.  The excitation-detector principle is also equivalent to perfectness.

We begin by restating Theorem~\ref{thm:main-intro} with two additional equivalent properties, which are mainly technical in nature.  We also first state the theorem as a result about p-theories $(S,b)$, because the additional data of the quadratic form $\theta$ plays no role.  The corresponding result for theories of excitations $(S,\theta)$ follows immediately and is stated as Corollary~\ref{cor:main-result}

\begin{theorem}[Detailed statement of Theorem~\ref{thm:main-intro}]
Given a p-modular p-theory $(S, b)$, the following are equivalent:
\begin{enumerate}
\item The bilinear form $\tilde{b} : S \times \widetilde{S} \to \Q / \Z$ is nondegenerate.
\item The bilinear form $\tilde{b}^{per} : S \times \widetilde{S}^{per} \to \Q / \Z$ is nondegenerate. 
\item The compactified p-theory $(S_N, b_N)$ is a modular p-theory of 2d abelian anyons for all $N \in \N$.
\item There exists some $N_0 \in \N$ such that for all $N \geq N_0$, the compactified p-theory $(S_N, b_N)$ is a modular p-theory of 2d abelian anyons.
\item The p-theory $(S,b)$ is perfect.
\end{enumerate}
\label{thm:main-detailed}
\end{theorem}

\begin{corollary}
Given a p-modular theory of excitations $(S,\theta)$, Theorem~\ref{thm:main-detailed} holds as written, replacing the compactified p-theories with compactified theories of 2d abelian anyons $(S_N, \theta_N)$ in conditions \#2 and \#4.  \label{cor:main-result}
\end{corollary}

\begin{proof}
Given a p-modular theory of excitations $(S,\theta)$, the associated p-theory $(S,b)$ is p-modular.  Each of the conditions only depends on $(S,b)$, so we simply apply Theorem~\ref{thm:main-detailed} to $(S,b)$.
\end{proof}

Now we proceed to prove Theorem~\ref{thm:main-detailed}.

\begin{lemma}
If a p-theory $(S,b)$ is p-modular, then $\tilde{b}$ and $\tilde{b}^{per}$ are left-nondegenerate.
\label{lem:pmod-tbrnd}
\end{lemma}

\begin{proof}
Take any nonzero $x \in S$, then there exists $y \in S$ with $b(x,y) \neq 0$.  By locality of $b$, for some sufficiently large $N \in \N$, we have $b(x,t^{m N} y) = 0$ for all nonzero $m \in \Z$.  Then for ${\tilde{y} =y\otimes ( \sum_{m \in \Z} t^{m N} )  \in \widetilde{S}^{per}}$, we have $\tilde{b}(x,\tilde{y}) =  \tilde{b}^{per}(x,\tilde{y}) = b(x,y) \neq 0$.
\end{proof}

\begin{lemma}
Let $S$ be a finitely generated torsion-free $R$-module where $R = \Z_{p^k}[t^\pm]$ with $p$ prime and $k \in \N$, and choose $N \in \N$.  Then the map $\rho_* : S^* \to {\Hom_{R_N} }(S_N, R_N)$ of Proposition~\ref{prop:rhostar} 
is surjective.  \label{lem:rhostar-surjective}
\end{lemma}

\begin{proof}
We choose a filtered basis $\{x^s_i \} \subset S$ for the $p$-torsion filtration, and a presentation as in Theorem~\ref{thm:f-presentations}.  Let $\lambda^s_i \in S^*$ be the maps defined in Proposition~\ref{prop:lambdas}.  We have generators $\chi^s_i = \rho(x^s_i)$ for $S_N$.  Let $\varphi^s_i = \rho_* ( \lambda^s_i)$, then 
\begin{equation}
\varphi^s_i ( \chi^t_j ) = \rho_*(\lambda^s_i) (\rho(x^t_j )) = \rho (\lambda^s_i ( x^t_j )) \text{.}  \nonumber
\end{equation}
Therefore,
\begin{equation}
\varphi^s_i ( \chi^t_j ) =  \left\{ \begin{array}{ll}
p^{k-1} \delta_{i j}  , & \qquad t = s \\
0 , & \qquad t < s
\end{array}\right. \nonumber
\end{equation}

Now take $\varphi \in {\Hom_{R_N} }(S_N, R_N)$.  We have $p \varphi(\chi^1_i) = \varphi(p \chi^1_i) = \varphi(0) = 0$, so $\varphi(\chi^1_i) = p^{k-1} \rho(r^1_i)$ for some $r^1_i \in R$.  We define $\varphi^{(1)} = \varphi - \sum_i \overline{\rho(r^1_i)} \varphi^1_i$, and clearly $\varphi^{(1)}(\chi^1_i) = 0$.  Next consider $p \varphi^{(1)}(\chi^2_i) = \varphi^{(1)}(p \chi^2_i) = 0$, because $p \chi^2_i = \rho (p x^2_i)$ is a linear combination of the $\chi^1_j$.  Therefore $\varphi^{(1)}(\chi^2_i) = p^{k-1} \rho(r^2_i)$ for some $r^2_i \in R$, and we can define $\varphi^{(2)} = \varphi^{(1)} - \sum_i \overline{\rho(r^2_i)} \varphi^2_i$, which vanishes on all the $\chi^2_j$ and $\chi^1_j$.  Continuing in this manner, we find $\varphi = \sum_s \overline{\rho(r^s_i)} \varphi^s_i$, and thus $\varphi = \rho_* ( \sum_s \overline{r^s_i} \lambda^s_i )$.  
\end{proof}

\begin{lemma}
Let $S$ be a finitely generated torsion-free $R$-module.  Choose $N \in \N$ and let $\gamma_N \in \tilde{R}$ be defined by $\gamma_N = \sum_{m \in \Z} t^{m N}$.  For any $x \in S$, we have $ x \otimes \gamma_N  = 0$ if any only if $\rho_N(x) = 0$.  \label{lem:gntx-nonzero}
\end{lemma}

\begin{proof}
First suppose $\rho_N(x) = 0$, then $x = (1-t^N) x'$ for some $x' \in S$.  Then
\begin{equation}
 x \otimes \gamma_N  =((1-t^N) x')  \otimes \gamma_N   = x'  \otimes ((1-t^N) \gamma_N)  = 0 \nonumber \text{.}
\end{equation}

Now suppose $\rho_N(x) \neq 0$.  We choose a filtered basis $\{x^s_i\}$ for $S$ and write $x = \sum_s r^s_i x^s_i$ for $r^s_i \in R$.
It is clear that $x\otimes \gamma_N $ only depends on the value of $\rho_N(x)$, so without loss of generality we can choose $x$ to be any convenient representative of $\rho_N(x) \in S_N$.  In particular we can assume each $r^s_i$ has $\deg_- r^s_i \geq 0$ and $\deg_+ r^s_i \leq N-1$.  If $r^k_i \neq 0$ for some $i$ and is divisible by $p$, then we can use the relations to eliminate $x^k_i$ in favor of terms with $s < k$.  Doing this will change the $r^s_i$ and may cause $\deg_- r^s_i < 0$ or  $\deg_+ r^s_i \geq N$ for some terms; if so, we can add elements of $I_N S$ to bring the degrees back within the desired range.  We can thus assume that every $r^k_i$ is either zero or has $[r^k_i] \neq 0$.  We can repeat this process considering the $r^{k-1}_i$ coefficients, and continue until every $r^s_i$ is either zero or has $[r^s_i] \neq 0$, with degrees ranging from 0 through $N-1$.

Let $s_0$ be the largest value of $s \leq k$ for which at least one of the $r^{s_0}_i \neq 0$.  We have
\begin{equation}
(  \lambda^{s_0}_i \otimes \operatorname{id} )(x \otimes \gamma_N  ) =   \lambda^{s_0}_i(x) \otimes \gamma_N = p^{k-1} r^{s_0}_i \gamma_N \text{.} \nonumber
\end{equation}
We have $p^{k-1} r^{s_0}_i \neq 0$ because $[r^{s_0}_i] \neq 0$.  Moreover because $p^{k-1} r^{s_0}_i$ has terms of degree zero through $N-1$, we have $p^{k-1} r^{s_0}_i \gamma_N \neq 0$, and it follows that $x \otimes \gamma_N   \neq 0$.
\end{proof}

\begin{lemma}
Let $\cR$ be a finite commutative ring and let $f \in \cR[t^\pm]$ be a nonzero polynomial with relative degree $\deg f  > 0$, whose highest and lowest degree coefficients are units in $\cR$.
Then there exists $N_0 \in \N$ such that for any $N \in \N$ divisible by $N_0$, there is a $\xi_N \in \cR[t^\pm]$ with $\xi_N f = 1 - t^N$.  The highest and lowest degree coefficients of $\xi_N$ are units in $\cR$.  Moreover we have $\deg \xi_N  < N$.
\label{lem:xi}
\end{lemma}

\begin{proof}
We consider the quotient ring $\cR' = \cR[t^\pm] / (f)$, which is a finite abelian group because the highest and lowest degree coefficients of $f$ are units, and is nonzero because $\deg f > 0$.  The map $\mu_t : \cR[t^\pm] \to \cR[t^\pm]$ defined by $\mu_t (r) = t r$ induces a group automorphism $\mu_t : \cR' \to \cR'$, because $t (f) = (f)$.  Because $\cR'$ is a finite group, its group of automorphisms is finite and so, there is some $N_0 \in \N$ for which $\mu^{N_0}_t = \mathbbm{1}$, where $\mathbbm{1} : \cR' \to \cR'$ is the identity automorphism.  For any $N \in \N$ divisible by $N_0$, we also have $\mu^N_t = \mathbbm{1}$.
Now we have $0 = (\mathbbm{1} - \mu^{N}_t) (1_{\cR'} ) = (\mathbbm{1} - \mu^{N}_t) ( 1 + (f) ) = (1 - t^{N}) + (f)$, which implies $1 - t^{N} \in (f)$, so we have $\xi_N \in \cR[t^\pm]$ with $\xi_N f = 1 - t^{N}$.  The highest and lowest degree coefficients of $\xi_N$ must be units in $\cR$ because the highest and lowest coefficients of $f$ are units.  The last statement holds because $\deg  f  + \deg \xi_N  = N$.  
\end{proof}

\begin{proof}[Proof of Theorem~\ref{thm:main-detailed}]
To prove $1 \implies 2$, we observe left-nondegeneracy of $\tilde{b}^{per}$ holds by Lemma~\ref{lem:pmod-tbrnd}. Right-nondegeneracy is trivial because $\tilde{b}^{per}$ is a restriction of $\tilde{b}$.  The implication $3 \implies 4$ is trivial.

To prove the remaining implications, we observe that each of the properties listed respects stacking, by Propositions~\ref{prop:pmod-perfect-respect-stacking}, \ref{prop:mod-respects-stacking} and~\ref{prop:ird-respects-stacking}.   Theorem~\ref{thm:prime-decomp} provides that $(S, b)$ is a stack of theories with prime-power fusion order $n = p^k$.  Therefore it is enough to focus on a theory with prime-power fusion order, which allows us to use the results of Sec.~\ref{sec:structure}.  For the remainder of the proof, we thus assume $(S, b)$ has prime-power fusion fusion order $p^k$, we let $R = \Z_{p^k}[t^\pm]$,  we view $S$ as an $R$-module, and set $S^* = {\Hom_R}(S,R)$.  Because $(S,b)$ is p-modular, Proposition~\ref{prop:nd-implies-tf} implies  $S$ is a torsion-free $R$-module.

$2 \implies 3$.  We will show the contrapositive, so we assume that some $(S_N, \Phi_N)$ is not modular.  Then there is $y \in S$ with $\rho_N(y) \neq 0$ and $b_N(\rho_N(x), \rho_N(y) ) = 0$ for all $x \in S$.  Therefore we have $\tilde{b}^{per}( x, y\otimes \gamma_N ) = 0$ for all $x \in S$ by Equation~\ref{eqn:tb-bN}.
Lemma~\ref{lem:gntx-nonzero} implies that $y\otimes \gamma_N  \neq 0$, so $\tilde{b}^{per}$ is degenerate.

$4 \implies 5$. We will show the contrapositive, so assume that $(S, \Phi)$ is not perfect.  We consider the $p$-torsion filtrations $S^t \subset S$ and $S^{* t} \subset S^*$.  We have free $A$-modules $S_t$ and $S^*_t$ with rank $r_t$ and $r^*_t$, respectively.  We define $\sigma : A \to R$ by $\sigma( \sum_{m \in \Z} a_m t^m ) = \sum_{m \in \Z} a_m t^m$, where $a_m \in \Z_p = \{ 0, \dots, p-1 \}$; clearly this is a section of the natural map $R \to A$ and $\sigma(0)=0$.

By Proposition~\ref{prop:filtration-maps}, we have injective maps $\Phi_t : S_t \to S^*_t$.  Moreover, because $\Phi$ is not an isomorphism, at least one of the $\Phi_t$ is not an isomorphism.  Let $t_0$ be the smallest value of $t$ for which $\Phi_t$ is not an isomorphism.  By injectivity of $\Phi_t$, we have $r_t \leq r^*_t$ for all $t$, because $A$ is a PID, and free submodules of a free module over a PID have rank less than or equal to the rank of the ambient module.  Therefore, using $\sum_t r_t = \sum_t r^*_t$ which holds by Corollary~\ref{cor:totalrankdualfiltration}
we have $r_t = r^*_t$ for all $t$.  

Let $\{ x^s_i \}$ and $x^{* s}_i$ be filtered bases for $S$ and $S^*$, respectively.  We assume these  bases have been chosen to put the matrix for $\Phi_{t_0}$ in Smith normal form.  To see that this can be done, we can put $\Phi_{t_0}$ in Smith normal form by changes of basis $[x^{t_0}_i]_{t_0} = u_{i j} [\tilde{x}^{t_0}_j]_{t_0}$ and $[x^{* t_0}_i]_{t_0} = v_{i j} [\tilde{x}^{* t_0}_j]_{t_0}$ for invertible $u, v \in M_{r_{t_0}}(A)$.  This lifts to a change of filtered bases $\{ x^s_i \}$ and $\{ x^{* s}_j\}$ because  $\sigma(u)$ and $\sigma(v)$ are invertible elements of $M_{t_0}(R)$ by Lemma~\ref{lem:invertible-matrix}.  Therefore $\Phi_{t_0}$ as a matrix is diagonal, and we write $\Phi_{t_0} ( [x^{t_0}_i ]_{t_0} ) = \Delta_i [ x^{* t_0}_i ]_{t_0}$, with no implied sum on the right-hand side.  We have $\Delta_i \in A$ nonzero for all $i = 1,\dots,r_{t_0}$ by injectivity of $\Phi_{t_0}$.  Moreover, at least one of the $\Delta_i$ is not a unit, since otherwise $\Phi_{t_0}$ would be an isomorphism.  If necessary we reorder basis elements so that $\Delta_1$ is not a unit.  We thus have $\deg_+ \Delta_1 - \deg_- \Delta_1  = \deg_+ \sigma(\Delta_1) - \deg_- \sigma(\Delta_1) > 0$.  Moreover, all the nonzero coefficients of $\sigma(\Delta_1)$ are units in $\Z_{p^k}$.    By Lemma~\ref{lem:xi}, there exists $N_0 \in \N$ such that for any $N \in \N$ divisible by $N_0$, there exists $\xi_N \in R$ whose highest and lowest nonzero coefficients are units in $\Z_{p^k}$ such that $\xi_N \sigma(\Delta_1) = 1 - t^{N}$, and  $\deg_+ \xi_N - \deg_- \xi_N < N$.

We have
\begin{equation}
\Phi_N ( \rho_N( \xi_N x^{t_0}_1 ) ) = \rho_{N *} ( \Phi( \xi_N x^{t_0}_1 ) ) = \rho_{N *} ( \xi_N \Delta_1 x^{* t_0}_1 ) =
\rho_N (\xi_N \Delta_1) \rho_{N *} ( x^{* t_0}_1 ) = 0     \nonumber \text{.}
\end{equation}
We claim $\rho_N ( \xi_N x^{t_0}_1 ) \neq 0$, implying that $\Phi_N$ is not injective and thus $(S_N, \Phi_N)$ is not modular.  To verify the claim, let $\varphi^s_i \in {\Hom_{R_N}}(S_N, R_N)$ be as in the proof of Lemma~\ref{lem:rhostar-surjective}.  Then $\varphi^{t_0}_1 ( \rho_N ( \xi_N x^{t_0}_1 )) = p^{k-1} \rho_N(\xi_N)$.  This is nonzero because $\deg_+ \xi_N - \deg_- \xi_N < N$, and because the highest and lowest degree coefficients of $\xi_N$ are units and thus are not annihilated by $p^{k-1}$.  Therefore $\rho_N ( \xi_N x^{t_0}_1 ) \neq 0$ as claimed.

$5 \implies 1$.  By Lemma~\ref{lem:pmod-tbrnd}, $\tilde{b}$ is left-nondegenerate; we need to show right-nondegeneracy.  By Proposition~\ref{prop:Theta-Bt-nd}, this is equivalent to right-nondegeneracy of $\tilde{B}$. 

Choose a filtered basis $\{ x^s_i \}$ for $S$.  Any $\tilde{y} \in \widetilde{S}$ can be written $\tilde{y} = \sum_s x^s_i \otimes \tilde{r}^s_i $ for some $\tilde{r}^s_i \in \tilde{R}$ (write $\tilde{y} = \sum_i y_i \otimes  \tilde{r}_i $ and expand $y_i$ in terms of filtered basis elements).  We can choose the $\tilde{r}^s_i$ so that each $\tilde{r}^s_i$ is either zero, or $p^{k-1} \tilde{r}^s_i \neq 0$.  To see this, suppose some $\tilde{r}^k_i$ is nonzero and has $p^{k-1} \tilde{r}^k_i = 0$.  Then $\tilde{r}^k_i = p \tilde{r}^{' k}_i$ for some $\tilde{r}^{' k}_i \in \tilde{R}$, and we can use the relations to eliminate the $x^k_i$ term in the expansion of $\tilde{y}$ in favor of terms with $s < k$.  We can repeat this process, starting with $s = k$ terms and moving to smaller values of $s$, until the expansion of $\tilde{y}$ satisfies the desired property.

Now assume that $\tilde{B}(x,\tilde{y}) = 0$ for all $x \in S$.  We will show that $\tilde{y} = 0$.  Observe that $\tilde{B}(x,\cdot) \in \Hom_R(\widetilde{S}, \tilde{R})$.  Given $\lambda \in S^*$, there exists $x \in S$ such that $\lambda = B(x, \cdot)$.   
Define $\tilde{\lambda} \in \Hom_R(\widetilde{S}, \tilde{R})$ by $\tilde{\lambda} =  {\lambda} \otimes \operatorname{id}_{\tilde{R}}$.  For any $y \in S$ and $\tilde{r} \in \tilde{R}$, we have $\tilde{\lambda}(y\otimes \tilde{r}) = B(x,y) \otimes \tilde{r}  = \tilde{B}( x, y\otimes \tilde{r})$.  Therefore $\tilde{\lambda} = \tilde{B}(x, \cdot)$.  By our assumption $\tilde{\lambda}(\tilde{y}) = 0$.  

We consider $\widetilde{\lambda^k_i}(\tilde{y}) = p^{k-1} \tilde{r}^k_i = 0$.  Therefore we have $\tilde{r}^k_i = 0$.  Repeating this argument for $\widetilde{\lambda^{k-1}_i}$ and moving to lower values of $s$, we find $\tilde{y} = 0$ as desired.

This completes the proof.  While it is logically superfluous, we also include a direct proof of $5 \implies 3$, which does not rely on considering infinitely supported planons. By Corollary~\ref{cor:finite}, it is enough to show $\Phi_N$ is surjective.  By the discussion of Section~\ref{sec:compact},
$\Phi_N \circ \rho_N = \rho_{N*} \circ \Phi$.  By assumption $\Phi$ is surjective, and Lemma~\ref{lem:rhostar-surjective} gives $\rho_{N *}$ surjective.  Therefore $\Phi_N \circ \rho_N$ is surjective, so $\Phi_N$ is surjective.
\end{proof}

%% file: decoupled.tex
%!TEX root = 1planar-v2.tex
\section{Planon-only fracton orders of prime fusion order}
\label{sec:decoupled}

In this section, as an application of the algebraic theory of planon-only fracton orders developed in this paper, we show that planon-only fracton orders of prime fusion order always consist of decoupled layers of 2d topological orders.  In proving this result, we consider perfect p-theories, because  Theorem~\ref{thm:main-detailed} implies that perfectness is a necessary condition for physical realizability.  We first introduce formally what we mean by decoupled layers in Section~\ref{sec:dlintro} and prove that planon-only fracton orders of prime fusion order are decoupled layers in Section~\ref{sec:primedecoupled}.

\subsection{Decoupled 2d layers}\label{sec:dlintro}

To begin, we need to give a formal definition of what it means for a p-theory to consist of decoupled 2d layers.  These theories of excitations and p-theories are precisely those induced from 2d theories of abelian anyons. For a reference with more general results on induction of quadratic and sesquilinear forms see, for example, \cite[Ch. 1, Ch. 7]{Knus_1991}.

\begin{defn}
Let $\mathcal{A}$ be a finite abelian group. Consider the $\Ztpm$-module $\mathcal{A}[t^\pm] =  \mathcal{A} \otimes \Ztpm$.  We write the elements of $\mathcal{A}[t^\pm]$ as Laurent polynomials with coefficients in $\mathcal{A}$. 
\begin{enumerate}[(a)]
\item If  $(\mathcal{A}, \theta)$ is a theory of abelian anyons as in Definition~\ref{defn:theoryofabeliananyons}, let $\theta^{\pm} \colon \mathcal{A}[t^\pm] \to\Q/\Z$ be defined by
 \begin{equation}
\label{eq:thetapm} \theta^{\pm}\Big( \sum_{n\in \Z} x_n t^n \Big) =\sum_{n\in \Z} \theta(x_n) 
\end{equation}
 for $x_n\in \mathcal{A}$.  
 \item If  $(\mathcal{A}, b)$ is a p-theory of abelian anyons, let $b^{\pm} \colon \mathcal{A}[t^\pm] \times \mathcal{A}[t^\pm] \to \Q/\Z$ be defined by
  \begin{equation}
\label{eq:bpm}
b^{\pm}\Big(\sum_{n\in \Z} x_n t^n, \sum_{m\in \Z}y_mt^m\Big) =\sum_{n \in \Z}b(x_n,y_n)
\end{equation}
 for $x_n,y_m\in \mathcal{A}$. 
 \end{enumerate}
\end{defn}

\begin{proposition}
Let $\mathcal{A}$ be a finite abelian group. 
\begin{enumerate}[(a)]
\item If  $(\mathcal{A}, \theta)$ is a theory of abelian anyons, then  $(\mathcal{A}[t^\pm], \theta^{\pm})$ is a theory of excitations. 
\item If  $(\mathcal{A}, b)$ is a p-theory of abelian anyons, then $(\mathcal{A}[t^\pm], b^{\pm})$ is a p-theory. 
 \end{enumerate}
 Furthermore, if $(\mathcal{A}, \theta)$ is a theory of abelian anyons and $(\mathcal{A}, b)$ is the associated p-theory of abelian anyons, then the  p-theory corresponding to the theory of excitations $(\mathcal{A}[t^\pm], \theta^{\pm})$ is  $(\mathcal{A}[t^\pm], b^{\pm})$.
 \end{proposition}

 \begin{proof}
 In each of the expressions Equations~\ref{eq:thetapm} and \ref{eq:bpm}, only finitely many of the terms are non-zero, so these are well-defined as functions. That $\theta^\pm$ is a quadratic form and $b^\pm$ a sesquilinear form follows immediately from the same facts for $\theta$ and $b$.   Translation invariance of both $b^\pm$ and $\theta^\pm$ and locality of $b^\pm$ are clear from the defining formula. To show locality of $\theta^\pm$, we first establish the last claim.
 Let $(\mathcal{A}, \theta)$ be a theory of abelian anyons and $(\mathcal{A}, b)$ be the associated p-theory of abelian anyons. Let  $(\mathcal{A}[t^\pm], c)$ be the p-theory associated to $(\mathcal{A}[t^\pm], \theta^{\pm})$. Then
 \begin{align*} 
 c\left(\sum_n x_nt^n, \sum_n y_nt^n\right) &= \theta^{\pm}\left(\sum_n x_n+y_n\right)-\theta^{\pm}\left(\sum_n x_n\right)-\theta^{\pm}\left(\sum_n y_n\right) \\
 &= \sum_n \theta(x_n+y_n)-\theta(x_n)-\theta(y_n) \\
 &= \sum_{n}b(x_n,y_n) = b^\pm\left(\sum_n x_nt^n, \sum_n y_nt^n\right)  
 \end{align*}
 So, $c= b^{\pm}$ as claimed. Locality of $\theta^\pm$ now follows from locality of $b^\pm$.
 \end{proof}
 
 \begin{proposition}
 Let $(\mathcal{A},b)$ be a modular p-theory of abelian anyons. Then $(\mathcal{A}[t^\pm], b^{\pm})$ is a perfect p-theory. 
Consequently,  if $(\mathcal{A}, \theta)$ is a modular theory of abelian anyons, then  $(\mathcal{A}[t^\pm], \theta^{\pm})$ is perfect theory of excitations.
 \end{proposition}
 \begin{proof}
 The second claim follows from the fact that $(\mathcal{A}[t^\pm], b^{\pm})$ is the p-theory associated to  $(\mathcal{A}[t^\pm], \theta^{\pm})$ and perfectness of the latter is defined as perfectness of  $(\mathcal{A}[t^\pm], b^{\pm})$.
 
We need to prove that $\Phi^\pm \colon \mathcal{A}[t^\pm] \to (\mathcal{A}[t^\pm])^*$ given by
\[ \Phi^\pm(x)(y) = \sum_{m}b^\pm(t^mx, y) t^m =  \sum_{m}\sum_{m'} b(x_{m'-m},y_{m'} ) t^m , \]
where $x=\sum_m x_mt^m$ and $y=\sum_m y_mt^m$, is an isomorphism. But $\Phi(x_m)(y_m) = b(x_m,y_m)$ gives an isomorphism $\Phi \colon \mathcal{A}\to \Hom_{\Z}(\mathcal{A}, \Q/\Z)$,
\[  (\mathcal{A}[t^\pm])^* = \Hom_\Ztpm(\mathcal{A}\otimes \Ztpm,\Q/\Z\otimes \Ztpm) \cong \Hom_\Z(\mathcal{A},\Q/\Z\otimes \Ztpm)  \cong \Hom_\Z(\mathcal{A},\Q/\Z)\otimes \Ztpm  , \] 
and $\Phi^\pm = \Phi\otimes \id_{ \Ztpm } \colon \mathcal{A}\otimes \Ztpm \to \Hom_\Z(\mathcal{A},\Q/\Z)\otimes \Ztpm$.
\end{proof}
 
We are now ready to define what it means to be decoupled layers.
 \begin{defn}
 Let $S$ be a finitely generated torsion free $R$-module for $R=\Z_n[t^\pm]$, with $n$ the fusion order of $S$.
 \begin{enumerate}[(a)]
 \item A theory of excitations $(S,\theta)$ is \textbf{decoupled layers} if it is isomorphic to $(\mathcal{A}[t^\pm], \theta^{\pm})$ for some theory of abelian anyons  $(\mathcal{A}, \theta)$. 
\item A p-theory $(S,b)$ is \textbf{decoupled layers} if it is isomorphic to $(\mathcal{A}[t^\pm], b^{\pm})$ for some p-theory of abelian anyons $(\mathcal{A}, b)$.
 \end{enumerate}
\end{defn}

We next give an equivalent condition for a theory to be decoupled layers in terms of certain good choices of generators for the fusion module $S$.
We say that a Laurent polynomial $f = \sum_{m \in \Z} c_m t^m$ is {\bf constant} if all its coefficients are zero except possibly $c_0$. We say that a matrix with Laurent polynomial entries is constant if all of its entries are.
\begin{proposition}
A p-theory $(S,b)$ is decoupled layers if and only if there exist generators $e_1, \dots, e_\ell$ for $S$ such that the following two properties are satisfied:
	\begin{enumerate}
		\item The generators obey intra-layer relations.  Precisely, the associated first syzygy module, \\
		$I = \{ v \in \Ztpm^\ell \mid
			\sum_i v_i e_i = 0 \}$,
			is generated by constant vectors.  Here, $v_i$ denotes the $i$th component of the vector $v$.
		\item 
		Mutual statistics decouples across layers: $b(e_i, t^k e_j) = 0$ for all $k \neq 0$ and all generators $e_i, e_j$.
		Equivalently, $B(e_i, e_j)$ is constant for all $i, j$. 
	\end{enumerate}
\label{prop:ptheory-dcl}
\end{proposition}
\begin{proof}
It is clear from the definition that a p-theory which is decoupled layers satisfies the two properties; let $a_1, \dots, a_\ell$ be a set of generators for $\mathcal{A}$, and take $e_i = a_i \otimes 1 \in \mathcal{A}[t^\pm]$.  

For the converse, we take $\mathcal{A}$ to be the subgroup (\emph{i.e.}, $\Z$-submodule) of $S$ generated by $e_1,\ldots, e_\ell$.  We have $\mathcal{A}$ finite because each $e_i$ is of finite order.  To avoid notational overload, let $\mathcal{C}$ be a group isomorphic to $\mathcal{A}$ and $\gamma : \mathcal{C} \to \mathcal{A}$ an isomorphism.  The function $\Gamma : \mathcal{C}[t^\pm] \to S$ given by $\Gamma( \sum_{m \in \Z} c_m t^m ) = \sum_{m \in \Z} t^m \gamma(c_m)$ is clearly  $\Ztpm$-linear and surjective.  Suppose $\Gamma( \sum_{m \in \Z} c_m t^m) = \sum_{m \in \Z} t^m \gamma(c_m) = 0$.  For each $m \in \Z$ we have $\gamma(c_m) = \sum_i z_{m i} e_i$ for coefficients $z_{m i} \in \Z$, and thus $\sum_i ( \sum_{m \in \Z} z_{m i} t^m ) e_i = 0$.  The assumption on the syzygy module implies  $\sum_i z_{m i} e_i = 0$ for all $m \in \Z$, and therefore $\gamma(c_m) = 0$ and $c_m = 0$.  We have shown $\Gamma$ is injective and thus an isomorphism.

Now we define $b_c : \mathcal{C} \times \mathcal{C} \to \Q / \Z$ by $b_c(c, c') = b(\gamma(c), \gamma(c'))$ for $c, c' \in \mathcal{C}$.   Considering the p-theory $(\mathcal{C}[t^\pm], b_c^\pm)$, it is straightforward to use the assumption on $(S,b)$ that mutual statistics decouples across layers to show $b_c^\pm( y, y') = b(\Gamma(y), \Gamma(y'))$ for all $y, y' \in \mathcal{C}[t^\pm]$.  This shows that $(S,b) \cong (\mathcal{C}[t^\pm], b^\pm_c)$.
\end{proof}

\begin{proposition}
A theory of excitations $(S,\theta)$ is decoupled layers if and only if there exist generators $e_1, \dots, e_\ell$ for $S$ such that property \#1 of Proposition~\ref{prop:ptheory-dcl} is satisfied, and such that topological spins factorize across layers.  Precisely, $\theta(e_i + t^k e_j) = \theta(e_i) + \theta(t^k e_j)$ for all $k \neq 0$ and all generators $e_i, e_j$.
\end{proposition}
\begin{proof}
If a theory is decoupled layers, then it is again clear from the definition that the desired properties are satisfied, so we only need to prove the converse.  The proof parallels that of Proposition~\ref{prop:ptheory-dcl}; the construction of the isomorphism $\Gamma : \mathcal{C}[t^\pm] \to S$ is identical.  We define $\theta_c : \mathcal{C} \to \Q / \Z$ by $\theta_c(c) = \theta(\gamma(c))$ for $c \in \mathcal{C}$, and it is clear that $\theta_c^\pm(y) = \theta(\Gamma(y))$ for all $y \in \mathcal{C}[t^\pm]$, so $(S,\theta) \cong (S, \theta^\pm_c)$.
\end{proof}

The following result, whose proof is trivial, shows that if our goal is to prove a theory of excitations is decoupled layers, it is enough to consider its p-theory.  

\begin{proposition}
Let $(S,\theta)$ be a theory of excitations and $(S,b)$ the associated p-theory.  Then $(S,b)$ is decoupled layers if and only if $(S, \theta)$ is decoupled layers.
\end{proposition}

\subsection{p-theories of prime fusion order are decoupled}\label{sec:primedecoupled}
Now we specialize to the case of prime fusion order.  Throughout this section, we view $S$ as an $R$-module where $R = \Z_p[t^\pm]$, with $b$ taking values in $\Z_p$ and $B$ taking values in $R$.  It will be important that $R$ is a PID.  We use the same symbol $B$ to denote both the $R$-valued bilinear form of a p-theory and the matrix $B_{i j} = B(e_i, e_j)$ of the same bilinear form in a generating set $e_1, \dots, e_\ell$. We say a square matrix $X$ with entries in $R$ is {\bf Hermitian} if $X = X^\dagger$, where $X^\dagger$ is the transpose and element-wise involution of $X$. A polynomial $f \in R$ is called Hermitian if $f = \overline{f}$.

The main result of this section is the following, which we will prove over the remainder of the section:

\begin{theorem}
	\label{theorem:decoupled-layers}
	Every perfect p-theory of prime fusion order is decoupled layers. 
\end{theorem}

\begin{lemma}
If $(S,B)$ is a p-modular p-theory of prime fusion order $p$, then $S$ is a free $R$-module, and property \#1 of Proposition~\ref{prop:ptheory-dcl} is satisfied in any basis.  If $(S,B)$ is perfect, then in any basis for $S$, $B_{ij} = B(e_i, e_j) \in M_\ell(R)$ is an invertible Hermitian matrix.  \label{lem:easy-part}
\end{lemma}

\begin{proof}
Because $B$ is nondegenerate, $S$ is torsion-free and thus free.  Let $e_1, \dots, e_\ell$ be a basis, then consider vectors $v \in \Ztpm^\ell$ with $\sum_i v_i e_i = 0$.  This holds precisely when $v_i \in (p) \subset \Ztpm$, so the first syzygy module $I$ is generated by the vectors $(p, 0, 0, \dots)$, $(0, p, 0, \dots)$, and so on, and property \#1 holds.

For the associated map $\Phi : S \to S^*$, we have $\Phi(e_i) = \sum_j B_{i j} e^*_j$, where $e^*_1, \dots, e^*_r \in S^*$ is the dual basis.  Because $\Phi$ is an isomorphism, the matrix $B$ is invertible.  Hermiticity of the matrix $B$ follows immediately from the fact that the bilinear form $B$ is Hermitian.
\end{proof}

\begin{defn}
Two matrices $A, B \in M_\ell(R)$ are {\bf congruent} if there exists an invertible $U \in M_\ell(R)$ such that $A = U^\dagger B U$.
\end{defn}

We will argue that any invertible Hermitian matrix $B \in M_\ell(R)$ is
congruent to a constant matrix. Thus, there is always a choice of basis such that property \#2 of Proposition~\ref{prop:ptheory-dcl} is also satisfied. 
The key to showing this is a generalization of Djokovi\'{c}'s dominant diagonal theorem \cite{Djokovic_1976} to the Laurent
polynomial case.

Given $B \in M_\ell(R)$, there is an associated p-theory $(R^\ell, B)$ where $B(v,w) =v^\dagger Bw$.

\begin{defn}
	Let $(S, B)$ be a p-theory.  A nonzero element $x \in S$ is called {\bf isotropic} if $B(x,x) = 0$.  The p-theory $(S,B)$ is called {\bf anisotropic} if $S$ contains no isotropic elements.  A matrix $B \in M_\ell(R)$ is anisotropic if  $v^\dagger B v=\sum_{i, j = 1}^\ell \overline{v_i} v_j B_{i j} = 0$ implies $v = 0$ for all $v \in R^\ell$.
\end{defn}

It is clear that all the diagonal entries of an anisotropic matrix are nonzero, and that any matrix congruent to an anisotropic matrix is anisotropic.

\begin{lemma}
	\label{lemma:cyclic-submodule-generator}
	Let $R$ be a PID and $S$ a free, finitely generated $R$-module. Let $x \in S$ be any element. Then
	there exists a basis $e_1, \dots, e_\ell$ for $S$ such that $x = r e_1$ for some $r \in R$. 
\end{lemma}

\begin{proof}
	Choosing a basis, $x$ can be expressed as a column vector. Putting this $\ell \times 1$ matrix into
	Smith normal form,  we obtain the lemma. 
\end{proof}

\begin{lemma}
	\label{lemma:metabolic-anisotropic-decomposition}
	Let $B \in M_\ell(R)$ be  invertible and Hermitian. Then there exists an invertible matrix
	$U$ such that 
	\begin{align}
		U^\dagger B U 
		= \begin{pmatrix} B(a_1) &  & &  \\ & \cdots &  & \\ & & B(a_l) & & \\  & & & B'\end{pmatrix},
	\end{align}
	where each $a_i$ is constant in $R$, we have defined $B(a) = \begin{pmatrix} 0 & 1 \\ 1 & a
	\end{pmatrix}$ for $a \in R$, and $B'$ is anisotropic. 
\end{lemma}

See also Ref.~\onlinecite{Knus_1991} \S7 Theorem 1.5.1. 

\begin{proof}
	If $B$ is already anisotropic, we are done.
	If not, by Lemma \ref{lemma:cyclic-submodule-generator},
	let $e_1, \dots, e_r$ be a basis for $S$ where $e_1$ is isotropic. Then we have
	\begin{align}
		B = \begin{pmatrix}
			0 & V
			\\
			V^\dagger & C
		\end{pmatrix},
	\end{align}
	where $V$ is a row vector in $R^{\ell-1}$. Let $U'$ be a matrix that puts $V$ into Smith normal
	form, $V U' = \begin{pmatrix} v_0 & 0 &  \dots  & 0 \end{pmatrix}$.
	Then we have
	\begin{align}
		B' =
		\begin{pmatrix}
			1 & \\ & U'^\dagger
		\end{pmatrix}
		B
		\begin{pmatrix}
			1 & \\ & U'
		\end{pmatrix}
		&=
		\begin{pmatrix}
			0 & v_0 & 0
			\\
			\overline{v_0} & a & W
			\\
			0 &  W^\dagger & \tilde B
		\end{pmatrix} \text{,}
	\end{align}
	for $W$ a row vector in $R^{\ell - 2}$.
	Since $v_0$ divides the determinant of  $B'$, invertibility requires
	that $v_0$ is a unit. By another congruence, we can assume $v_{0} = 1$.
	Next, we can eliminate $W$ by the congruence
	\begin{align}
		B'' =
		\begin{pmatrix}
			1 & &
			\\
			& 1 &
			\\
		- W^\dagger & & 1
		\end{pmatrix}
		B'
		\begin{pmatrix}
			1 & & - W
			\\
			& 1 &
			\\
		  &  & 1
		\end{pmatrix}
		=
		\begin{pmatrix}
			& 1 & \\
			1 & a &
			\\
			&&& \tilde B
		\end{pmatrix}
	\end{align}
	Now, we can express $a = a_0 + a_+ + a_-$, where $a_0$ is constant, $a_+$ contains terms of strictly positive degree
	and $a_- = \overline{a_{+}}$ by Hermiticity.
	By a further congruence, we have 
	\begin{align}
		B''' =
		\begin{pmatrix}
			1 & \\ - a_- & 1
			\\
			&&& 1
		\end{pmatrix}
		B''
		\begin{pmatrix}
			1 & - a_+
			\\
			& 1
			\\
			&&&1
		\end{pmatrix}
		&=
		\begin{pmatrix}
			& 1
			\\
			1 & a_0
			\\
			&&& \tilde B
		\end{pmatrix}.
	\end{align}
	Continuing inductively, the proof is complete.
\end{proof}

\begin{corollary}\label{cor:rank2plusanisotropic}
A perfect p-theory of prime fusion order is isomorphic to a stack of rank 2 p-theories of the form $(R^2, B(a))$ for $a\in \Z_p$, and a perfect p-theory $(S',B')$ where $B'$ is anisotropic.  Each $(R^2, B(a))$ summand is decoupled layers.  
\end{corollary}

It remains to treat the case where $(S, B)$ is anisotropic.  The following simple result on degrees of Laurent polynomials will be useful; we omit the straightforward proof.
\begin{lemma}
If a Laurent polynomial $f \in R$ is a product of nonzero elements $f = f_1 \cdots f_m$ for $f_i \in R$, then $\deg_+ f = \sum_i \deg_+ f_i$ and $\deg_- f = \sum_i \deg_- f_i$.  \label{lem:degpm}
\end{lemma}

We recall Definition~\ref{defn:degrees}; for a Laurent polynomial $f \in R$, the total degree is
the pair $\delta f = (\deg_- f, \deg_+ f)$, and the total degree of the zero polynomial is defined to be
 $\delta 0 = (\infty, -\infty)$. Now we equip the set of total degrees with some partial orderings. 

	\begin{defn}
			For $f,g \in R$ we say 
			\[\begin{cases} \delta g <_+ \delta f \\ \delta g <_- \delta f \\ \delta g <_\pm \delta
				f \\ \delta g
				\leq \delta f \end{cases} \quad \quad
			\text{provided} \quad \quad
			\begin{cases}
				\text{$\deg_+ g < \deg_+ f$ and $\deg_- g \geq \deg_- f$}
				\\
				\text{$\deg_+ g \leq \deg_+ f$ and $\deg_- g > \deg_- f$}
				\\
				\text{$\deg_+ g < \deg_+ f$ and $\deg_- g > \deg_- f$}
				\\
				\text{$\deg_+ g \leq \deg_+ f$ and $\deg_- g \geq \deg_- f$}
			\end{cases}.\]
			  \label{defn:partial-orders}
		\end{defn}

It is straightforward to check that $<_+$, $<_-$ and $<_{\pm}$ are strict partial orders, that $\leq$ is the non-strict partial order associated with $<_\pm$, and that $\delta g <_{\pm} \delta f$ if and only if $\delta g <_+ \delta f$ and $\delta g <_- \delta f$.  Moreover, $\delta g <_+ \delta f$ or $\delta g <_- \delta f$ implies $\delta g \leq \delta f$.  Another useful fact is that $\delta f <_+ \delta g$ if and only if $\delta \overline{f} <_- \delta \overline{g}$.  It is also clear that $<_\pm$ gives a total order on the set of total degrees of Hermitian polynomials, because we have $\deg_+ f = - \deg_- f$ for nonzero such polynomials.  We adapt a definition of \cite{Djokovic_1976}:

\begin{defn}
	A matrix $B \in M_\ell(R)$ has a {\bf dominant diagonal} if its diagonal entries are nonzero,
	$\delta B_{ij} <_+ \delta B_{ii}$ for all $j > i$,
	and $\delta B_{ij} <_- \delta B_{ii}$ for all $j < i$.
\end{defn}

The condition of nonzero diagonal entries is superfluous except when $\ell = 1$, in which case the other two conditions are vacuous.

\begin{lemma}
	\label{lemma:dd-det}
	Suppose $B \in M_\ell(R)$ has a dominant diagonal. Then $\delta \det B
	= \delta \prod_i B_{ii}$. 
\end{lemma}

\begin{proof}
    Recall that the determinant of a matrix can be given as the signed sum over
    permutations $\sigma \in S_\ell$ of products of the form $B_\sigma := \prod_{i=1}^n B_{i\sigma(i)}$.  Since $B$ has a dominant diagonal, for any $\sigma \in S_\ell$, for all $i$ we have $\delta B_{i \sigma(i)} \leq \delta B_{i i}$, and thus $\deg_+ B_{i \sigma(i)} \leq \deg_+ B_{i i}$ and $\deg_- B_{i \sigma(i)} \geq \deg_- B_{i i}$.  Assume $\sigma$ is not the identity permutation, so then there is an $i_1$ such that $i_1 < \sigma(i_1)$ and another $i_2$ such that $i_2 > \sigma(i_2)$.  We have $\delta B_{i_1 \sigma(i_1)} <_+ \delta B_{i_1 i_1}$ and thus $\deg_+ B_{i_1 \sigma(i_1)} < \deg_+ B_{i_1 i_1}$.  Similarly, $\delta B_{i_2 \sigma(i_2)} <_- \delta B_{i_2 i_2}$ and thus $\deg_- B_{i_2 \sigma(i_2)} > \deg_- B_{i_2 i_2}$.  Using Lemma~\ref{lem:degpm}, this implies  $\delta B_\sigma <_{\pm} \delta B_{\text{id}}$.  The claim	follows, because if $\delta g <_\pm \delta f$, then $\delta (f + g) = \delta f$. 
\end{proof}

Hermitian matrices with dominant diagonals allow for a sort of division with remainder:

\begin{lemma}
	Suppose $B \in M_\ell(R)$ is Hermitian and has a dominant diagonal. Then for any vector $v \in R^\ell$ there exist $q, w \in R^\ell$ such that $v = B q + w$ and $\delta w_i <_+ \delta B_{ii}$ for all $1 \leq i \leq \ell$.   \label{lem:dd-div-with-rem}
\end{lemma}

See Lemma~2 of Ref.~\onlinecite{Djokovic_1976}.

\begin{proof}
Let $e_1, \dots, e_\ell$ be the standard basis for $R^\ell$, so that, in components, $(e_i)_j = \delta_{ij}$, and $B(e_i) = \sum_{j=1}^\ell B_{ji} e_j$.  Write $d_i = \deg_+ B_{ii}$. By Hermiticity, $\deg_- B_{ii} = - d_i$.  Define 
\begin{align}
V = \{ x \in R^\ell \mid \delta x_i <_+ \delta B_{ii}  \text{ for all } 1 \leq i \leq \ell \} \nonumber \text{,}
\end{align}
and, for any nonnegative $c \in \Z$,
\begin{align}
V(c) = \{ x \in R^\ell \mid \deg_+ x_i < d_i + c \text{ and } \deg_- x_i \geq -d_i - c  \text{ for all } 1 \leq i \leq \ell \} \nonumber \text{.}
\end{align}
These are  $\Z_p$-submodules of $R^\ell$, and $V = V(0)$.  We will show that $\im B + V = R^\ell$ as $\Z_p$-modules, where in this expression we view $B$ as a $\Z_p$-linear map $B: R^\ell \to R^\ell$.  Because every $x \in R^\ell$ is contained in some $V(c)$ it suffices to show that $V(c) \subset \im B + V$ for all $c \geq 0$. 

We proceed by induction.  Obviously we have $V(0) \subset \im B + V$.  We assume that $V(c) \subset \im B + V$, which immediately implies $\im B + V(c) \subset \im B + V$.  We will show that  $V(c+1) \subset \im B + V(c) \subset \im B + V$.  It is enough to show that $t^{d_i + c} e_i \in \im B + V(c)$ and $t^{-d_i - c - 1} e_i \in \im B + V(c)$ for all $i$, because other $\Z_p$-basis elements for $V(c+1)$ of the form $t^m e_i$ lie in $V(c)$.  

We show $t^{d_i + c} e_i \in \im B + V(c)$ using an inductive argument.  
We consider the expression
\begin{align}
B(t^c e_\ell) = t^c B_{\ell \ell} e_\ell + \sum_{i < \ell} t^c B_{i \ell} e_i \text{.} \nonumber
\end{align}
The left-hand side lies in $\im B$, while the second term on the right-hand side lies in $V(c)$ because $\deg_+ t^c B_{i \ell} < d_i + c$ and $\deg_- t^c B_{i \ell} \geq -d_i + c \geq -d_i-c$.  Therefore $t^c B_{\ell \ell} e_\ell \in \im B + V(c)$.   
We can write $t^c B_{\ell \ell}e_\ell = \sum_{m=-d_\ell + c}^{d_\ell + c-1}a_mt^me_\ell + a_{d_\ell + c}t^{d_\ell + c}e_\ell$ for $ a_{d_\ell + c}\neq 0$ (since $d_\ell = \deg_+ B_{\ell\ell}$). The first term is in $V(c)$ as shown above, and so $t^{d_\ell + c}e_\ell \in \im B + V(c)$. 
Now we suppose that for all $i > i_0$, $t^{d_i + c} e_i \in \im B + V(c)$.  We consider
\begin{align}
B(t^c e_{i_0} ) = t^c B_{i_0 i_0} e_{i_0} + \sum_{i < i_0} t^c B_{i i_0} e_i + \sum_{i > i_0} t^c B_{i i_0} e_i \text{.} \label{eqn:ni1}
\end{align}
The second term on the right-hand side lies in $V(c)$ because $\deg_+ t^c B_{i i_0} < d_i + c$ and $\deg_- t^c B_{i i_0} \geq -d_i + c \geq -d_i - c$.  For the third term, the dominant diagonal property gives $\deg_+ t^c B_{i i_0} \leq d_i + c$ and $\deg_- t^c B_{i i_0} > -d_i + c \geq -d_i - c$.  By the inductive assumption, this term lies in $\im B + V(c)$.  Therefore Equation~\ref{eqn:ni1} implies
$t^c B_{i_0 i_0} e_{i_0} \in \im B + V(c)$, which in turn gives $t^{d_{i_0} + c} e_{i_0} \in \im B + V(c)$.

A very similar argument shows that $t^{-d_i - c - 1} e_i \in \im B + V(c)$ for all $i$.  The argument, which we provide below, is essentially the same except that it starts at $i=1$ (instead of at $i = \ell$) and proceeds inductively to larger values of $i$ (instead of moving to smaller values).  To proceed, we consider the expression
\begin{align}
B(t^{-c-1} e_1) = t^{-c-1} B_{1 1} e_1 + \sum_{i > 1} t^{-c-1} B_{i 1} e_i \text{.} \nonumber
\end{align}
The left-hand side lies in $\im B$, while the second term on the right-hand side lies in $V(c)$ because $\deg_+ t^{-c-1} B_{i 1} \leq d_i -c -1 < d_i + c$ and $\deg_- t^{-c-1} B_{i 1} > -d_i - c - 1$, so $\deg_- t^{-c-1} B_{i 1} \geq -d_i-c$.  Therefore $t^{-c-1} B_{1 1} e_1 \in \im B + V(c)$. As before, this implies $t^{-d_1 - c -1} e_1 \in \im B + V(c)$.  Now we suppose that for all $i < i_0$, $t^{-d_i - c - 1} e_i \in \im B + V(c)$.  We consider
\begin{align}
B(t^{-c-1} e_{i_0} ) = t^{-c-1} B_{i_0 i_0} e_{i_0} + \sum_{i < i_0} t^{-c-1} B_{i i_0} e_i + \sum_{i > i_0} t^{-c-1} B_{i i_0} e_i \text{.} \label{eqn:ni2}
\end{align}
The third term on the right-hand side lies in $V(c)$ because $\deg_+ t^{-c-1} B_{i i_0} \leq d_i - c - 1 < d_i + c$ and
$\deg_- t^{-c-1} B_{i i_0} > - d_i - c - 1$ so $\deg_- t^{-c-1} B_{i i_0} \geq -d_i - c$.  For the second term we have $\deg_+ t^{-c-1} B_{i i_0} < d_i - c - 1 < d_i + c$ and $\deg_- t^{-c-1} B_{i i_0} \geq -d_i - c - 1$.  By the inductive assumption, this term lies in $\im B + V(c)$.  Therefore Equation~\ref{eqn:ni2} implies $t^{-c-1} B_{i_0 i_0} e_{i_0} \in \im B + V(c)$, which in turn gives $t^{-d_{i_0} - c - 1} e_{i_0} \in \im B + V(c)$.
\end{proof}

\begin{theorem}[Dominant Diagonal]
	\label{theorem:dominant-diagonal}
	An anisotropic Hermitian matrix over $R$ is congruent to a matrix with a dominant diagonal. 
\end{theorem}

Our proof follows the strategy of Djokovi\'{c} \cite{Djokovic_1976}; upon introducing the partial orders of Definition~\ref{defn:partial-orders} on total degree to replace the obvious total order on degree of ordinary polynomials, the generalization to Laurent polynomials is straightforward.
However, even in the ordinary polynomial case, we were not able to verify the final assertion in the proof of Ref.~\onlinecite{Djokovic_1976} for ourselves -- namely, that $(u, s+1) < (b,s)$ in the notation used there. For this reason, we give a proof that slightly differs from that of Djokovi\'{c}. Specifically, we have made a small modification to the definition of the property $P(s)$ (see below).

\begin{proof}[Proof of Theorem~\ref{theorem:dominant-diagonal}]
	We say that an anisotropic Hermitian matrix $A$ satisfies property $P(s)$ (with $1 \leq s \leq \ell$) if $\delta A_{11} \leq \dots \leq
	\delta A_{ss}$ and the $s \times s$ submatrix in the upper left corner has a dominant diagonal.  Note that every anisotropic Hermitian matrix satisfies property $P(1)$.
	
	Let $B \in M_\ell(R)$ be anisotropic and Hermitian. If $\ell = 1$, the theorem is trivial, so assume $\ell
	\geq 2$. 
	Let $\mathscr S$ be the set of pairs $(C, s)$ where $C \in M_\ell(R)$ is congruent to $B$ and satisfies
	property $P(s)$. We have $(B,1) \in \mathscr S$, so $\mathscr S$ is nonempty.  
	
	It is straightforward to check that the following relations on $\mathscr S$ are strict partial orders:
	\begin{itemize}
		\item $(C, s) <_1 (C', s')$ if $s > s'$ and $\delta C_{tt} = \delta C'_{tt}$ for all $1 \leq t \leq
			s'$.
		\item $(C, s) <_2 (C', s')$ if, for some $t \leq \min(s, s')$, $\delta C_{tt} <_{\pm} \delta
			C'_{tt}$ and $\delta C_{uu} = \delta C_{uu}'$ for all $1 \leq u < t$. 
		\item $(C, s) < (C', s')$ if $(C, s) <_1 (C', s')$ or $(C, s) <_2 (C', s')$. 
	\end{itemize}

For any $(C, s) \in \mathscr S$ with $s < \ell$, we can produce another $(C_1, s_1) \in
	\mathscr S$ with $(C_{1}, s_1) < (C, s)$. We partition $C$ in block form as 
	\begin{align}
		C = \begin{pmatrix}
			D & E 
			\\
			E^\dagger & F
		\end{pmatrix},
	\end{align}
	where $D \in M_{s}(R)$ has a dominant diagonal. By Lemma~\ref{lem:dd-div-with-rem} applied to the columns of $E$, there are $s \times (\ell-s)$ matrices $Q$ and $E'$ such that $E = D Q + E'$, where $\delta E'_{ij} <_+	\delta D_{ii}$ (and thus also $\delta \overline{E'_{i j}} <_- \delta D_{i i}$) for all $1 \leq i \leq s$ and $1 \leq j \leq \ell - s$.	
	We consider the congruent matrix 
	\begin{align}
	 C' = \begin{pmatrix} 1 & 0 \\ - Q^\dagger & 1 \end{pmatrix} C \begin{pmatrix} 1 & - Q \\ 0 & 1\end{pmatrix}
	 = \begin{pmatrix}
			D & E'
			\\
			E'^\dagger & F'
		\end{pmatrix} \text{.}
	\end{align}

	Now there are
	two cases. First, if $\delta F'_{11} \geq \delta D_{ss}$, we have $(C', s+1) \in \mathscr S$ because, for $1 \leq i \leq s$,
	\begin{align}
	\delta C'_{i (s+1)} 	= \delta E'_{i 1} <_+ \delta D_{i i} = \delta C'_{ii} \nonumber \text{,}
	\end{align}
	and
	\begin{align}
	\delta C'_{(s+1) i} = \delta (E')^\dagger_{1 i} = \delta \overline{E'_{i 1}} <_-  \delta D_{i i} \leq \delta F'_{11} = \delta C'_{(s+1)(s+1)} \text{.} \nonumber
	\end{align}
	We have $(C', s+1) <_1 (C, s)$, and we take $(C_1, s_1) = (C',s+1)$.  Second, if $\delta F'_{11} <_\pm \delta D_{ss}$, 
	there is a smallest $t \in \{1, \dots, s\}$ such that $\delta F'_{11} <_{\pm} \delta
	D_{tt}$.  By the congruence that swaps the $(s+1)$ and $t$ rows and columns of $C'$, we
	obtain a matrix $C''$ such that $(C'', t) \in \mathscr S$ and $(C'', t) <_2 (C, s)$, and we take
	$(C_1, s_1) = (C'', t)$.

	Choosing $(C_0, s_0) = (B, 1)$ and iterating the above process, we obtain a  decreasing sequence 
	\begin{align}
	(C_0, s_0) > (C_{1}, s_1) > \cdots > (C_m, s_m) > \cdots \text{.} \nonumber
	\end{align}  
	The sequence continues until it reaches an element $(C_m, s_m)$ with $s_m = \ell$.  If this occurs, we are done, because $C_m$ has a dominant diagonal and is congruent to $B$.  If not, then $s_m < \ell$ for all $m \in \N$ and we have produced an infinite decreasing sequence; we complete the proof by arguing that this cannot occur.
	
	Suppose we have an infinite
	decreasing sequence $(C_0, s_0) > (C_1, s_1) > \dots$. 
	Let $t$ be an integer such that $s_m = t$ for infinitely many $m$, and consider the 
	refined sequence $(C_{m_1}, s_{m_1}) > (C_{m_2}, s_{m_2}) > \dots$ such that $s_{m_i} =
	t$ for all $i$; this is equivalently a decreasing sequence
	$(C_{m_1}, t) >_2 (C_{m_2}, t) >_2 \dots$.
	In the refined sequence, at each step, the total degree $\delta (C_{m_i})_{11}$ either strictly decreases
	(\emph{i.e.} $\delta (C_{m_i})_{11} >_\pm \delta (C_{m_{i+1}})_{11}$) or remains the same. At each step where
	$\delta (C_{m_i})_{11}$ strictly decreases, the relative degree $\deg (C_{m_i})_{11}$ also strictly
	decreases, but this can occur only a finite number of times in the sequence because relative
	degree is always nonnegative.  Therefore by moving sufficiently far along the sequence, we obtain
	a new infinite decreasing sequence where $\delta (C_{m_i})_{11}$ is unchanged at each step. It follows
	that $\delta (C_{m_i})_{22}$ must either decrease or remain the same at each step, and it can only
	decrease a finite number of times.  Continuing in this manner we can move sufficiently far along
	the sequence that the total degree of every diagonal entry is unchanged at each step, a
	contradiction. 
	\end{proof}

\begin{corollary}
	\label{cor:inv-herm-an-mat}
	An invertible anisotropic Hermitian matrix with entries in $R$ is congruent to a diagonal matrix with only constant entries. 
\end{corollary}

\begin{proof}
	By Theorem~\ref{theorem:dominant-diagonal}, such a matrix is congruent to one with a dominant diagonal. 
	By Lemma \ref{lemma:dd-det}, an invertible Hermitian matrix with a dominant diagonal must have
	Hermitian units on the diagonal; therefore the diagonal entries are all constant.  It follows immediately from the definition of dominant diagonal that the off-diagonal entries are zero and thus constant.
\end{proof}

\begin{proof}[Proof of Theorem \ref{theorem:decoupled-layers}]
Suppose $(S,B)$ is a perfect p-theory of prime fusion order. By Lemma~\ref{lem:easy-part}, $S$ is free and property \#1 of Proposition~\ref{prop:ptheory-dcl} is satisfied in any basis, and moreover $B_{ij} =B(e_i,e_j)$ is an invertible Hermitian matrix in $M_\ell(R)$. By Corollary~\ref{cor:rank2plusanisotropic}, $(S,B)$ is a stack of p-theories that are decoupled layers and a perfect p-theory $(S',B')$ such that $B'_{ij}=B'(e_i,e_j)$ is an anisotropic invertible Hermitian matrix. By Corollary \ref{cor:inv-herm-an-mat}, $B'$ is congruent to a matrix with constant entries and so $(S',B')$ is decoupled layers. Therefore, $(S,B)$ is a stack of p-theories that are decoupled layers, hence is itself decoupled layers.
\end{proof}

%% file: discussion.tex
%!TEX root = 1planar-v2.tex
\section{Discussion}
\label{sec:discussion}

We conclude with a discussion of some questions raised by the results of this paper.  We found that the excitation-detector principle is equivalent to a necessary condition for physical realizability of a planon-only fracton order -- namely, perfectness of a theory of excitations -- which is stronger than the principle of remote detectability.  We conjecture that this condition is also sufficient for physical realizability; that is, every perfect finite-order planon-only theory of excitations can be realized by a suitable quantum lattice model.  A promising strategy to show this is to generalize the result that every modular 2d theory of abelian anyons can be realized by a multi-component ${\rm U}(1)$ Chern-Simons theory \cite{Wall_1963, Wall_1972, Nikulin_1980, Wang_2020} to the setting of infinite-component Chern-Simons theories, which are appropriate to describe planon-only fracton orders.  Work in this direction is in progress \cite{Dachuan_inprogress}.

Another important direction for future work will be to obtain more examples of planon-only fracton orders, and to characterize them in terms of simple and easily comparable invariants.  Indeed, the characterization of planon-only fracton orders in terms of such invariants is quite poorly understood, and we expect that the algebraic theory developed here, together with the structure theorems proved in Section~\ref{sec:structure}, will provide a useful starting point for future progress.  Part of this work will be to study the Witt classification of planon-only fracton orders; physically this is a coarser classification than phase-equivalence, where two states are considered equivalent if they can be joined by a gapped interface.  Moreover, in this paper we solved the problem of classifying planon-only fracton orders of prime fusion order $p$, by showing that all such fracton orders are decoupled layers.\footnote{This indeed gives a classification because theories of abelian anyons have been classified \cite{Galindo_2016}.}
It may be tractable to go further and classify planon-only fracton orders of fusion order $p^k$ for small values of $k$, starting with $k = 2$.

Beyond these specific future directions, promising areas for near-term study include planon-only fracton orders with infinite-order excitations and irrational braiding statistics \cite{MacDonald_1989,MacDonald_1990,Sondhi_2000, Sondhi_2001, Ma_2022, Chen_2023}, and p-modular fracton orders with planarity $k > 1$ \cite{Wickenden_2024}.  We hope our results will spur further progress on the characterization and classification of fracton quantum matter more generally. Toward this end, we anticipate that the notion of a perfect theory of excitations will serve as a blueprint for a cohesive mathematical theory of general abelian fracton orders. The development of such a framework, based on a novel structure dubbed a \emph{braided fusion complex}, will be the topic of future work \cite{Shirley_inprogress}.

\acknowledgments{We are grateful to Arpit Dua, Jeongwan Haah, Daniel Ranard and Dominic Williamson for useful discussions.  We are also grateful to Dachuan Lu and Bowen Yang for collaborations on related work.  The research of EW and MH is supported by the U.S. Department of Energy, Office of Science, Basic Energy Sciences (BES) under Award number DE-SC0014415.  The work of WS is supported by the Ultra-Quantum Matter Simons Collaboration (Simons Foundation grant no. 651444) and the Leinweber Institute for Theoretical Physics.  This material is based upon work supported by the National Science Foundation under Grant No. DMS 2143811 (AB).  The work of EW, WS and MH benefited from meetings of the UQM Simons Collaboration, supported by Simons Foundation grant no. 618615.}

%% file: localization.tex
%!TEX root = 1planar-v2.tex
\section{Prime factorization of theories of excitations}
\label{app:local}

In this section, we explain Theorem~\ref{thm:prime-decomp}. That is, we explain why a  $p$-theory   $(S,b)$ of fusion order $n$ is a stack of theories of fusion order $p^k$ where $p$ is a prime divisor of $n$ and $k$ is the largest integer such that $p^k|n$. This kind of decomposition is standard and could stand without further justification except perhaps a few key references to commutative algebra textbooks. However, we add details here to keep the paper self-contained.

\subsection{Background on prime localizations}

We start with a detour to discuss localization at a prime $p$.  
 Throughout, by a ring $R$, we continue to mean an associative, commutative and unital ring.

For an $R$-module $M$ and a prime number $p$, define $M_{(p)}$ as follows. Elements are fractions $m/u$ with $u\in \Z$ coprime to $p$, with equivalence  $m/u = m'/u'$ if and only if there exists $v$ coprime to $p$ such that $ v(u'm-um')=0$.
We can add elements of $M_{(p)}$ following the rules for fractions,
$m/u+m'/u' = (u'm+um')/uu' $.
For $M=R$, we can also multiply following the rules for fractions,
$(r/u)(r'/u') = (rr')/(uu')$.
Then $R_{(p)}$ is a ring again, and $ M_{(p)}$ is an $R_{(p)}$-module with
$ (r/u)(m/u') = (rm)/(uu')$. 
Define 
\[\eta_{M,p} \colon M \to M_{(p)}\]
by $\eta_{M,p}(m) = m/1$. This is called the \emph{localization map} and it is a map of $R$-modules.
Note that $\eta_{M,p}(m) =m/1\in M_{(p)}$ is zero if and only if there exists $u\in\Z$ coprime to $p$ such that $um=0$. 

If $f\colon M \to N$ is a homomorphism of $R$-modules, we define 
$f_{(p)} \colon M_{(p)} \to N_{(p)}$ by
$f_{(p)}( m/u) =f(m)/u$. Then $f_{(p)}$ is a homomorphism of $R_{(p)}$-modules and $f_{(p)}\circ \eta_{M,p}=\eta_{N,p} \circ f$, \emph{i.e.,} 
the diagram
\begin{equation}\label{eq:localizationnatural}
\xymatrix{M  \ar[r]^-{\eta_{M,p}}  \ar[d]_-f & M_{(p)} \ar[d]^-{f_{(p)}} \\ 
N  \ar[r]^-{\eta_{N,q}}  &  N_{(p)}  
}
\end{equation}
is commutative.

\begin{example}
For $R=\Z$, the localization
$\Z_{(p)}$  can be identified with the subset of rational numbers whose denominators are coprime to $p$.
\end{example}

We use this example to give an equivalent description of $M_{(p)}$, which is shown in \cite[Prop. 3.5]{Atiyah_1969}.
\begin{proposition}
 Define a map
$\Z_{(p)} \otimes M \to M_{(p)}$
by
$ \frac{v}{u}\otimes m \mapsto (vm)/u$.
This is an isomorphism of $R$-modules.  Under this identification, the localization map sends $m$ to $1\otimes m$.
% I.e., $\eta_{M,p}$ corresponds to $\eta_{\Z,q} \otimes \id_M$.
 If $f\colon M\to N$ is a map of $R$-modules, then $f_{(p)}$ corresponds to the homomorphism which maps $\frac{v}{u}\otimes m$ to $\frac{v}{u}\otimes f(m)$. %, i.e., $f_{(p)}$ is $\id\otimes f$.
\end{proposition}

A powerful property of localization is the fact that it preserves many of the operations one encounters working with modules. The following result combines Prop. 3.3, 3.4 \&  3.7 of Chapter 3 in \cite{Atiyah_1969}.

\begin{proposition}
Let $M,N$ and $L$ be $R$-modules.
\begin{enumerate}[(i)]
\item If $ L \xrightarrow{f} M  \xrightarrow{g} N $
is an exact sequence of $R$-modules, then 
$ L_{(p)} \xrightarrow{f_{(p)}}M_{(p)} \xrightarrow{g_{(p)}} N_{(p)} $
is an exact sequence of $R_{(p)}$-modules.
\item There are isomorphisms of $R_{(p)}$-modules $(M\oplus N)_{(p)} \cong M_{(p)}\oplus N_{(p)}$,  $(M\cap N)_{(p)} \cong M_{(p)} \cap N_{(p)}$, $(M/N)_{(p)} \cong M_{(p)}/ N_{(p)}$ and
$M_{(p)} \otimes_{R_{(p)}} N_{(p)} \cong (M\otimes_R N)_{(p)}$. 
\end{enumerate}
\end{proposition}

As a result of the above, we obtain a criterion for detecting when a module is trivial. 
\begin{proposition}
An $R$-module $M$ is zero if and only if $M_{(p)}$ is zero for all prime numbers $p$. 
\end{proposition}
\begin{proof}
Any $R$-module is also a $\Z$-module. Apply \cite[p.40, Prop. 3.8]{Atiyah_1969} with $A=\Z$. 
\end{proof}

\begin{example}
Let $R = \Z_n$. Then $R_{(p)}$ is $\Z_{p^k}$ where $k$ is the largest integer such that $p^k|n$. In particular, it is zero if $p$ is coprime to $n$. Writing $\Z_n \cong \Z_{p^k} \times \Z_{m}$ where $m$ is coprime to $p$, the localization map is the projection onto $\Z_{p^k}$ and has a canonical splitting $\nu_p \colon \Z_{p^k} \to \Z_n$ (as $\Z_n$-modules), namely, the map which sends $1 \in \Z_{p^k}$ to the unique element of $\Z_n$ which is 1 modulo $p^k$ and 0 modulo $m$.  Furthermore, if $p'$ is coprime to $p$, then $\eta_{p'}\nu_p=0$. Similarly, if $R=\Z_n[t^{\pm}]$, $R_{(p)} = \Z_{p^k}[t^{\pm}]$ and the localization map has a canonical splitting which we call $\nu_{R,p}$, and again, if $p'$ is coprime to $p$,
\[\eta_{R,p'}\nu_{R,p} =\begin{cases} \id_{R_{(p)}} & p'=p\\
0 & p'\neq p.
\end{cases}
\]
\end{example}

\begin{proposition}
Let $R$ be $n$-torsion and $M$ be an $R$-module. If $M_{(p)}=0$ for all primes $p$ such that $p|n$, then $M=0$. 
\end{proposition}
\begin{proof}
If $p'$ is coprime to $n$, then $M_{(p')}=0$ as well since $R_{(p')}=0$. So, $M$ is zero after localization at every integer  prime, hence, it is zero.
\end{proof}

\begin{proposition}\label{lem:sumloc}
Let $R$ be a ring. Suppose there is an  $R$-module homomorphism $\nu_{R,p} \colon R_{(p)} \to R$ such that  $\eta_{R,p}\nu_{R,p} =\id_{R_{(p)}}$ for every $p$ and $\eta_{R,p'}\nu_{R,p} =0$ when $p\neq p'$. Suppose further that $R_{(p)}=0$ for all but finitely many primes $p$. 
Let $M$ be an $R$-module. Then
\[ \eta_M \colon M \to   \bigoplus_{p \ \mathrm{prime}} M_{(p)} \]
given by $\eta_M(x) = \sum_p \eta_{M,p}(x)$   is an
  isomorphism of $R$-modules. Furthermore, if $f \colon M\to N$ is a homomorphism of $R$-modules, then the following diagram 
\[ \xymatrix{M  \ar[r]^-{\eta_M}_-\cong  \ar[d]_-f & \bigoplus_p M_{(p)} \ar[d]^-{\bigoplus_p f_{(p)}} \\ 
N  \ar[r]^-{\eta_N}_-\cong  & \bigoplus_p N_{(p)}  
}\]
commutes. That is, $\eta_N\circ f = \left(\bigoplus_p f_{(p)}\right) \circ \eta_M$.
\end{proposition}
\begin{proof}
When $p$ is coprime to $n$, $M_{(p)}=0$ so we can ignore these summands in the direct sum decomposition. So all sums are finite and the map $\eta_M$ is well-defined.
We have $M_{(p)} \cong R_{(p)} \otimes_{R} M$ and $\eta_{M,p} = \eta_{R,p}\otimes \id_M$. Define $\nu_{M,p}  \colon M_{(p)} \to M$ by
$\nu_{M,p} = \nu_{R,p}\otimes \id_M$. Then $\eta_{M,p} \nu_{M,p}  = \id_M$ and  $\eta_{M,p'} \nu_{M,p}  = 0$ if $p'\neq p$.
Let 
$ \sum_p x_p$ be an element of $\bigoplus_{p} M_{(p)} $ and $x = \sum_p  \nu_{M,p}(x_p) $ in $M$. Then 
\[ \eta_M(x)=  \sum_p \eta_M( \nu_{M,p}(x_p)) =  \sum_{p,p'}\eta_{M,p'}(\nu_{M,p}(x_p)) = \sum_p x_p.  \]
So $ \eta_M$ is surjective. Note that the kernel of $ \eta_M$ is the intersection of the kernels of $\eta_{M,p}$ for all primes $p$.
 We have an exact sequence
$0\to \ker(\eta_{M,p})\to M \to M_{(p)}$.
Localization is an idempotent operation, \emph{i.e.}, the natural map $\eta_{M_{(p)}, p}\colon M_{(p)} \to (M_{(p)})_{(p)}$ is an isomorphism. Since it preserves exact sequences,  
$0\to (\ker(\eta_{M,p}))_{(p)}\to M_{(p)} \to M_{(p)}$
is exact hence $\ker(\eta_{M,p})_{(p)}=0$. Therefore,
\[( \ker(\eta_M) )_{(p)}=\big( \bigcap_{p'} \ker(\eta_{M,p'})\big)_{(p)} =\bigcap_{p'} \ker(\eta_{M,p'})_{(p)} \subset  \ker(\eta_{M,p})_{(p)} =0 \]
for all primes $p$.
This implies that $\ker( \eta_M)=0$. So, $\eta_M$ is an isomorphism. 

The commutativity of the diagram follows directly from \eqref{eq:localizationnatural}.
\end{proof}

Localization also interacts well with duals of finitely presented modules. 
\begin{proposition}[{\cite[Lem. 3.3.8]{Weibel_1994}}]\label{lem:homdual}
Let $M$ be a finitely presented $R$-module. Then the $R_{(p)}$-module homomorphism
\[ \gamma \colon {\Hom_{R}}(M,R)_{(p)} \to  {\Hom_{R_{(p)}}}(M_{(p)},R_{(p)})\]
which takes $ \varphi/a$ to $\gamma( \varphi/a)( m/b) = \varphi(m)/(ab)$ is an isomorphism. 
\end{proposition}
\begin{proof}
This is Lemma 3.3.8 of \cite{Weibel_1994} applied to the multiplicative subset of $\Z\backslash(p)$ of $\Z$. 
\end{proof}

The isomorphism of Proposition~\ref{lem:homdual} allows us to unambiguously write
\[ M^*_{(p)} = (M^*)_{(p)} \cong (M_{(p)})^*.\]
When $R$ is Noetherian, any finitely generated module is finitely presented. 

\subsection{Localization of $p$-theories}
We now apply these concepts to $p$-theories. Recall that if $(S,b)$ is a $p$-theory of fusion order $n$, then $S$ is a finitely generated $R$-module where $R=\Z_n[t^{\pm}]$. Furthermore, $S^*$ is also finitely generated by Proposition~\ref{prop:noetherian-fg-dual}.

\begin{defn}
Let $(S,b)$ be a $p$-theory of fusion order $n$. Define $(S_{(p)}, b_{(p)})$ to be the $p$-theory of fusion order $p^k$, where $S_{(p)}$ is the localization of $S$ at $p$ and $b_{(p)}$ is the form associated to  the localization of the associated map $\Phi$, \emph{i.e.},
$\Phi_{(p)} \colon S_{(p)} \to  S_{(p)}^*$. That is,
\[B_{(p)}(x/a,y/b) = \Phi(x)(y)/(ab) \]
and $b_{(p)}(x/a,y/b)$ is the coefficient of the unit in $B_{(p)}(x/a,y/b)$. We call $(S_{(p)}, b_{(p)})$ the \textbf{associated $p$-local p-theory}.
\end{defn}

The following result is a reformulation of Theorem~\ref{thm:prime-decomp}.
\begin{theorem}\label{thm:sumoflocal}
Let $(S,b)$ be a $p$-theory. Then
\[(S,b) \cong  \mathop{\bigominus}_{p \  \mathrm{prime}} (S_{(p)}, b_{(p)}).\]
\end{theorem}
\begin{proof}
Let $(S,b) = (S',b')\ominus (S'',b'')$. Then $S^* =  (S')^* \oplus (S'')^* $ and the associated map $\Phi$ is  the direct sum of the associated maps, \emph{i.e.},  $\Phi = \Phi'\oplus \Phi''$. By Proposition~\ref{lem:sumloc}, the diagram 
\begin{equation} \label{eq:sumptheories}
\xymatrix{S \ar[r]^-{\eta_S}_-\cong \ar[d]_-{\Phi}  & \displaystyle\mathop{\bigoplus}_{\text{$p$ prime}} S_{(p)}  \ar[d]^-{\bigoplus \Phi_{(p)}}\\
 S^* \ar[r]^-{\eta_{S^*}}_-\cong &  \displaystyle\mathop{\bigoplus}_{\text{$p$ prime}} S_{(p)}^*
}
\end{equation}
is commutative, \emph{i.e.}, $(\bigoplus \Phi_{(p)}) \circ \eta_S  = \eta_{S^*} \circ \Phi$. The summands on the right-hand side of \eqref{eq:sumptheories} are non-zero only if $p|n$. The right-hand side is the stack of the associated $p$-local p-theories for $p|n$.
\end{proof}

%% file: ktheory.tex
%!TEX root = 1planar-v2.tex
\section{Additive Invariant}
\label{app:ktheory}

In this section, we give the proof of Theorem~\ref{thm:totalrankinvariance}, which states that the total rank of an $A$-filtration of a finitely generated module $M$ only depends $M$, and not on the choice of filtration, and so is an invariant of $M$. This invariant is the zeroth $K$-group of the additive category of finitely generated $R$-modules, and is often denoted by $G_0(R)$. Theorem~\ref{thm:totalrankinvariance} is a special case of the Fundamental Theorem for $G_0$-theory of Rings \cite[II.6.5]{Weibel_2013}. To be self-contained, we explain this theorem, following \cite[Ch. II, \S6]{Weibel_2013}.

We assume that our associative, commutative, unital ring $R$ is also Noetherian.
Let $G_0(R)$ be the abelian group with presentation given by a generator $[M]$ for each finitely generated $R$-module $M$, and a relation  $[M] = [L]+[N]$ for each short exact sequence 
\begin{equation}\label{eq:exactsectmodR}
0 \to L \to M \to N \to 0
\end{equation}
of finitely generated $R$-modules.\footnote{To avoid size issues, we can choose $M$ from a small, equivalent, subcategory of the category of finitely generated $R$-modules.}
In the group $G_0(R)$, 
\begin{enumerate}
\item $[0] =0$,
\item if $M\cong N$, then $[M] = [N]$, and
\item $[M\oplus N] = [M] + [N]$.
\end{enumerate}
The group $G_0(R)$ has the following universal property. Let $Z$ be an abelian group. Any function which assigns an element $d(M)$ of $Z$ to each finitely generated $R$-module $M$  so that $d(M) = d(L)+d(N)$ when there is an exact sequence \eqref{eq:exactsectmodR} extends to a group homomorphism
$d \colon G_0(R)  \to Z$.

\begin{example}
If $F$ is a field, then finitely generated modules over $F$ are just finite-dimensional vector spaces. The dimension extends to an isomorphism
$\dim\colon G_0(F) \to \Z$ and $G_0(F)\cong \Z$ generated by $[F]$.
\end{example}

Under certain conditions, given a ring homomorphism $f\colon R \to S$ (between Noetherian rings), we can define certain group homomorphisms between $G_0(R)$ and $G_0(S)$.
Note that $f$ allows us to view $S$ as an $R$-module by letting $r\cdot s = f(r)s$. In fact, any $S$-module $N$ can be viewed as an $R$-module via $r\cdot n = f(r)n$. We denote this $R$-module by $f_*N$.  We can also view any $R$-module $M$ as an $S$-module by defining $f^*M = S\otimes_R M$. 
\begin{enumerate}
\item If $f_*S$ is a finitely generated $R$-module and $N$ is a finitely generated $S$-module, then $f_*N$ is a finitely generated $R$-module and
$f_* \colon G_0(S) \to G_0(R)$
 defined by $f_*([N])=[f_*N]$ is a group homomorphism.
\item Suppose that $f_*S$ is a flat $R$-module, then
$f^*\colon G_0(R) \to  G_0(S) $
defined by  $f^*([M])=[f^* M]$ is a homomorphism.
\item (Serre's Formula) More generally, noting that $f^* M = \Tor_0^R(S,M)$, suppose that $f_*S$ has a finite resolution by flat $R$-modules, then the formula 
\[f^*([M]) =\sum_{i\geq 0} (-1)^i [\Tor^R_i(S,M)]\]
is a finite sum and defines a homomorphism $f^*\colon G_0(R) \to  G_0(S) $.
\end{enumerate}

\begin{example}
If $R$ is an integral domain, then it embeds in its field of fractions $f\colon R \to F$ and $F$ is flat over $R$. We thus get a homomorphism 
$f^* \colon G_0(R) \to G_0(F)$.
Following it with the dimension, we get a homomorphism
$\dim \colon G_0(R)  \to \Z $
which sends $[M]$ to $\dim(F\otimes_R M)$.
\end{example}

\begin{example}
Suppose that $A$ is a P.I.D., then any finitely generated $A$-module $M$ is isomorphic to one of the form
$A^r \oplus A/(a_1)  \oplus \cdots \oplus A/(a_1) $ where, by definition, $r=\mathrm{rank}(M)$.
Since $A/(a)$ sits in an exact sequence
\[\xymatrix{ 0 \ar[r] &  A  \ar[r]^-{a} & A \ar[r] & A/(a) \ar[r] & 0} \]
we see that $[A/(a)]=0$ and that
$[M] =   [A^r] + [A/(a_1)]  + \cdots + [A/(a_1)] = r[A]$. The rank thus extends to 
an  isomorphism $\mathrm{rank} \colon G_0(A) \to \Z$, and $G_0(A)\cong \Z$ generated by $[A]$.
 \end{example}
 
 \begin{example}
 Our next example is $G_0(\Z_{p^k})$. This is treated in Examples 6.2.2 and 6.2.4 of \cite{Weibel_2013}. A finitely generated $\Z_{p^k}$-module is a finite abelian group annihilated by $p^k$. Any finite abelian group $M$ has a composition series 
 \[  0=M^0 \subset M^1 \subset M^2 \subset \cdots \subset M^\ell=M\]
 where the quotients $M^s/M^{s-1}$ are simple groups.  
 By the Jordan--H\"older Theorem, the length $\ell$ of the composition series is a well-defined function of $M$ which turns exact sequences into sums, and so extends to a homomorphism $\mathrm{length} \colon G_0(\Z_{p^k}) \to \Z$. The only simple $\Z_{p^k}$-module is $\Z_p$ with $\mathrm{length}(\Z_p)=1$. It follows that  $G_0(\Z_{p^k}) \cong \Z$ generated by $[\Z_p]$.
 \end{example}

Our goal, which was Theorem~\ref{thm:totalrankinvariance}, is an immediate consequence of the following result.
\begin{theorem}[Fundamental Theorem of $G_0$-Theory]
Let $R=\Z_{p^k}[t^{\pm }]$ and $A=R/(p)$. Let $f\colon R \to A$ be the quotient map. Then
\[f_* \colon G_0(A) \to G_0(R)\]
is an isomorphism, and so $G_0(R) \cong \Z$ generated by $[A]$. For any $A$-filtration of a finitely generated $R$-module $M$ by $C^s$ ($s=0,\ldots, \ell$), if  $r_s$ the rank of $C_s$ as an $A$-module,
we have
\[[M] =\left( \sum_{s=1}^\ell r_s \right)[f_*A] \in G_0(R). \]
\end{theorem}
\begin{proof}
We first show surjectivity of $f_*$.
Fix any $A$-filtration of $M$. 
(Any finitely generated $R$-module $M$ has an $A$-filtration as we can take the $p$-torsion filtration.) Since $R$ is Noetherian, each $C^s$ is finitely generated and so are the quotients $C_s=C^s/C^{s-1}$. The exact sequences $0 \to C^{s-1} \to C^s \to C_s \to 0$ give the identity $ [C^s] = [C_s] +[C^{s-1}]$. 
Inductively, we deduce that
$ [M] = \sum_{s=1}^{\ell}  [C_s] $. 
Since $C_s$  is an $A$-module, it is in the image of $f_*$. Since  $A$ is a P.I.D., $[C_s] = r_{s} [A]$ in $G_0(A)$.
Therefore, $f_*$ is surjective and 
\[[M] = \left( \sum_{s=1}^\ell r_s \right)[ f_*A].\]

To show injectivity, we consider the ring homomorphism $\rho\colon R  \to R/(1-t)\cong \Z_{p^k}$. The exact sequence 
\[\xymatrix{ 0 \ar[r] &  R  \ar[r]^-{1-t} & R \ar[r] & \Z_{p^k} \ar[r] & 0} \]
 is a resolution of $ \Z_{p^k} $ by flat $R$-modules.
So, we have a homomorphisms
$\rho^* \colon G_0(R) \to G_0(\Z_{p^k})$
given by 
$\rho^*([M]) = \sum_i (-1)^i [\Tor^R_i(\Z_{p^k}, M)]$.
For $M=f_*A$, the $\Tor$ groups are computed as the homology of 
\[\xymatrix{ 0 \ar[r] &  A  \ar[r]^-{1-t} & A \ar[r]& 0} .\]
We thus have $\Tor_0^{R}(\Z_{p^k}, f_*A) = \Z_{p}$ and $\Tor_1^R(\Z_{p^k}, f_*A)=0$.   Therefore,
$\rho^*\circ f_*([A]) = \rho^*([f_*A]) = [\Z_p]$ in $G_0(\Z_{p^k})$. We conclude that
$ \rho^*\circ f_* \colon G_0(A) \to G_0(\Z_{p^k})$
is an isomorphism, and so $f_*$ is injective.
\end{proof}